\documentclass[10pt]{article}
\usepackage{amsfonts,amsmath,amssymb,amsthm}
\usepackage[pdftex]{graphicx}
\usepackage[
    left=1in,
    right=1in,
    top=1in,
    bottom=1in,
    marginparwidth=0.75in,
    marginparsep=0.25in
]{geometry}

\usepackage{natbib} 
\usepackage{setspace}
\usepackage{titletoc}
\usepackage{makecell}
\usepackage{slantsc}
\usepackage{mathabx}
\usepackage{mdframed}
\usepackage{ulem}
\usepackage{yfonts}
\usepackage{enumerate}
\usepackage[toc,title,titletoc,header]{appendix}
\usepackage[colorlinks,citecolor=blue,hypertexnames=false]{hyperref}
\usepackage{etoolbox}
\usepackage{bbm}
\usepackage{array,longtable}
\usepackage[font=small]{caption}
\usepackage{tikz,pgfplots}
\usetikzlibrary{arrows.meta}
\newcolumntype{C}{>{\begin{math}}c<{\end{math}}} % automatic math mode, centered
\allowdisplaybreaks

\usepackage{float}
\usepackage{multirow}
\usepackage{booktabs}
\usepackage{adjustbox}
\usepackage{nicematrix}
\usepackage{comment}
\usepackage{subcaption}

\usepackage{float}
\usepackage{tikz}
\usepackage{pgfplots}
\pgfplotsset{compat=1.18}

\def\qed{\rule{2mm}{2mm}}
\def\indep{\perp \!\!\! \perp}

\newtheorem{theorem}{Theorem}[section]
\newtheorem{lem}{Lemma}[section]

\newtheorem{prop}{Proposition}[section]
\theoremstyle{definition}
\newtheorem{exmp}{Example}[section]
\newtheorem{rem}{Remark}[section]
\newtheorem{assu}{Assumption}[section]

\newtheorem{algorithm}{Algorithm}[section]

\surroundwithmdframed[
    linewidth=0.6pt,
    linecolor=black!55,
    backgroundcolor=black!0,
    roundcorner=2pt,
    innertopmargin=0pt,
    innerbottommargin=8pt,
    innerleftmargin=10pt,
    innerrightmargin=10pt
]{algorithm}

\usepackage[bb=dsserif]{mathalpha}
\usepackage{bm}

\AtEndEnvironment{rem}{~\qed}
\AtEndEnvironment{exmp}{~\qed}

\DeclareMathOperator{\var}{Var}

\DeclareMathOperator*{\argmax}{argmax}

\newcommand{\bbG}{\mathbb{G}}
\newcommand{\bbE}{\mathbb{E}}
\newcommand{\bbR}{\mathbb{R}}
\newcommand{\bbX}{\mathbb{X}}
\newcommand{\calG}{\mathcal{G}}
\newcommand{\calF}{\mathcal{F}}
\newcommand{\calH}{\mathcal{H}}
\newcommand{\calN}{\mathcal{N}}

\newcommand{\calS}{\mathcal{S}}

\newcommand{\Leb}{\text{Leb}}
\newcommand{\He}{\text{He}}
\newcommand{\FWER}{\text{FWER}}

\pgfplotsset{compat=1.18}

\begin{document}

\title{Uncertainty in Compound Decisions\thanks{I thank Isaiah Andrews and Anna Mikusheva for their generous guidance and support. This paper also benefited from conversations with Alberto Abadie, Nikhil Agarwal, Victor Chernozhukov, Nathaniel Hendren, Chenxi Jiang, Andrew Komo, Haruki Kono, Soonwoo Kwon, Ricky Li, Jos\'e Luis Montiel Olea, Sarah Moon, Ashesh Rambachan, Azeem Shaikh, Yucheng Sheng, Steven Shi, Ot\'avio Tecchio, Santiago Torres, and Vod Vilfort. I additionally thank participants in the MIT econometrics lunch, third-year lunch, and second-year paper workshop for helpful comments.}}

\author{Andreas Petrou-Zeniou
\thanks{Department of Economics, MIT.
\url{apetrouz@mit.edu}}
}

\maketitle

\begin{spacing}{1.0}
\begin{abstract}
{\fontsize{10}{12}\selectfont
Researchers and policymakers frequently make a large number of decisions simultaneously over many units using noisy estimates of unit-level effects. Researchers often assume that unit-level effects are drawn from a common effect distribution, allowing for optimal decision making. In applications, however, this effect distribution is unknown to the econometrician, introducing uncertainty in the ultimate decision making problem. This paper considers the problem of quantifying this uncertainty. For instance, researchers may seek to understand how an oracle decision-maker with knowledge of the effect distribution would compare different feasible decisions, or understand which decisions an oracle would plausibly make. Analyzing these questions from the perspective of an oracle decision-maker involves comparing the conditional means of unit-level effects given estimates thereof, suggesting a sufficient inferential objective of uniform coverage of these conditional means. Based on an equivalence between uniform coverage of conditional means and multiple testing, I develop methods that satisfy this coverage. Finally, I consider an application to the problem of identifying high-economic-mobility neighborhoods. My methods sharpen inference relative to existing methods and identify units that can be confidently included in oracle selections.}
\end{abstract}
\end{spacing}

\onehalfspacing

\section{Introduction}

Many decision problems in economics are compound decisions (\citealt{Robbins1956}), in that they involve comparing or classifying a large number of units or populations, such as schools or neighborhoods, according to some noisily-observed feature of each unit, such as value-added. In studies of teacher or school value-added, for example, researchers seek to differentiate high-performing schools from low-performing schools (\citealt{MorgansteinWasserstein2014}, \citealt{Angristetal2025}). Relatedly, in healthcare, policymakers may seek to encourage patients to select high value-add nursing homes (\citealt{EFM2025}). In studies of employer discrimination, policymakers may seek to select a set of firms to audit for civil rights violations (\citealt{KlineWalters2021}, \citealt{Klineetal2022}, \citealt{Klineetal2024}). In problems of place-based policy-targeting, researchers may seek to distinguish between high and low economic mobility neighborhoods (\citealt{Chettyetal2014}, \citealt{Chettyetal2018aer}, \citealt{Chettyetal2026}, \citealt{ChettyandHendren2018}). To make such decisions, researchers often assume that these features are drawn independently from a common distribution, providing optimal decision rules. In many applications, this distribution of unit-level features or effects is unknown to the econometrician. As a result, the optimal decision described above is infeasible, and can be thought of as the decision made by an oracle with knowledge of the unknown distribution of effects. This article considers a range of inferential exercises that arise given uncertainty around an oracle decision-maker's actions and beliefs. 

To fix ideas, consider a recent experiment due to \citealt{Bergmanetal2023}, who seek to identify high economic-mobility neighborhoods in Seattle and thus nudge low-income housing voucher recipients to move to these areas. While true economic mobility is not observed, policymakers have access to noisy estimates thereof (see e.g. \citealt{Chettyetal2018aer}, \citealt{ChettyandHendren2018}). In using these estimates to differentiate between low and high-mobility tracts, the decision-maker faces a key dilemma between selecting tracts with high and noisy estimates of mobility and tracts with lower mobility and more precise estimates. To navigate this dilemma, the oracle decision maker obtains conditional mean economic mobility at each neighborhood given these estimates and the distribution of mobility across neighborhoods. This distribution is unknown to the econometrician, and different choices of this distribution can lead to dramatically different selections of high-mobility neighborhoods. Given this uncertainty, this paper considers questions such as: can policymakers be confident that neighborhoods with noisily estimated, high mobility should be included in an oracle decision-makers selection of tracts? Given a realized selection, can we be confident that decision-making could or could not be improved by including more, or fewer, tracts with low and precisely estimated mobility? How do different ways of navigating the planner's dilemma compare according to an oracle decision-maker?

To answer the above questions, I first observe that for many decision problems considered in applied work in economics, the oracle decision-maker's actions depend on the distribution of effects only through the conditional means of unit-level effects given the realized, noisy estimates thereof. Consequently, this paper advances a notion of coverage of conditional means which holds simultaneously over potential realizations of the data. To understand why this notion of coverage suffices to address the above exercises, notice that these exercises frequently involve comparing conditional mean effects across many units. For example, analyzing the dilemma described above involves comparing many high-noise tracts with high estimated mobility to many lower-noise tracts with lower estimated mobility. Consequently, coverage of conditional means of unit-level effects given individual, potential realizations of the data, which I refer to as marginal inference, does not suffice for our objectives; coverage of individual conditional means does not guarantee simultaneous validity of these comparisons. \citealt{IgnatiadisandWager2021} develop methods satisfying this desired marginal coverage, which I extend to allow for heteroskedasticity and which I demonstrate via simulation can dramatically undercover in the exercises above. 

This paper's main contribution is to develop methods that uniformly cover conditional means over potential realizations of the data in order to guard against such false discoveries. Although such uniform confidence bands can be constructed from a confidence set for the effect distribution, they should focus power on its decision-relevant features. In service of this objective, I make two observations. First, I show that the problem of constructing such uniform confidence bands is equivalent to a multiple testing problem in which the econometrician seeks to control the familywise error rate for a particular class of null hypotheses. To demonstrate the utility of this first observation, I next observe that the null hypotheses defining this multiple testing problem can be written as zero restrictions on a class of linear functionals of the effect distribution. 

Combining these observations has three important consequences. First, the observations identify the decision-relevant features of the effect distribution; namely the linear functionals of the effect distribution corresponding to the null hypotheses considered. Second, the linearity of these null hypotheses on the effect distribution allows me to draw on a large literature on minimax estimation and inference for linear functionals in nonparametric settings (see e.g. \citealt{Donoho1994}, \citealt{IgnatiadisandWager2021}, \citealt{ArmstrongKolesar2018}) that adapt to the estimation difficulty of each linear functional of interest. Finally, due to my multiple testing equivalence, I build on canonical multiple-testing methods (see e.g. \citealt{RomanoWolf2005}) to obtain a step-down approach that accounts for the bias of the minimax test statistics used in the inferential procedure. I provide a theoretical analysis of this step-down procedure, and demonstrate in theorem \ref{thm:bootstrap_validity} the uniform asymptotic consistency in level of the procedure using multiplier bootstrap critical values. The key technical difficulty in obtaining this result is that when allowing for studentization, the class of minimax test statistics in this paper can grow increasingly complex in large sample sizes (see e.g. \citealt{CCK_2014}, \citealt{CCK_2016}, \citealt{CCK_ac_2014}). 

My methods build on recent work due to \citealt{IgnatiadisandWager2021}, who study inference for the conditional means considered in this paper.\footnote{Work due to \citealt{Guetal2025} also considers the problem of inference on the conditional means considered in this paper. They review the methods due to \citealt{IgnatiadisandWager2021}, in addition to providing a method based on \citealt{FSST2021}.} \citealt{IgnatiadisandWager2021} provide two approaches to inference, one which satisfies uniform coverage as considered in this article, and the other which satisfies only marginal coverage. The former approach uses a particular confidence set for the effect distribution to construct uniform confidence bands, which can fail to focus power on decision-relevant features of the effect distribution. The latter, marginal, approach is based on a similar test inversion to that considered in this paper. However, the inversion in \citealt{IgnatiadisandWager2021} is based on tests that do not control the familywise error rate, and therefore does not provide the uniform guarantees desired in this paper. Nevertheless, \citealt{IgnatiadisandWager2021} derive test statistics for their null hypotheses, and I extend their arguments to allow for the precision of the estimated unit-level effects to vary across units in deriving my test statistics. Finally, the multiple testing equivalence established in this paper builds on a similar equivalence between set inference and multiple testing derived in \citealt{RomanoShaikh2010}. In fact, I establish a further equivalence between uniform confidence band construction to set inference as considered in \citealt{ChernozhukovHongTamer2007} and \citealt{RomanoShaikh2010}, which I discuss in remark \ref{rem:set-inference-connection}. 

To establish the practical relevance of my methods, I consider an application to the policy-targeting task of \citealt{Bergmanetal2023}. I consider two main exercises: which neighborhoods can policymakers be confident that an oracle decision-maker would select? How would an oracle decision-maker compare different realized selections? In my application, my methods deliver substantially tighter confidence sets than existing methods due to \citealt{IgnatiadisandWager2021}, shifting economic conclusions and suggesting the step-down procedure developed in this paper is a worthwhile methodological investment for applied work.

In my first application, I consider the problem of differentiating between low and high-mobility census tracts as in \citealt{Bergmanetal2023}. For each tract, I take as data local economic mobility estimates due to \citealt{Chettyetal2018aer}, which measure the conditional mean earnings percentile in adulthood of low income children growing up in each given census tract. Given these mobility estimates, my findings suggest that many Seattle neighborhoods can be confidently included and excluded (9\% and 44\%, respectively) from an oracle selection of high-mobility tracts. In a second exercise, I show that the gaps in conditional mean mobility between feasibly-classified low and high-mobility tracts are on the order of five percentile points in adulthood earnings, at a lower confidence bound. Interpreted causally, my analyses suggest that researchers can be confident that an oracle decision-maker would expect economically large gains in the adulthood earnings of children moving to most of the high-mobility tracts selected by \citealt{Bergmanetal2023}. 

My analyses are complementary to recent work studying the informativeness of noisy mobility estimates for the decision-making objectives of \citealt{Bergmanetal2023}. Existing exercises (\citealt{andrewsetal2023}, \citealt{Mogstadetal2023}, \citealt{petrouzeniou_and_shaikh_2025}) effectively compare to studying a fully-informed oracle with knowledge of true economic mobility at each tract, rather than the distribution of mobility across tracts. Building on the analysis of \citealt{AndrewsandChen2025}, the exercises considered in this paper can be viewed as providing evaluations and certifications for feasible decisions according to an oracle decision-maker with knowledge of the effect distribution, whereas existing work provides evaluations and certifications for the fully-informed oracle. The researcher's choice of inferential exercise therefore depends on the notion of an oracle decision-maker of interest to the policymaker. Given recent work providing methods to approximate the distributional oracle decision-maker in the \citealt{Bergmanetal2023} policy-targeting exercise (see e.g. \citealt{chen2024}, \citealt{barnhard2026}), the exercises in this paper complement existing inferential exercises studying the fully-informed oracle. In particular, existing work (\citealt{Mogstadetal2023}, \citealt{petrouzeniou_and_shaikh_2025}) highlights that economic mobility estimates are too noisy to draw informative conclusions on the actions of a fully-informed oracle decision-maker. In contrast, my results suggest that in policy-targeting tasks using mobility estimates, researchers can instead draw far more informative conclusions when considering an oracle decision-maker with knowledge of the distribution of mobility across geographies.

The remainder of the paper proceeds as follows. Section \ref{sec:setup} introduces the compound decision problem via the selection problem of \citealt{Bergmanetal2023}, my inferential objectives, and a set of further examples. Section \ref{sec:equivalence_results} discusses a key simplification of my inferential objectives by drawing a connection to set inference and multiple testing. Section \ref{sec:methods} leverages this connection to develop a step-down approach to uniform confidence band construction. Section \ref{sec:economic-mobility-application} applies my methods to the setting of \citealt{Bergmanetal2023}. Section \ref{sec:conclude} concludes. The appendix contains proofs, details on applications, and an additional application to the setting of \citealt{Klineetal2022}.

\section{Setup, Notation, and Motivating Examples}
\label{sec:setup}

This section introduces the compound decision problems considered in this paper. I will use as a running example the selection problem of \citealt{Bergmanetal2023} described in the previous section, though I conclude with a range of further examples to highlight the range of compound decision problems considered in this paper. In general, I consider units $i=1,\ldots,n$, each associated with a unit-level feature or effect $\theta_i$. These unit-level effects are unobserved to the econometrician, but may be observed with noise. I denote by $\hat\theta_i$ this estimate, and I denote by $\sigma_i$ the precision with which $\hat\theta_i$ is estimated. Both $\hat\theta_i$ and $\sigma_i$ are known. Formally, $\hat\theta_i|\theta_i,\sigma_i$ is drawn from a Gaussian centered at $\theta_i$ with variance $\sigma_i^2$. In the context of the targeting exercise of \citealt{Bergmanetal2023}, these $\theta_i$ denote economic mobility at each tract -- the conditional earnings percentile of low-income children growing up in tract $i$ -- and the corresponding $\hat\theta_i$ denote estimates thereof.

To relate unit-level effects across units, I assume that $\theta_i$ are drawn iid across units from a common distribution of unit-level effects $G$, which can be thought of as the distribution of economic mobility across census tracts. Consequently, the data are sampled according to the following model:
\begin{equation}
\label{eq:EB_DGP}
    \hat\theta_i|\theta_i,\sigma_i
    \overset{ind}{\sim}\mathcal{N}(\theta_i,\sigma_i^2),~
    (\theta_i,\sigma_i)\overset{iid}{\sim}G
\end{equation}
In what follows, I will assume that $G$ lies in a convex class $\calG$. I will also assume that the convex class $\calG$ is dominated by the Lebesgue measure,\footnote{More generally, the methods developed in this paper allow $\calG$ to be dominated by some known measure $\nu$.} such that the density of each $G$ with respect to the Lebesgue measure is given by $g_G$. The normality of effect estimates $\hat\theta_i$ can be motivated by the central limit theorem, since the $\hat\theta_i$ are often distributed as sample means.\footnote{When the data are exchangeable, the existence of a common distribution generating latent unit-level effects can be justified by de Finetti's theorem. The data can be treated as exchangeable when restricting attention to permutation equivariant decision rules. See \citealt{KoenkerGuForthcoming} for further discussion.} I additionally allow for dependence between unit-level effects and the precision with which these effects are estimated.\footnote{For instance, low-income families may be less likely to reside in high-mobility areas, meaning that on average, estimates of mobility in high-mobility areas are noisier than in low-mobility areas. See \citealt{chen2024} for further discussion. The dependence between $\theta$ and $\sigma$ is restricted by $\calG$, and conditions must be imposed for bootstrap validity to hold.}

Policymakers often seek to make unit-level decisions for each $i$. For example, a policymaker may seek to label a subset of the $n$ tracts as high-mobility. To this end, define a decision $\delta_i$ lying in $\{0,1\}$, where $\delta_i = 1$ denotes selecting tract $i$. The policymaker can only select one-third of tracts, so I assume that the vector of classifications $\delta_i$ satisfies $\sum_i \delta_i = \lfloor n/3\rfloor$. As in \citealt{Bergmanetal2023}, the policymaker may seek to maximize a welfare function consisting of total economic mobility among the selected tracts:
\begin{equation}
\label{eq:util_max_welfare_defn}
    W(\boldsymbol{\delta},\boldsymbol{\theta}) := \sum_i \delta_i\theta_i~.
\end{equation}
This welfare function depends on tract-level mobility, which are unknown to the econometrician. Consequently, the decision-maker computes the average welfare by taking expectations over tract-level mobility. The decision-maker's goal is therefore to choose a map from data $(\boldsymbol{\hat\theta},\boldsymbol{\sigma})$ to tract-level selections in $\{0,1\}^n$ that maximizes the total mobility among selected tracts, on average:
\begin{equation}
\label{eq:bayesian_decision_problem}
    \boldsymbol{\delta}^*(\cdot) = \argmax_{\boldsymbol{\delta}} \bbE_G W(\boldsymbol{\delta},\boldsymbol{\theta})~.
\end{equation}
To solve \eqref{eq:bayesian_decision_problem}, it suffices to maximize conditional expected welfare given the realized data.\footnote{This is a consequence of the law of iterated expectations.} Namely, for the welfare defined in \eqref{eq:util_max_welfare_defn}:
\begin{equation}
\label{eq:optimal_decision_selection}
    \boldsymbol{\delta}^*(\boldsymbol{\hat\theta},\boldsymbol{\sigma}) = \argmax_{\substack{\boldsymbol{\delta}\in\{0,1\}^n\\ \sum_i\delta_i = \lfloor n/3\rfloor}}\bbE_G \left(
W(\boldsymbol{\delta},\boldsymbol{\theta}) |\boldsymbol{\hat\theta},\boldsymbol{\sigma}
    \right) =\argmax_{\substack{\boldsymbol{\delta}\in\{0,1\}^n\\ \sum_i\delta_i = \lfloor n/3\rfloor}}\sum_i\delta_i\bbE_G \left( \theta_i|\hat\theta_i,\sigma_i\right) ~.
\end{equation}
Towards interpretation, the above provides that to analyze the decision problem \eqref{eq:bayesian_decision_problem}, it suffices to consider the total conditional mean economic mobility across selected tracts. Moreover, because the total conditional mean economic mobility across selected tracts is additively separable in the conditional mean mobility at each tract, it suffices to consider the conditional mean economic mobility at each individual tract. Therefore, the problem \eqref{eq:optimal_decision_selection} admits a simple solution: it is optimal to select the top-third of tracts according to their conditional mean economic mobility.

The optimal selection of census tracts described in \eqref{eq:optimal_decision_selection} depends on the distribution of latent effects $G$, which is unknown to the econometrician. In this sense, the decision $\boldsymbol{\delta}^*$ described in \eqref{eq:optimal_decision_selection} is that of an oracle decision-maker with knowledge of the effect distribution $G$. Given uncertainty surrounding $G$, how should researchers evaluate uncertainty surrounding the downstream, oracle decision-maker's objectives?

\subsection*{Consequences of Effect Distribution Uncertainty}

To illustrate the ambiguity relating to decision-making that arises due to uncertainty in $G$, I notice that the optimal selection of tracts can vary substantially with the true latent distribution of effects. To demonstrate this, recall the planner's dilemma: should a policymaker select noisy tracts with high estimated mobility, or less noisy tracts with lower estimated mobility? More precisely, consider two tracts $h$ and $\ell$, such that estimated mobility at tract $h$ exceeds that at tract $\ell$, namely $\hat\theta_h>\hat\theta_\ell$, but estimation precision at tract $h$ is lower than at $\ell$, namely $\sigma_h > \sigma_\ell$. Whether the oracle decision-maker selects tract $h$ over tract $\ell$ depends on whether the conditional mean mobility $\bbE_G(\theta_h|\hat\theta_h,\sigma_h)$ at tract $h$ exceeds that at tract $\ell$ given by $\bbE_G(\theta_\ell|\hat\theta_\ell,\sigma_\ell)$. The oracle decision-maker may naturally penalize the estimate $\hat\theta_h$ more than $\hat\theta_\ell$, since it is more noisily estimated and thus less informative of true mobility $\theta_h$. In this regard, the oracle decision-maker's selection depends on the magnitude of the gap $\hat\theta_h-\hat\theta_\ell$. Moreover, mobility may be positively correlated with estimation precision in a phenomenon \citealt{chen2024} refers to as precision dependence: fewer low-income children grow up in high-mobility tracts, and therefore economic mobility is more noisily-estimated in high-mobiltiy neighborhoods. This introduces the opposite effect to the above: if noisily estimated tracts tend to be higher-mobility, the oracle decision-maker may prefer to select tract $h$ over tract $\ell$. 

While the oracle decision can vary widely with the distribution of effects generating the data, the joint distribution of the observed data $(\hat\theta_i, \sigma_i)$ contains information on the joint distribution $G$ of unit-level effects and estimation precision $(\theta_i,\sigma_i)$. For example, suppose that $\sigma_i = \sigma$ for all units $i$. In this case, the variance of estimated mobility $\hat\theta_i$ above the variance of the estimation error $\sigma_i$ provides information on the variance, or dispersion, of the distribution of mobility. Similarly, allowing $\sigma_i$ to vary across units, the correlation between estimated mobility $\hat\theta_i$ and estimation precision $\sigma_i$ is informative of the correlation between true mobility and estimation precision. Consequently, some optimal selections of tracts are more plausible than others given the observed data.

To provide examples of two plausible effect distributions, consider figure \ref{fig:two_consistent_priors}. Given an effect distribution $G$, let $F_G$ denote the joint distribution of the observed data $(\hat\theta_i,\sigma_i)$ assuming the data are generated by $G$. While the true joint distribution of the observed data is unknown, its empirical counterpart is known, from which a confidence set for the true $F_G$ can be constructed. I use the economic mobility estimates due to \citealt{Chettyetal2018aer} to obtain this empirical joint distribution and corresponding confidence set for $F_G$. In figure \ref{fig:two_consistent_priors}, I plot the densities of two effect distributions, $G_1$ and $G_2$, for which $F_{G_i}$ lies in this confidence set. In this sense, the mobility distributions $G_1$ and $G_2$ are both consistent with the observed data. Visual inspection of figure \ref{fig:two_consistent_priors} suggests that different patterns of precision dependence (i.e. correlation between mobility and estimation precision) can be consistent with the data. $G_1$, plotted in the left panel, displays positive dependence, consistent with the intuition described above. On the other hand, $G_2$ displays negative dependence. 

\begin{figure}[H]

    \centering

    \begin{minipage}[t]{0.48\textwidth}
        \centering
        \includegraphics[
            width=\linewidth,
            height=4.7cm,
            keepaspectratio
        ]{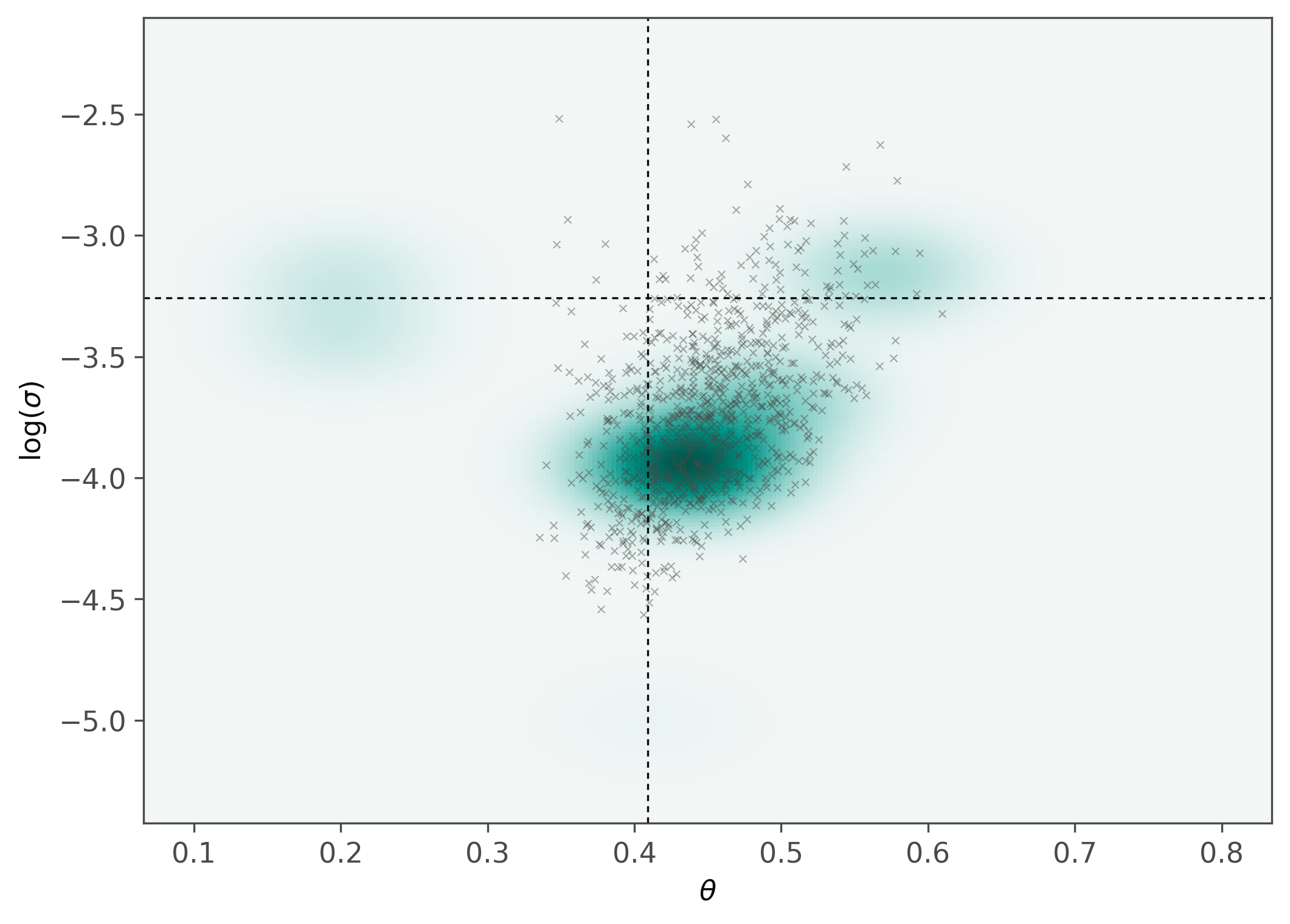}
    \end{minipage}
    \hfill
    \begin{minipage}[t]{0.48\textwidth}
        \centering
        \includegraphics[
            width=\linewidth,
            height=4.7cm,
            keepaspectratio
        ]{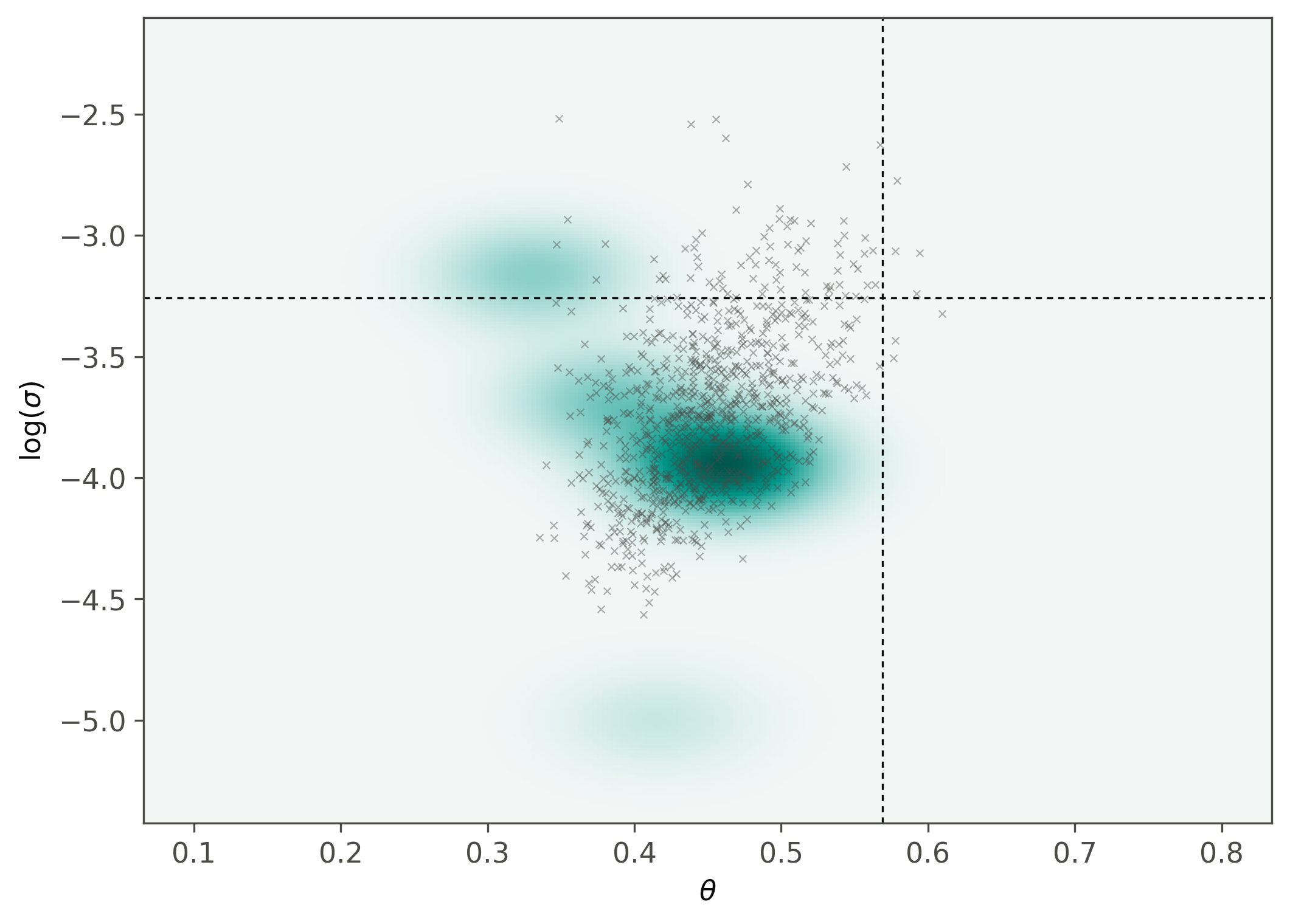}
    \end{minipage}
\caption{Two distributions $G_1$ and $G_2$ consistent with the empirical distribution of the observed data. The left panel plots the density of $G_1$, and the right panel plots the density of $G_2$.}
\label{fig:two_consistent_priors}
\end{figure}

These differing dependence patterns in the distribution of economic mobility can yield different optimal decisions. To see this, I plot the resulting optimal selection of tracts treating both $G_1$ and $G_2$ as the true effect distributions in figure \ref{fig:two_consistent_decisions}. Visual inspection of figure \ref{fig:two_consistent_decisions} validates the intuition relating precision dependence to downstream decisions. When the effect distribution displays positive dependence between mobility and estimation precision, the resulting oracle decision-maker is more likely to select high-noise, high-mobility tracts, and vice versa. Comparing panels 1 and 2 of figure \ref{fig:two_consistent_decisions} suggests that different mobility distributions consistent with the mobility estimates of \citealt{Chettyetal2018aer} can yield drastically different decisions. Ambiguity around the distribution generating economic mobility across neighborhoods leads to ambiguity about the downstream decision.

\begin{figure}[H]
    \centering

    \begin{minipage}[t]{0.48\textwidth}
        \centering
        \includegraphics[
            width=\linewidth,
            height=4.7cm,
            keepaspectratio
        ]{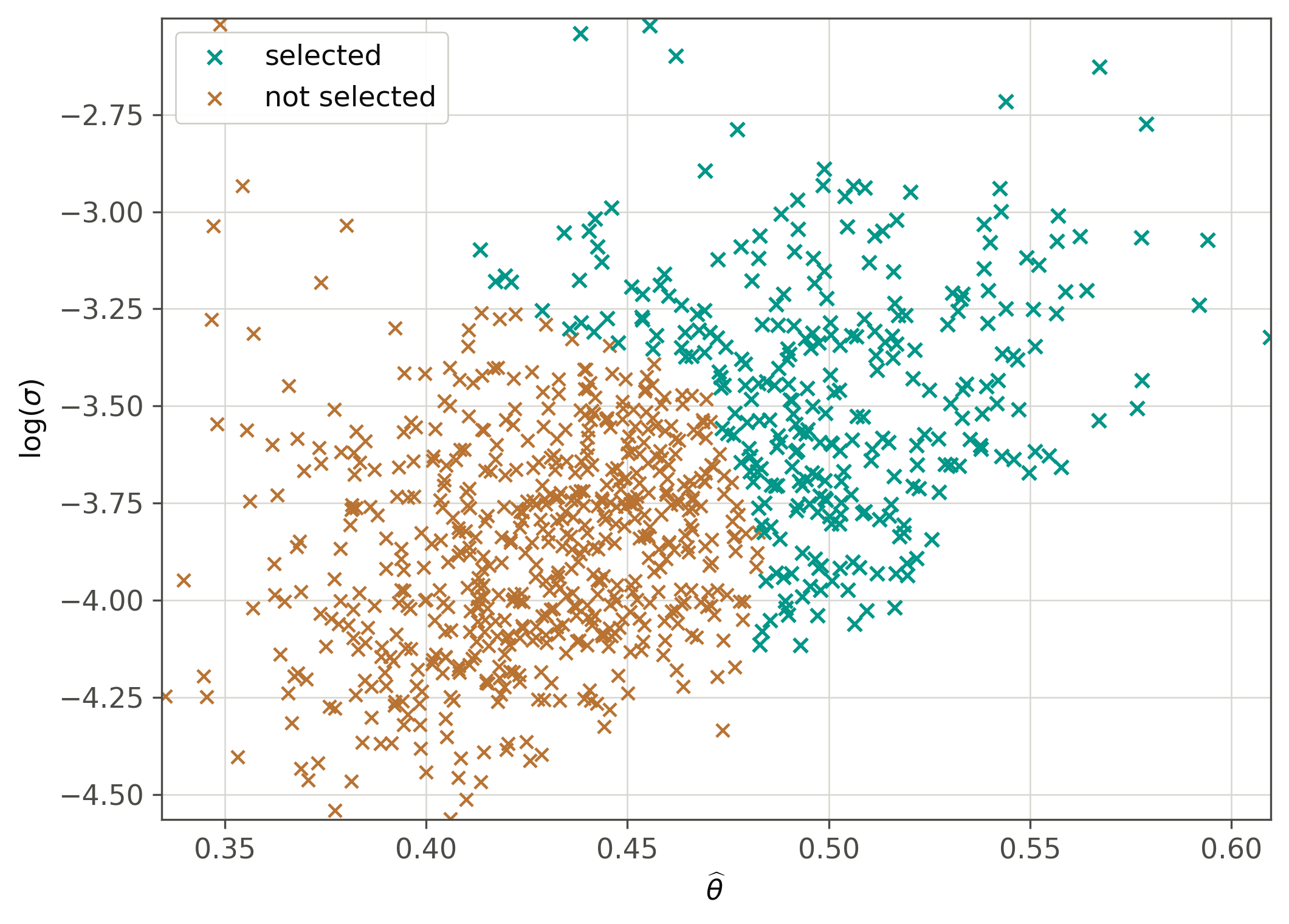}
    \end{minipage}
    \hfill
    \begin{minipage}[t]{0.48\textwidth}
        \centering
        \includegraphics[
            width=\linewidth,
            height=4.7cm,
            keepaspectratio
        ]{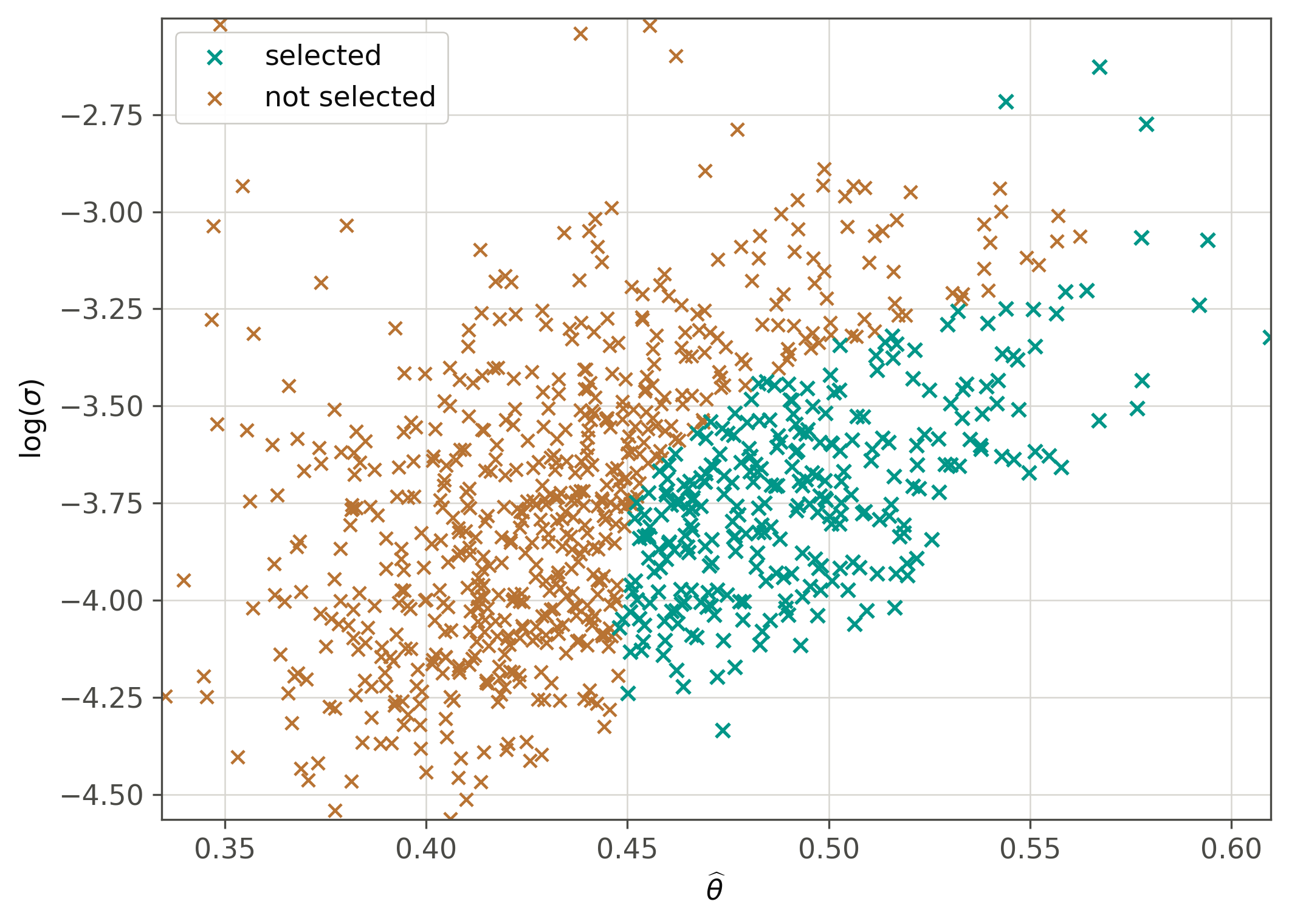}
    \end{minipage}
    \caption{Optimal decisions given $G_1$ and $G_2$. The left panel displays the optimal decision given $G_1$, and the right panel displays the optimal decision given $G_2$. Seattle census tracts are plotted according to their $(\hat\theta_i, \log\sigma_i)$. Selected tracts are colored in teal, and are otherwise colored in copper.}
    \label{fig:two_consistent_decisions}
\end{figure}

\subsection*{Appropriate Uncertainty Quantification}

Uncertainty around $G$ translates into uncertainty around the oracle decision-makers beliefs and actions. For example, can researchers be confident that an oracle decision-maker would select a given tract with high-estimated mobility over other tracts? Relatedly, which tracts can researchers be confident that an oracle decision-maker would select? Can researchers be confident that an oracle decision-maker prefers high-mobility low-precision tracts to low-mobility, high-precision tracts, and if so, by how much? For any two feasible decisions, can we be confident that an oracle decision-maker would prefer one over the other?

To address such questions, this paper proposes a notion of uniform coverage of conditional means. Let $\mu_G(t,s) := \bbE_G(\theta_i|\hat\theta_i = t, \sigma_i=s)$. Formally, we seek $\widehat{CS}(\cdot,\cdot)$ which satisfies:
\begin{equation}
\label{eq:uniform_coverage}
    \liminf_{n\to\infty}\inf_{G\in\mathcal{G}} P_{G}\left(
\mu_G(t,s) \in \widehat{CS}(t,s) ~ \text{for all}~ t,s
    \right)\geq 1-\alpha~.
    \end{equation}
To see why confidence sets satisfying \eqref{eq:uniform_coverage} suffice to address the inferential questions above, notice that such exercises involve pairwise comparisons of conditional mean mobility across many tracts. For example, determining the set of tracts in the oracle optimal selection involves determining those tracts for which conditional mean economic mobility is relatively large compared to other tracts. Evaluating the oracle decision-maker's dilemma between selecting noisy tracts with high estimated mobility and lower-mobility, lower-noise tracts involves comparing pairs of many such tracts. Comparing decisions involves comparing the conditional mean mobility at tracts selected by one decision but not the other. I therefore propose \eqref{eq:uniform_coverage} as a method that guarantees simultaneous validity of these comparisons.

It is instructive to compare confidence sets $\widehat{CS}(\cdot,\cdot)$ satisfying \eqref{eq:uniform_coverage} to those satisfying marginal coverage:
\begin{equation}
    \label{eq:marginal_coverage}
    \liminf_{n\to\infty}\inf_{G\in\mathcal{G}} P_{G}\left(
\mu_G(t,s) \in \widehat{CS}(t,s)
    \right)\geq 1-\alpha~ \text{for all}~ t,s~.
\end{equation}
While \eqref{eq:marginal_coverage} guarantees coverage of individual conditional means, it can fail to cover the many comparisons of conditional means that allow researchers to evaluate decisions. In simulations calibrated to the economic mobility estimates of \citealt{Chettyetal2018aer}, I compare the coverage properties of an approach providing marginal coverage,\footnote{Which extends methods developed in \citealt{IgnatiadisandWager2021} to allow for dependence between the $\theta_i$ and $\sigma_i$.} with the methods providing simultaneous coverage developed in this paper. Figure \ref{fig:coverage_sims} provides the results from this simulation. The left panel compares simultaneous coverage (the expression on the left-hand side of \eqref{eq:uniform_coverage}) for both approaches. In the right panel, I use confidence sets for conditional mean mobility to obtain a confidence set for a top-third selection according to conditional means over a grid of representative $\hat\theta_i, \sigma_i$ pairs. This latter confidence set reports the tracts which researchers can confidently include or exclude in an oracle optimal selection. The coverage results for this latter exercise are displayed in the right panel of figure \ref{fig:coverage_sims}.

\begin{figure}[H]
    \centering

    \begin{minipage}[t]{0.48\textwidth}
        \centering
        \includegraphics[
            width=\linewidth,
            height=4.7cm,
            keepaspectratio
        ]{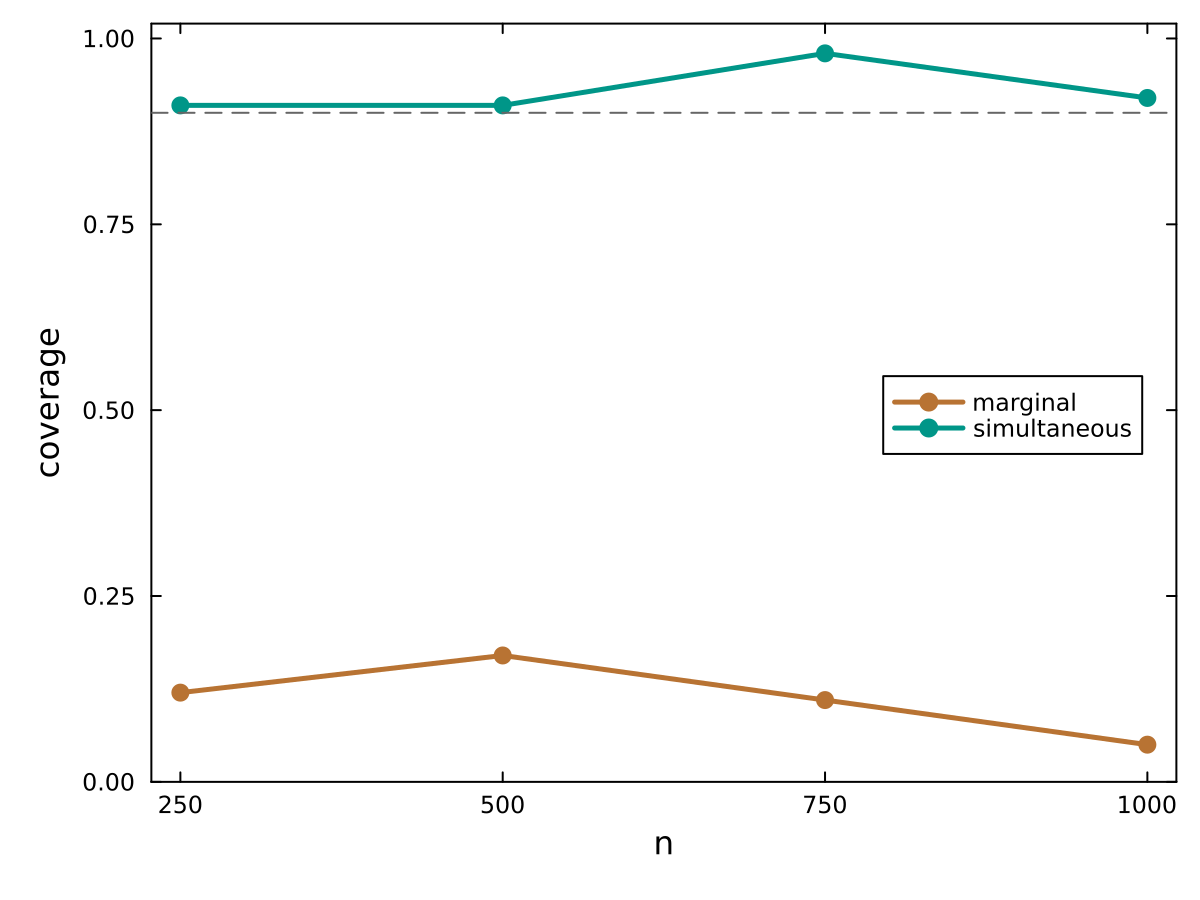}
    \end{minipage}
    \hfill
    \begin{minipage}[t]{0.48\textwidth}
        \centering
        \includegraphics[
            width=\linewidth,
            height=4.7cm,
            keepaspectratio
        ]{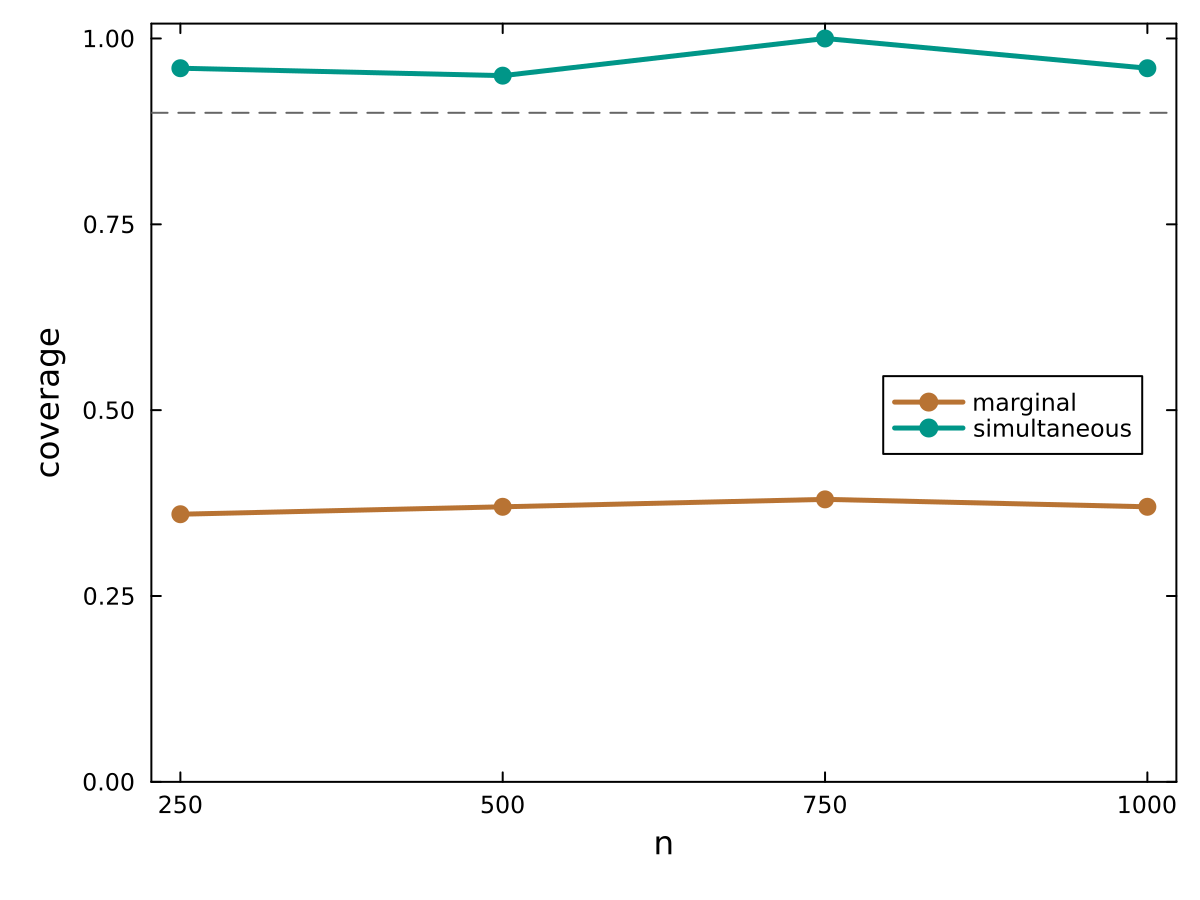}
    \end{minipage}
    \caption{Simulation study evidence for coverage for methods satisfying \eqref{eq:uniform_coverage}, plotted in teal, and \eqref{eq:marginal_coverage}, plotted in copper. The left panel plots simultaneous coverage of the conditional means $\mu_G(\hat\theta_i,\sigma_i)$, while the right panel plots coverage of the subset of tracts lying in the top-third according to their $\mu_G(\hat\theta_i,\sigma_i)$, i.e. the oracle optimal selection. }
    \label{fig:coverage_sims}
\end{figure}

Crucially, visual inspection of figure \ref{fig:coverage_sims} suggests that, not only do marginal confidence sets satisfying \eqref{eq:marginal_coverage} dramatically undercover when targeting simultaneous coverage of conditional means, but moreover dramatically undercover when used to make statements about the oracle decision-maker's problem. These are the exact statements of interest in this article. Consequently, this paper considers the problem of constructing confidence sets satisfying \eqref{eq:uniform_coverage}.

% \subsubsection*{Comparison to Inference on Latent Effects}

% \textbf{Par:} 

\subsection*{Further Examples}
\label{subsec:further_exmp}

To demonstrate the broad applicability of the setting described in this section, I list a few decision problems distinct from the selection problem of \citealt{Bergmanetal2023} considered thus far that arise in applied work in economics. In each of the examples described below, the oracle optimal decision depends on $G$ only through the conditional means of effects, or functions of the effects, given the observed data $\hat\theta_i$ and $\sigma_i$. Consequently, the same uncertainty quantification exercises considered above can be considered in each of the applications below using the method satisfying \eqref{eq:uniform_coverage} developed in this paper.

\begin{exmp}
\label{ex:selection_exmp_1}
    \textbf{(Selection)} The selection problem described above is not limited to the economic mobility application of \citealt{Bergmanetal2023}. For instance, \citealt{Angristetal2025} consider the problem of selecting the highest value-added school, given a choice set. \citealt{EFM2025} suggest that nudging patients to move from low to high value-added nursing homes within their market could improve patients' outcomes in underserved, rural areas. 

    In each of these problems, the welfare in \eqref{eq:util_max_welfare_defn} characterizes the decision-maker's objective of maximizing total value-added at selected neighborhoods, schools, or nursing homes. To capture different constraints on the decision-makers selection (i.e. selection of the best school, or the highest value-add nursing home in a given market), one can impose different constraints on the decision rule $\boldsymbol{\delta}$.
\end{exmp}

\begin{exmp}
    \label{ex:util_max_1}
    \textbf{(Utility Maximization)} A related problem to the above is that of utility, or welfare-maximization. For example, \citealt{Klineetal2022} are interested in selecting discriminatory employers $i$ with black-white callback differentials $\theta_i$, subject to firm-specific audit costs $K_i$. A positive callback differential $\theta_i$ corresponds to discrimination against black applicants, and vice-versa. Restricting the unit-level decision $\delta_i\in\{0,1\}$, I obtain the following welfare:\footnote{See section 11 in \citealt{Klineetal2022}, or \citealt{Chenetal2025} for further discussion of utility maximization in this setting.}
    \begin{equation*}
    \label{eq:util_max_loss}
        W^{(util)} (\boldsymbol{\delta}, \boldsymbol{\theta}) = \sum_{i=1}^n \delta_i(\theta_i-K_i)~.
    \end{equation*}
    Similar welfare functions, albeit with different action spaces and constraints, arise naturally in the methodological literature on empirical welfare analysis as discussed in \citealt{Moon2025}, who studies welfare maximization in the context of the Marginal Values of Public Funds (MVPF) of \citealt{HendrenandSprung_Keyser2020}. For a welfare more tailored to discrimination applications, we consider a welfare where $\delta_i \in \{0,1\}^2$:
    \begin{equation*}
    \label{eq:util_max_loss_ts}
        W^{(util2)} (\boldsymbol{\delta}, \boldsymbol{\theta}) = \sum_{i=1}^n \delta_{i,1}(\theta_i-K_i) + \sum_{i=1}^n \delta_{i,2}(-\theta_i-K_i)~.
    \end{equation*}
    This two-sided welfare corresponds to auditing for discrimination against either black or white applicants. Under the welfare above, the Bayes-optimal decision consists of the thresholding rules $\delta_{i,1}^{OB}(G) := \mathbb{1}(\bbE_G(\theta_i|\hat\theta_i,\sigma_i) \geq K_i)$ and $\delta_{i,2}^{OB}(G) := \mathbb{1}(\bbE_G(-\theta_i|\hat\theta_i,\sigma_i) \geq K_i)$.
\end{exmp}

\begin{exmp}
\label{ex:high_dim_test_1}
    \textbf{(High-Dimensional Testing)} Local false discovery rate control can be interpreted as a compound decision (see \citealt{SunandCai2007}, \citealt{xie2011optimal}). The discrimination correspondence study of \citealt{Klineetal2022}, considers the problem of detecting discriminatory firms. Restricting $\delta_i\in\{0,1\}$, one can interpret $\delta_i=1$ as labeling firm $i$ as discriminatory against black applicants. The following welfare for $\lambda\in(0,1)$:
    \begin{equation*}
    \label{eq:test_loss}
        W^{(test)}(\boldsymbol{\delta},\boldsymbol{\theta}) = -\left\{\sum_{i=1}^n \underbrace{\lambda \delta_i\mathbb{1}(\theta_i\le 0)}_{\text{False Discovery}} + \underbrace{(1-\lambda)(1-\delta_i)\mathbb{1}(\theta_i>0)}_{\text{False Non-Discovery}}\right\}~,
`    \end{equation*}
    corresponds to trading off false discoveries against false non-discoveries. In addition to the application of \citealt{Klineetal2022}, \citealt{GuKoenker2023} consider a similar welfare function to the above, derive the Bayes-optimal decision rule, and apply their methods to the problem of ranking U.S. dialysis centers. Under the welfare above, the Bayes-optimal decision follows a thresholding rule $\delta_i := \mathbb{1}(P_G(\theta_i\geq 0|\hat\theta_i,\sigma_i) \geq \lambda)$. Here, the posterior probability of discrimination against black applicants characterizes the Bayes-optimal decision. The methods in this paper extend to address cases when the conditional means of interest are $\bbE_G( b(\theta_i)|\hat\theta_i,\sigma_i)$ for arbitrary $b$, meaning that decisions based on conditional probabilities can be addressed using my methods.
\end{exmp}

\section{Equivalence to Multiple Testing}
\label{sec:equivalence_results}

Our goal is to obtain confidence sets satisfying \eqref{eq:uniform_coverage}. Before introducing the key observations that underlie the approach to inference developed in this paper, I first note that any confidence set for the effect distribution $G$ suffices to obtain confidence sets satisfying \eqref{eq:uniform_coverage}. However, such confidence sets for $G$ also provide simultaneous confidence bounds for any feature of $G$, not only the conditional means of interest to the oracle decision-maker of study in this paper. Consequently, naive constructions of such confidence sets for $G$ may be conservative for the problem \eqref{eq:uniform_coverage}. The desideratum motivating the approach to inference developed in this paper, informally stated, is to obtain confidence bands satisfying \eqref{eq:uniform_coverage} while focusing power on the conditional means of interest to the oracle decision-maker.

In service of this objective, I provide two key observations that allow a tractable approach to inference satisfying \eqref{eq:uniform_coverage}. I first establish an equivalence between \eqref{eq:uniform_coverage} and multiple testing, where the econometrician seeks to control the familywise error rate for a particular class of null hypotheses. I then obtain a characterization of these null hypotheses in terms of linear functionals of the effect distribution $G$, and conclude by remarking on the utility of these two observations.

To establish the first equivalence, I first define for each $t$ in $\bbR$ corresponding to a realized unit-level effect estimate, for each $s > 0$ corresponding to realized unit-level estimate precision, and for each $c$ in $\bbR$ corresponding to a potential value of conditional mean economic mobility, the following hypothesis:
\begin{equation}
\label{eq:null_hypotheses}
    H_{c,t,s}:\mu_G(t,s)  = c~. 
\end{equation}
This hypothesis corresponds to the statement that the conditional mean effect $\theta_i$ at a tract satisfying $\hat\theta_i = t$ and $\sigma_i=s$ equals $c$. Consider the problem of testing the family of null hypotheses $H_{c,t,s}$ while controlling the familywise error rate:
\begin{equation*}
    \FWER_G := P_G\left(
\left\{
\text{Reject at least one $(c,t,s)$ such that $H_{c,t,s}$ holds.}
\right\}
    \right)~.
\end{equation*}
In particular, the following proposition establishes equivalence between the problem of obtaining a testing procedure satisfying asymptotic control of the familywise error rate:
\begin{equation}
\label{eq:fwer_control}
    \limsup_{n\to\infty}\sup_{G\in\calG}\FWER_G \le \alpha~,
\end{equation}
and of obtaining a confidence set satisfying \eqref{eq:uniform_coverage}. 

\begin{prop}
    \label{prop:MHT_equivalence}
    Given any confidence band satisfying \eqref{eq:uniform_coverage}, a testing procedure satisfying \eqref{eq:fwer_control} can be constructed. Similarly, given any testing procedure satisfying \eqref{eq:fwer_control}, a confidence band satisfying \eqref{eq:uniform_coverage} can be constructed. In either case, the coverage probability of the confidence band equals one minus the familywise error rate of the corresponding test:
    \begin{equation*}
        P_{G}\left(
\mu_G(t,s) \in \widehat{CS}(t,s) ~ \text{for all}~ t,s
    \right) = 1-\FWER_G~.
    \end{equation*}
\end{prop}
\begin{proof}
    To see this, for any confidence band $\widehat{CS}(\cdot,\cdot)$, define $\phi_{c,t,s}$ as follows:
    \begin{equation}
    \label{eq:construct_phi_from_cs}
        \phi_{c,t,s} := \mathbb{1}(c\not\in \widehat{CS}(t,s))~.
    \end{equation}
    Similarly, \eqref{eq:construct_phi_from_cs} can be used to construct $\widehat{CS}(\cdot,\cdot)$ from $\phi_{c,t,s}$. The following equivalence of events:
    \begin{align*}
        \{\exists (t,s):\mu_G(t,s) \not\in \widehat{CS}(t,s)\} &= \{
        \exists (t,s):\phi_{\mu_G(t,s),t,s} = 1
        \}\\
        & = \{
        \exists (c,t,s): \text{$H_{c,t,s}$ is true, $\phi_{c,t,s}$ rejects}
        \}
    \end{align*}
    yields the proposition.
\end{proof}
Proposition \ref{prop:MHT_equivalence} is based on a construction where, for a family of tests of the hypotheses \eqref{eq:null_hypotheses}, the surviving nulls correspond to inclusions in a uniform confidence band, and vice versa. Consequently, the event that the uniform confidence band fails to cover conditional mean effects at some potential realization of the data corresponds to a false discovery in the corresponding multiple testing problem. Figure \ref{fig:MHT_equivalence} visualizes this intuition:

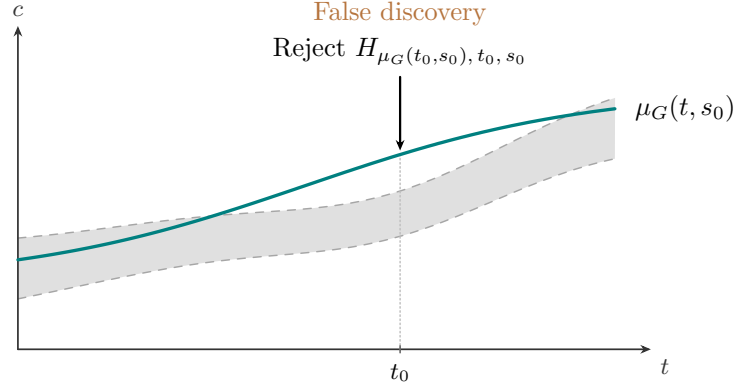
\begin{figure}[H]
\centering
\begingroup

\definecolor{fdCopper}{RGB}{182,116,61}
\def\fdT{0.64}

\begin{tikzpicture}[
    x=1cm,
    y=1cm,
    line cap=round,
    line join=round,
    >={Stealth[length=4pt,width=3.3pt]},
    declare function={
        fdmean(\x)=0.19+0.55/(1+exp(-4.2*(\x-0.48)));
        fdwlo(\x)=0.080+0.030*(2*\x-1)^2;
        fdwup(\x)=0.045+0.018*(2*\x-1)^2;
        fdmove(\x)=0.15*exp(-((\x-0.64)/0.29)^2);
        fdlo(\x)=fdmean(\x)-fdwlo(\x)-fdmove(\x);
        fdup(\x)=fdmean(\x)+fdwup(\x)-fdmove(\x);
    }
]

\begin{scope}[
    shift={(1.05,0.55)},
    x=7.90cm,
    y=4.65cm
]

% Confidence band.
\path[fill=black!12]
    plot[domain=0:1,samples=180,variable=\x]
        (\x,{fdlo(\x)})
    --
    plot[domain=1:0,samples=180,variable=\x]
        (\x,{fdup(\x)})
    -- cycle;

\draw[black!35,dashed,line width=.55pt]
    plot[domain=0:1,samples=180,variable=\x]
        (\x,{fdlo(\x)});

\draw[black!35,dashed,line width=.55pt]
    plot[domain=0:1,samples=180,variable=\x]
        (\x,{fdup(\x)});

% Axes.
\draw[->,black!80,line width=.65pt]
    (0,0)--(1.06,0)
    node[
        anchor=north west,
        font=\fontsize{9}{11}\selectfont
    ] {$t$};

\draw[->,black!80,line width=.65pt]
    (0,0)--(0,.92)
    node[
        anchor=south,
        font=\fontsize{9}{11}\selectfont
    ] {$c$};

% Guide, tick, and label at t_0.
\draw[black!25,densely dotted,line width=.55pt]
    (\fdT,0)--(\fdT,{fdmean(\fdT)});

\draw[black!65,line width=.55pt]
    (\fdT,-.010)--(\fdT,.010);

\node[
    anchor=north,
    yshift=-2pt,
    font=\fontsize{9}{11}\selectfont
]
    at (\fdT,0) {$t_0$};

% True conditional mean, holding s_0 fixed.
\draw[teal,line width=1.25pt]
    plot[domain=0:1,samples=180,variable=\x]
        (\x,{fdmean(\x)});

\node[
    anchor=west,
    text=black,
    xshift=4pt,
    font=\fontsize{10}{12}\selectfont\mdseries
]
    at (1,{fdmean(1)})
    {$\mu_G(t,s_0)$};

% False discovery at (t_0,s_0).
\node[
    anchor=south,
    align=center,
    text=black,
    inner sep=2pt,
    font=\fontsize{10}{12}\selectfont\mdseries
]
    (fdfalse) at (\fdT,.79)
    {\textcolor{fdCopper}{False discovery}\\[-.3ex]
     Reject $H_{\mu_G(t_0,s_0),\,t_0,\,s_0}$};

\draw[
    ->,
    line width=.85pt,
    shorten <=2pt,
    shorten >=1.5pt
]
    (fdfalse.south)--(\fdT,{fdmean(\fdT)});

\end{scope}

\end{tikzpicture}

\endgroup
\caption{A false rejection of
$H_{\mu_G(t_0,s_0),t_0,s_0}$ corresponds to failure of the confidence
band to cover the true conditional mean at $(t_0,s_0)$.}
\label{fig:MHT_equivalence}
\end{figure}

\subsection*{Simplifying the Null Hypotheses}

The above proposition \ref{prop:MHT_equivalence} establishes equivalence between the problem of interest in this paper, \eqref{eq:uniform_coverage}, and the classical problem of multiple hypothesis testing with familywise error rate control (see e.g. chapter 9 of \citealt{lehmann_romano_2022}). While this is a useful characterization, it is unclear how the representation afforded by proposition \ref{prop:MHT_equivalence} helps obtain an inferential procedure satisfying both \eqref{eq:uniform_coverage} and the inferential desideratum. In what follows, I obtain a characterization of the null hypotheses $H_{c,t,s}$ defining this multiple testing problem in terms of linear functionals of $G$. This characterization identifies the decision-relevant features of $G$, and thus helps address the inferential desideratum introduced above.

Extending arguments due to \citealt{IgnatiadisandWager2021} provides a convenient characterization of the null hypotheses $H_{c,t,s}$ as restrictions on linear functionals of the effect distribution $G$. To this end, notice that conditional mean effects can be written:
\begin{equation}
\label{eq:fractional_characterization}
    \mu_G(t,s) = \frac{\int \theta\varphi\left(
    \frac{t-\theta}{s}
    \right)g_G(s,\theta)d\theta}{\int \varphi\left(
    \frac{t-\theta}{s}
    \right)g_G(s,\theta)d\theta}~.
\end{equation}
To understand the consequences of \eqref{eq:fractional_characterization}, notice that the numerator and denominator of the right hand side of the above are both linear functionals of the effect distribution $G$. Consequently, the hypotheses $H_{c,t,s}$ can be rewritten:
\begin{equation*}
    H_{c,t,s}:L_{(c,t,s)}(G):=\int \left(\theta-c\right)\varphi\left(
    \frac{t-\theta}{s}
    \right)g_G(s,\theta)d\theta = 0~.
\end{equation*}
Observe that $L_{(c,t,s)}(G)$ is a linear functional of $G$. To simplify notation, let $\gamma$ denote the triple $(c,t,s)$, and rewrite the above family of hypotheses as $H_\gamma:L_\gamma(G) = 0$. Let $\Gamma$ denote $\mathbb{R}^2\times\bbR_+$, which contains each $\gamma$, and let $\Gamma(G)$ denote the collection of true nulls $\{\gamma\in\Gamma:H_\gamma\text{ holds}\}$. 

Consequently, the problem \eqref{eq:uniform_coverage} is not only equivalent to a multiple testing problem, but to a multiple testing problem where the null hypotheses entail restrictions on linear functionals of the effect distribution $G$. First, notice that this characterization isolates the features of the effect distribution $G$ to which the econometrician should focus power: the linear functionals $L_\gamma(G)$. Moreover, this characterization identifies the exact features of $G$ that an inferential procedure must cover: the $L_\gamma(G)$ such that $L_\gamma(G) = 0$. That is, to obtain a testing procedure satisfying \eqref{eq:fwer_control}, it is necessary and sufficient to cover the $L_\gamma(G)$ for which $H_\gamma$ holds. To build on this intuition, section \ref{sec:methods} develops a step-down approach to inference that seeks to protect exactly these null linear functionals, while focusing power on the collection of functionals $(L_\gamma(G))_{\gamma\in \Gamma}$.

\begin{rem}
    \label{rem:set-inference-connection}
    This paper is not the first to remark on equivalence between problems of uniform coverage as in \eqref{eq:uniform_coverage} and multiple hypothesis testing. Work due to \citealt{RomanoShaikh2010} establishes a similar equivalence between inference for the partially identified set defined by criterion functions (see e.g. \citealt{ChernozhukovHongTamer2007}). Our proposition \ref{prop:MHT_equivalence} is closely related to these previous equivalence results. In particular, \eqref{eq:uniform_coverage} can be alternatively viewed as a special case of the problem of inference for the partially identified set: coverage of the conditional mean function as in \eqref{eq:uniform_coverage} is equivalent to coverage of the zero-level set of $L_\gamma(G)$ as functions of $\gamma$, which coincides with the set of true nulls $\Gamma(G)$. Motivated by their equivalence results, \citealt{RomanoShaikh2010} obtain a step-down approach to inference for the identified set, providing further motivation for the step-down methods developed in section \ref{sec:methods}. 
\end{rem}

\section{Methods}
\label{sec:methods}

This section develops methods satisfying \eqref{eq:uniform_coverage}, applying the equivalence results developed in section \ref{sec:equivalence_results}. In particular, the results of section~\ref{sec:equivalence_results} suggest that it is necessary and sufficient to construct a familywise-error-rate-controlling family of tests of the hypotheses $\left\{H_\gamma:L_\gamma(G)=0\right\}_{\gamma\in\Gamma}$. My methods draw on a classical literature on multiple testing \citep{RomanoWolf2005, RomanoShaikh2010}.  The present setting differs from the classical settings in the key respect that estimators of the $L_\gamma(G)$ need not be unbiased.\footnote{\citealt{Pensky2015} demonstrates that in the homoskedastic version of the model considered in \eqref{eq:EB_DGP}, the minimax rate of estimation for many linear functionals of $G$ is not $n^{-1/2}$. Unbiased estimation therefore entails estimators with variance that grows with $n$. It therefore may be desirable to introduce bias in estimation to allow for reduced variance in estimation.}

To address this departure from the classical setting, I develop a bias-aware, step-down procedure for inference that takes arbitrary estimators of the $L_\gamma(G)$ and allows for biased estimation in subsection \ref{sec:step_down_details}. To obtain test statistics, I build on work due to \citealt{Donoho1994} and \citealt{IgnatiadisandWager2021} to derive minimax estimators which, for each functional $L_\gamma(G)$, minimize worst-case bias subject to a constraint on variance. I provide a high-level description of these test statistics in remark \ref{rem:choice-of-t-stat}, and provide an in-depth description and derivation of these test statistics in appendix \ref{app:minimax_details}. Given a bias-aware step-down procedure, I turn to asymptotics. Proposition \ref{prop:step_down_validity_est_cv} shows that this step-down procedure is uniformly asymptotically valid under certain validity and monotonicity conditions on critical values, stated in assumption \ref{assu:cv_validity_monotonicity}. Theorem \ref{thm:bootstrap_validity} verifies these conditions for multiplier bootstrap critical values. Verifying these conditions is complicated by the fact that the minimax test statistics need not be Donsker, or even VC-type, after studentization.

\subsection{Step-down Procedure for Inference}
\label{sec:step_down_details}

By proposition \ref{prop:MHT_equivalence} and the subsequent characterization of the null hypotheses $H_\gamma$, it suffices to obtain a collection of tests $\{\phi_\gamma\}_{\gamma\in\Gamma}$ satisfying \eqref{eq:fwer_control}. To this end, first assume that, for every $\gamma\in\Gamma$, there exists an estimator $\widehat L_{n,\gamma}$ of the functional $L_\gamma(G)$ that can be written as a sample mean:
\begin{equation*}
    \widehat L_{n,\gamma}
        :=
        \frac{1}{n}
        \sum_{i=1}^n
        Q_\gamma(\widehat\theta_i,\sigma_i)~,
\end{equation*}
where $Q_\gamma$ is some function from $\bbR\times\bbR_+$, the space of potential realizations of $(\hat\theta_i,\sigma_i)$ to the real numbers $\bbR$. Before describing my approach to inference, I define the bias of $\hat{L}_{n,\gamma}$ and variance of $Q_\gamma$ as follows:
\begin{equation*}
B_\gamma(G) := \bbE_G\hat{L}_{n,\gamma} - L_\gamma(G)~,\quad \sigma_G^2(Q_\gamma) = \var_G(Q_\gamma) = \bbE_G\left(Q_\gamma(\hat\theta_i,\sigma_i)^2\right) - \bbE_G\left(Q_\gamma(\hat\theta_i,\sigma_i)\right)^2~.
\end{equation*}
The bias of $\hat L_{n,\gamma}$ is typically unknown, and I therefore denote by $\bar{B}_\gamma := \sup_{G\in\calG}|B_\gamma(G)|$ the worst-case bias over the class $\calG$. Finally, for any $\Gamma' \subseteq\Gamma$, let $J(x,\Gamma',G)$ denote the distribution of the supremum of the estimation error of the $\hat{L}_{n,\gamma}$ over $\Gamma'$:
\begin{equation*}
    \label{eq:oracle-cdf}
        J_n(x,\Gamma',G)
        :=
        P_G\!\left(
            \sup_{\gamma\in\Gamma'}
            \left|
                \hat T_{n,\gamma}
            \right|
            \leq x
        \right)~, \quad \hat T_{n,\gamma} :=\frac{
                    \sqrt n\left(
                        \widehat L_{n,\gamma}
                        -
                        \mathbb E_G\widehat L_{n,\gamma}
                    \right)
                }{
                    \sigma_G(Q_\gamma)
                }~.
    \end{equation*}
Finally, I define the corresponding $1-\alpha$ critical value by $c(1-\alpha,\Gamma',G)$, defined as: 
\begin{equation*}
    c(1-\alpha,\Gamma',G):=
        \inf\left\{
            x\geq0:
            J(x,\Gamma',G)\geq1-\alpha
        \right\}~.
\end{equation*}

Ideally, given the equivalence outlined in section \ref{sec:equivalence_results}, hypotheses $H_\gamma$ should be rejected according to critical values $c(1-\alpha,\Gamma(G),G)$. However, because these critical values depend on the true nulls $\Gamma(G)$, which is the exact object of inferential interest, such approaches are not feasible. To obtain a feasible approach to inference that approximates this ideal benchmark, I build on and extend work due to \citealt{RomanoWolf2005} and \citealt{RomanoShaikh2010} to obtain a step-down, bias-aware procedure that iteratively rejects hypotheses $H_\gamma$. Crucially, this procedure highlights that it is not necessary to control for false discoveries among already-rejected hypotheses, targeting coverage of the functionals $L_\gamma(G)$ for which $H_\gamma$ holds.
\begin{algorithm}[Step-Down Procedure (Known Critical Values)]
\label{alg:step-down-procedure-known-cv}

\medskip

\noindent\textbf{Initialization.}
Set $\Gamma_0 := \Gamma$.

\medskip

\noindent\textbf{Iteration.}
For $k=1,2,\ldots$, set
\[
    \Gamma_k
    :=
    \left\{
        \gamma\in\Gamma_{k-1} :
        \left|\widehat L_{n,\gamma}\right|
        \leq
        \overline B_\gamma
        +
        c(1-\alpha,\Gamma_{k-1},G)
        \frac{\sigma_G(Q_\gamma)}{\sqrt n}
    \right\}~.
\]
Stop at the first $k$ such that $\Gamma_k=\Gamma_{k-1}$,
and return $\Gamma_k$.

\end{algorithm}
Where $\Gamma_{\hat k}^{sd}$ denotes the outcome of this step-down procedure, define $\widetilde{CS}^{sd}(\cdot,\cdot)$ as follows:
\begin{equation}
    \label{eq:step_down_definition}
    \widetilde{CS}^{sd}(t,s) := \{c : (c,t,s) \in \Gamma_{\hat{k}}\}~.
\end{equation}
\begin{prop}
    \label{prop:step_down_finite_sample}
    The step-down uniform confidence set \eqref{eq:step_down_definition} satisfies:
    \begin{equation*}
        P_G\left(
\mu_G(t,s) \in \widetilde{CS}^{sd}(t,s)\text{ for all }t,s
        \right)\ge 1-\alpha~,
    \end{equation*}
    for all $G$ in $\calG$.
\end{prop}
The proof of proposition \ref{prop:step_down_finite_sample} is included in appendix \ref{app:proofs_of_main_results}. At this point, I comment on four features of the procedure. First, the set of true nulls $\Gamma(G)$ encodes shape restrictions on the conditional mean function that are implied by the hierarchical sampling model \eqref{eq:EB_DGP}. The step-down procedure treats $\Gamma(G)$ only as the unknown set of true null hypotheses and does not otherwise exploit this structure. Procedures that exploit the shape restrictions implied by the \eqref{eq:EB_DGP} may therefore be more informative, although constructing such procedures is non-obvious and an interesting direction for future research. Second, the step-down construction eliminates the need to know $\Gamma(G)$ when determining the collection of hypotheses over which critical values are computed. The oracle critical values nevertheless continue to depend on the unknown distribution $G$, because the distribution of $\widehat L_{n,\gamma}$ depends on the latent distribution generating the data. Subsection~\ref{sec:asymptotic_validity} demonstrates how these distributions and critical values can be approximated using a multiplier bootstrap. Third, the construction above uses simultaneous, two-sided, bias-aware confidence sets for the functionals $L_\gamma(G)$. For ease of exposition, the maximum-bias and sampling-error components enter the confidence bounds additively. As discussed by \citet{ArmstrongKolesar2018}, this construction implicitly permits the bias to be chosen adversarially for each realization of the sampling error. One can address this source of conservatism by developing simultaneous analogues of the folded-normal critical values of \citet{ArmstrongKolesar2018}. Finally, the procedure takes the estimators $\widehat L_{n,\gamma}$, or equivalently the functions $Q_\gamma$, as given. As discussed previously, the choice of estimator poses a crucial tradeoff between bias and variance. In Appendix \ref{app:minimax_details}, I extend the methods of \citet{Donoho1994}, \citet{ArmstrongKolesar2018}, and \citet{IgnatiadisandWager2021} to construct minimax estimators that navigate this tradeoff in a principled manner. These are the estimators used in my bootstrap-validity analysis and ultimately implemented in the applications. Remark \ref{rem:choice-of-t-stat} details this choice of estimator.

\begin{rem}[Choice of Test Statistic]
\label{rem:choice-of-t-stat}
So far, the step-down methods provided allow for an arbitrary choice of test-statistic. I build on work due to \citet{Donoho1994}, \citet{ArmstrongKolesar2018}, and \citet{IgnatiadisandWager2021} to provide a principled choice of test-statistic. In Appendix \ref{app:minimax_details}, I derive a class of minimax estimators that allow for (i) rich forms of dependence between $\theta_i$ and $\sigma_i$; and (ii) a method for trading off worst-case bias and variance. Let $\widehat{\var}(Q_\gamma)$ denote an estimator of the variance of $Q_\gamma$. In appendix \ref{app:minimax_details}, I show how to obtain $Q_\gamma$ which solve, for each $\gamma$ in $\Gamma$, the following optimization problem:
    \begin{equation}
    \label{eq:minimax_program_intuition}
    \min_{Q_\gamma}\sup_{G\in\calG} B_\gamma(G)^2~\text{s.t.}~\frac{1}{n}\widehat{\var}(Q) \leq V_n~,
\end{equation}
    where $V_n$ denotes an upper constraint on the variance of $\hat{L}_{n,\gamma}$. Solutions to the program in \eqref{eq:minimax_program_intuition} minimize maximum bias over $\calG$, subject to a variance constraint, and are in this sense minimax. Alongside characterizing $Q_\gamma$ in \ref{app:minimax_details}, I discuss how to choose $V_n$ to trade off bias and variance for the purposes of the step-down procedure considered in this paper. Finally, similar minimax estimators are derived in \citealt{IgnatiadisandWager2021}, and are used to obtain marginal coverage of individual conditional means. I  emphasize that the affine minimax estimators developed and used in this paper differ from those of \citet{IgnatiadisandWager2021}, in that my affine minimax estimators allow for heteroskedasticity, and in particular dependence between the $\theta_i$ and $\sigma_i$. Finally, I emphasize that choosing minimax test-statistics for each functional $L_\gamma(G)$ allows researchers to focus power on these decision-relevant features of $G$.
\end{rem}

\subsection{Asymptotic Validity}
\label{sec:asymptotic_validity}

The above approach for constructing confidence sets satisfying \eqref{eq:uniform_coverage} assumes that the critical values $c(1-\alpha,\Gamma',G)$ are known. In practice, these critical values must be estimated. This section formalizes this estimation problem and provides a multiplier bootstrap algorithm to do so. Similarly, the $\sigma_G(Q_\gamma)$ are unknown, and are assumed to be estimated by $\hat\sigma_n(Q_\gamma)$. I first establish general sufficient conditions for estimated critical values to deliver an asymptotically valid approach to inference.

\begin{assu}
\label{assu:cv_validity_monotonicity}
    $\hat{c}_n$ satisfy validity and monotonicity, in the following sense:
    
    \noindent (i) $\hat{c}_n$ are valid in the sense that:
    \begin{equation*}
\label{eq:cv_validity}
    \limsup_{n\to\infty}\sup_{G\in\calG} P_G\left(
\sup_{\substack{\gamma\in \Gamma(G)}}
\left|
\frac{
                    \sqrt n\left(
                        \widehat L_{n,\gamma}
                        -
                        \mathbb E_G\widehat L_{n,\gamma}
                    \right)
                }{
                    \hat\sigma_n(Q_\gamma)
                }
\right| >\hat{c}_n(1-\alpha,\Gamma(G)) 
    \right)\leq \alpha~.
\end{equation*}

    \noindent (ii) $\hat{c}_n$ are monotone, in the sense that whenever $\Gamma_1\subseteq\Gamma_2\subseteq\Gamma$:
\begin{equation*}
\label{eq:cv_monotonicity_1}
    \hat{c}_n(1-\alpha,\Gamma_2)\geq \hat{c}_n(1-\alpha,\Gamma_1)~.
    \end{equation*}
\end{assu}

Towards interpretation, part (i) of the above assumption suggests that the critical values $\hat{c}_n$ need not be valid for every possible subset of $\Gamma$, but rather only for $\Gamma(G)$. In order to correct for the fact that critical values containing $\Gamma(G)$ may not be valid, part (ii) of the above assumption imposes monotonicity of these critical values in the sense of set inclusion. Proposition \ref{prop:step_down_validity_est_cv} verifies that these assumptions are sufficient for the uniform validity of this paper's methods.

Under this assumption, proposition \ref{prop:step_down_validity_est_cv}, stated below, ensures the asymptotic validity of the following algorithm:
\begin{algorithm}[Step-Down Procedure]
\label{alg:step-down-procedure}

\medskip

\noindent\textbf{Initialization.}
Set $\Gamma_0 := \Gamma$.

\medskip

\noindent\textbf{Iteration.}
For $k=1,2,\ldots$, set
\[
    \Gamma_k
    :=
    \left\{
        \gamma\in\Gamma_{k-1} :
        \left|\widehat L_{n,\gamma}\right|
        \leq
        \overline B_\gamma
        +
        \hat{c}_n(1-\alpha,\Gamma_{k-1})
        \frac{\hat\sigma_n(Q_\gamma)}{\sqrt n}
    \right\}~.
\]
Stop at the first $k$ such that $\Gamma_k=\Gamma_{k-1}$,
and return $\Gamma_k$.

\end{algorithm}
Where $\Gamma_{\hat k}^{sd}$ denotes the outcome of this step-down procedure, define $\widehat{CS}^{sd}(\cdot,\cdot)$ as follows:
\begin{equation}
    \label{eq:step_down_definition_asymp}
    \widehat{CS}^{sd}(t,s) := \{c : (c,t,s) \in \Gamma_{\hat{k}}^{sd}\}~.
\end{equation}

\begin{prop}
\label{prop:step_down_validity_est_cv}
    Suppose that assumption \ref{assu:cv_validity_monotonicity} holds. Then $\widehat{CS}^{sd}(\cdot,\cdot)$ satisfies \eqref{eq:uniform_coverage}.
\end{prop}
The proof of proposition \ref{prop:step_down_validity_est_cv} is included in appendix \ref{app:proofs_of_main_results}.

To obtain critical values satisfying \eqref{assu:cv_validity_monotonicity}, I proceed via the multiplier bootstrap. Let $\xi_i$ be i.i.d. standard normal random variables for $i=1,\ldots,n$, and denote by $\xi$ the collection $(\xi_i)_{i=1}^n$. For some realization of $\xi$, I define $\hat{L}_{n,\gamma}(\xi)$ as follows:
\begin{equation*}
    \hat{L}_{n,\gamma}(\xi) := \frac{1}{n}\sum_{i=1}^n \xi_i \left(Q_\gamma(\hat\theta_i,\sigma_i) -\hat{L}_{n,\gamma}\right)~.
\end{equation*}
Let $J_{n}^{MB}(x,\Gamma)$ denote the distribution of the multiplier bootstrap statistic, holding $\hat\theta_i$ and $\sigma_i$ fixed, and integrating over the law of $\xi$:
\begin{equation*}
    J_{n}^{MB}(x,\Gamma) = P_\xi\left(
\sup_{\gamma\in\Gamma}
\left|
\frac{\sqrt{n}\hat{L}_{n,\gamma}(\xi)}{\hat\sigma_n(Q_\gamma)}
\right|\le x
    \right)~.
\end{equation*}
I define the multiplier bootstrap critical values $\hat{c}_n(1-\alpha,\Gamma)$ as follows:
\begin{equation*}
    \hat{c}_n(1-\alpha,\Gamma) := \inf\left\{
x\ge 0: J_{n}^{MB}(x,\Gamma)\geq 1-\alpha
    \right\}~.
\end{equation*}
In appendix \ref{app:bootstrap_theory}, I demonstrate that these multiplier bootstrap critical values satisfy assumption \ref{assu:cv_validity_monotonicity}. The following theorem is proven therein.
\begin{theorem}
\label{thm:bootstrap_validity}
    Let $ n \to\infty$, and let assumptions \ref{assu:smoothness}, \ref{assu:modulus_binding}, and \ref{assu:pilot_independence} hold. Then the multiplier bootstrap critical values $\hat{c}_n(1-\alpha,\Gamma)$ satisfy assumption \ref{assu:cv_validity_monotonicity}, and  $\widehat{CS}^{sd}(\cdot,\cdot)$ using the bootstrap critical values satisfies \eqref{eq:uniform_coverage}.
\end{theorem}

The proof of theorem \ref{thm:bootstrap_validity} is contained in appendix \ref{app:bootstrap_theory}.

The assumptions in theorem \ref{thm:bootstrap_validity} are largely standard. Assumption \ref{assu:smoothness} establishes bounds on the support of $\sigma$ and the H\"older smoothness of the marginal distribution of $\sigma$, and of the conditional distribution of $\theta$ given $\sigma$, as a function of $\sigma$. Intuitively, assumption \ref{assu:smoothness} ensures that the conditional law of the effects $\theta$ given precision $\sigma$ cannot vary too much with small changes in $\sigma$. Assumption \ref{assu:smoothness} is similar to assumptions 3 and 4 of \citealt{chen2024}. Assumption \ref{assu:modulus_binding} asserts that the variance constraint defining the minimax problem \eqref{eq:minimax_program_intuition} binds, and additionally imposes the same condition as assumption A in theorem 2 in \citealt{IgnatiadisandWager2021}. Assumption \ref{assu:pilot_independence} is identical to assumption E in theorem 2 in \citealt{IgnatiadisandWager2021}. Finally, while not needed for validity, assumption \ref{assu:variance_uniform_estimability} can be used to show that estimation error in the $\hat\sigma_n$ is negligible, as shown in theorem \ref{thm:cv_assumptions_verified_var_est}. Assumption \ref{assu:variance_uniform_estimability} is a standard consistency condition on the variance estimators $\hat\sigma_n(Q_\gamma)$, which follows condition H4 of \citealt{CCK_2016}. 

The proof of theorem \ref{thm:bootstrap_validity} is contained in appendix \ref{app:bootstrap_theory}. The key challenge in proving theorem \ref{thm:bootstrap_validity} is that the collection of studentized statistics $\hat T_{n,\gamma}$ need not be Donsker, or even VC-type, since variances are not bounded away from zero.\footnote{Similar challenges with studentization arise in \citealt{chetverikov2012}.} The proof proceeds in three steps. I first verify the conditions needed to obtain a Gaussian coupling for the statistic $\hat T_{n,\gamma}$ via theorem 2.1 of \citealt{CCK_2014}. I then obtain an analogous Gaussian coupling for the corresponding multiplier bootstrap statistic, extending arguments due to \citealt{CCK_2016} to allow for the fact that the studentized test-statistics are not VC-type. Finally, application of anti-concentration results due to \citealt{CCK_2014} and \citealt{CCK_ac_2014} yields the necessary convergence of distribution functions for the bootstrap critical values to satisfy the conditions in assumption \ref{assu:cv_validity_monotonicity}. 

\section{Application: Economic Mobility Revisited}
\label{sec:economic-mobility-application}

Recent work asks studies the geography of opportunity in the United States \citep{Chettyetal2014,ChettyandHendren2018,Chettyetal2018aer}. Most recently, \citealt{Chettyetal2018aer} derive and report estimates of economic mobility in US commuting zones and neighborhoods, and advocate the use of these estimates for downstream policy targeting tasks. For example, in his testimony before the congressional budget committee \citep{Chetty2014MobilityTestimony}, Chetty stated:
\begin{quote}
    Perhaps the most cost-effective way to improve mobility may be to publicize local statistics on economic mobility and other related outcomes. Simply drawing attention to the areas that need improvement can motivate local policy makers to take action.
\end{quote}
Subsequent work has applied such local statistics to downstream policy tasks. In this section, I revisit a particular example of such an exercise, namely the policy targeting exercise of \citealt{Bergmanetal2023} considered throughout this paper. In particular, \citealt{Bergmanetal2023} seek to draw the attention of low-income housing voucher recipients to high-mobility neighborhoods. As described in section \ref{sec:setup}, they select a subset of high-mobility neighborhoods in Seattle and inform treated voucher recipients, as part of the Creating Moves to Opportunity program, that these neighborhoods are high mobility.

A recent debate highlights that the economic mobility estimates may not be sufficiently informative for informative decision making in the spirit of \citet{Bergmanetal2023}. For example, \citet{Mogstadetal2023} find that the economic mobility estimates of \citet{Chettyetal2018aer} may contain more noise than signal, and that one cannot confidently conclude that the tracts selected in the spirit of \citet{Bergmanetal2023} lie in the true top third according to true economic mobility. Similarly, \citet{petrouzeniou_and_shaikh_2025} find similar results even when relaxing the notion of type-I-error control to allow researchers to focus power on tracts of interest. Motivated by these findings, my applications study the effectiveness of tract-level mobility estimates according to an oracle decision-maker with knowledge of the distribution of economic mobility across tracts.

In order to shed light on the informativeness of these economic mobility estimates for decision-making, I consider two exercises. In the first exercise, I ask which tracts researchers can confidently conclude an oracle decision maker would select as high mobility. While the results of this exercise can be suggestive of differences in mobility across tracts, according to an oracle decision-maker, the magnitudes of these differences may be small. To rule out this possibility, I ask in my second exercise how large of an improvement in economic mobility researchers can confidently conclude an oracle decision maker would associate with moving from a low- to a high-mobility neighborhood. Large gaps in conditional mean mobility between low and high estimated mobility tracts suggest that an oracle decision-maker believes that tracts with low estimated mobility are, on average, substantially lower-opportunity than those with high estimated mobility. The remainder of this section provides results from these two exercises.

\subsection*{Inference on the oracle selection of tracts}
\label{subsec:mobility-oracle-selection}

Consider the first exercise: obtaining a confidence set for the oracle decision. Recall that the optimal, oracle selection of tracts selects the top-third of census tracts according to their conditional mean mobility $\mu_G(\hat\theta_i,\sigma_i)$. To construct a confidence set for this decision using the confidence bands derived in this paper, let $\widehat{CS}_i^{sd}$ be a confidence set for conditional mean mobility at tract $i$. In particular, given $\widehat{CS}^{sd}(\cdot,\cdot)$, take $\widehat{CS}_i^{sd} := \widehat{CS}^{sd}(\hat\theta_i,\sigma_i)$. Confidence bounds for pairwise differences in conditional means can be constructed by comparing the infimum of $\widehat{CS}_i$ to the supremum of $\widehat{CS}_j$, for $i\ne j$, and vice versa. Consequently, one can construct simultaneous confidence sets for the ranks of tracts according to conditional mean mobility.\footnote{See \citealt{Mogstadetal2023} and \citealt{Kleinetal2020} for further discussion of such constructions.} In particular, I obtain confidence sets $\widehat{CS}_{\mathrm{in}}^{sd}$ and $\widehat{CS}_{\mathrm{exc}}^{sd}$, which denote the set of tracts confidently included and excluded in the oracle optimal selection, respectively. With asymptotic probability of at least $1-\alpha$, these confidence sets satisfy:
\begin{equation*}
    \widehat{CS}_{\mathrm{in}}^{sd} \subseteq \left\{
i: \delta_i^* = 1
    \right\} \subseteq \left( \widehat{CS}_{\mathrm{exc}}^{sd}\right)^c ~.
\end{equation*}

Figure~\ref{fig:confidence_sets_inner_outer_pts} plots the inner and outer confidence sets for the top-third Bayes-optimal selection using both the step-down method developed in this paper and the $F$-localization method of \citet{IgnatiadisandWager2021} derived formally in appendix \ref{app:Floc}. Tracts are plotted according to their $\hat\theta_i$ and $\sigma_i$. Tracts that researchers can confidently include in the oracle decision are highlighted in teal, whereas tracts that researchers can confidently exclude from the oracle decision are highlighted in copper. Tracts in neither set, corresponding to those tracts for which a researcher is uncertain about the oracle classification, are highlighted in gray.

\begin{figure}[H]
    
    \begin{subfigure}{0.45\textwidth}
        \centering
        \includegraphics[width=\linewidth]{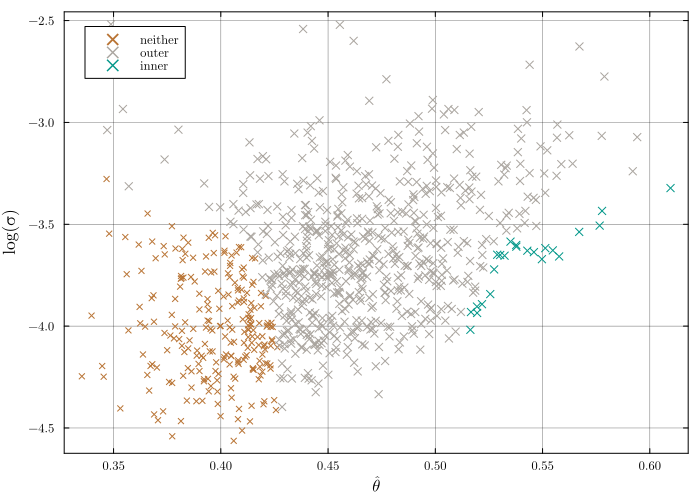}
        \caption{$F$-localization}
    \end{subfigure}
    \hfill
    \begin{subfigure}{0.45\textwidth}
        \centering
        \includegraphics[width=\linewidth]{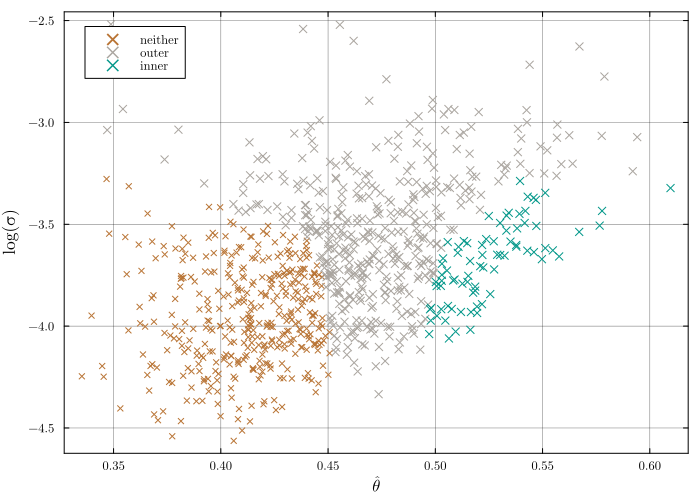}
        \caption{Step-Down}
    \end{subfigure}
    \caption{Inner and outer confidence sets for top-third Bayes-optimal selection, in $(\hat\theta,\log(\sigma))$ space.}
    \label{fig:confidence_sets_inner_outer_pts}
    \end{figure}

Visual inspection of Figure~\ref{fig:confidence_sets_inner_outer_pts} yields two conclusions. First, the step-down procedure is substantially more informative than $F$-localizations. Using the step-down procedure, we can confidently conclude that 80 tracts are selected and 399 tracts are not selected. Relative to the 25 confident inclusions and 241 confident exclusions under $F$-localizations, this represents a $220\%$ increase in confident inclusions and a $66\%$ increase in confident exclusions. Consequently, I conclude that the set of oracle decisions that are consistent with the observed data is relatively narrow. Second, my methods allow researchers to place greater confidence on the classifications of low-noise tracts. This not only matches the intuition around confidence sets for the latent $\theta_i$ themselves, but also provides the key intuition that the data are more informative about the oracle decision for tracts with a higher signal-to-noise ratio. In particular, there is little researchers can conclude about the conditional mean of $\theta_i$ when the $\sigma_i$ are sufficiently large as to ensure that information about the conditional law of $\theta_i$ given $\sigma_i$ is limited. 

In appendix \ref{app:KDE_results_mob}, I compare my step-down methods to alternative implementations of the $F$-localizations of \citealt{IgnatiadisandWager2021}, and show that the performance gains afforded by my methods are robust to choices in implementation of alternative methods. In appendix \ref{app:additional_figs_mob}, I provide additional results from this exercise. I plot in appendix figure \ref{fig:sea_maps} a map of the inner and outer confidence sets for the Bayes-optimal selection using both my methods and using $F$-localizations. My methods are most-informative in densely-populated, urban Seattle tracts. Appendix figure \ref{fig:sea_maps_urban} provides these inner and outer confidence set maps for urban Seattle tracts. Appendix figures \ref{fig:rank_LCBs} and \ref{fig:rank_UCBs} plot lower and upper confidence sets, respectively, for the ranks of Seattle census tracts on a map. Appendix figures \ref{fig:urban_rank_LCBs} and \ref{fig:urban_rank_UCBs} repeat this exercise for urban Seattle tracts.

\subsection*{Welfare implications of a feasible decision rule}
\label{subsec:mobility-welfare}

While the previous exercise allows us to quantify uncertainty around the oracle decision, it does not allow researchers to draw conclusions about the welfare of feasible decision rules. For instance, while the tracts in $\widehat{CS}_{\mathrm{in}}^{sd}$ and $\widehat{CS}_{\mathrm{exc}}^{sd}$ may differ in terms of conditional mean mobility, the gaps in mobility between these confidently included and excluded tracts may be small. In my second exercise, I consider the problem of comparing low and high-mobility tracts as selected by a feasible decision rule that selects tracts according to estimated conditional means. Given this feasible rule, I obtain confidence bounds on the gains in conditional mean economic mobility associated with moving from an arbitrary low-mobility tract to an arbitrary high-mobility tract. For each pair consisting of a feasibly-labeled low-mobility tract $i$ and a feasibly-labeled high-mobility tract $j$, define $\Delta_{ij}(G)
        :=
        \mu_G(\widehat\theta_j,\sigma_j)
        -
        \mu_G(\widehat\theta_i,\sigma_i)$. In the following exercise, I use the step-down confidence bands developed in this paper to obtain simultaneous lower confidence bounds on $\Delta_{ij}(G)$ across all such pairs.

         Figure~\ref{fig:LCB_welfare_gains_diff_selections} presents a histogram of these lower confidence bounds across all pairs of low- and high-mobility tracts. As above, I present results using both the step-down method developed in this paper (plotted in teal) and the $F$-localization method (plotted in copper) of \citet{IgnatiadisandWager2021}. Values lying to the right of zero correspond to pairs for which we can confidently conclude that the conditional mean gain in economic mobility is positive. The larger the area of the histogram to the right of zero, the larger the proportion of tract pairs for which researchers can conclude that there is a positive mobility gain. The vertical dotted lines represent the average lower confidence bound across pairs of low- and high-mobility tracts using each method. The farther to the right these dotted lines, the larger the magnitudes of gains in conditional mean economic mobility associated with moving from low to high-mobility tracts,

\begin{figure}[H]

\centering
        \includegraphics[width=0.5\linewidth]{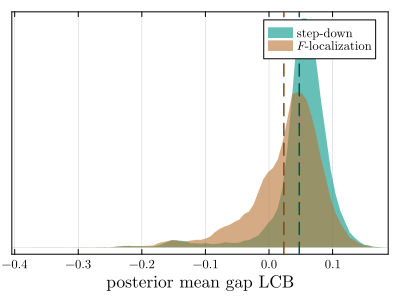}
    \caption{Lower confidence bounds on conditional mean gains in mobility associated with moving from low to high-mobility tracts, as selected by an EB decision rule.}
    \label{fig:LCB_welfare_gains_diff_selections}
    \end{figure}

    A visual inspection of these histograms yields two conclusions. First, for the vast majority of pairs of tracts labeled low and high mobility by the feasible decision maker, an oracle decision-maker would believe that moving from the low- to the high-mobility tract generates a positive gain in conditional mean mobility. The average lower confidence bound corresponds to a five percentile-point improvement in economic mobility. The step-down method yields positive lower confidence bounds for $91.9\%$ of these tract pairs. Second, the methods developed in this paper are suggestive of substantially larger gains across pairs of tracts than $F$-localizations. $F$-localizations analogously conclude that only $73.7\%$ of tract pairs lead to a positive gain in conditional mean economic mobility.

\subsection*{Implications for the classification of high-opportunity neighborhoods}
\label{subsec:mobility-implications}

Given the uncertainty around the distribution of economic mobility across Seattle tracts, one may naturally ask whether an oracle decision-maker would believe that the CMTO program actually directs families to high-mobility neighborhoods. Assuming that the correlational estimates of tract-level mobility used in my exercises are causal, would an oracle DM ascribe positive conditional mean welfare gains to the CMTO program?

First, as noted in \citet{Mogstadetal2023}, the answer to this question depends on which assumptions the researcher is willing to make on movement patterns between low- and high-mobility neighborhoods. In this vein, notice that in the exercises undertaken in this paper, we cannot reject the possibility that a small fraction of high-noise tracts selected as high mobility by a feasible decision rule are among the lowest-mobility tracts in Seattle according to oracle conditional mean mobility. Consequently, under an extreme thought experiment wherein all treated individuals move to these tracts, we cannot reject the possibility that the oracle decision-maker would conclude that the CMTO program is welfare decreasing. However, the set of tracts for which this is true is a small fraction of the feasibly-selected tracts. In particular, for the majority of feasibly-selected tracts, my exercises suggest that researchers can be confident that the oracle decision-maker would place these tracts among the top half, and in many cases among the top third, according to conditional mean economic mobility.

However, these extreme cases are instructive in one sense that is informative for decision-making. In particular, there is little one can say about the oracle conditional mean mobility at tracts where economic mobility is measured with high noise. In these low signal-to-noise regimes, one can learn little about the distribution of economic mobility relevant for the conditional mean, and as a result one can say little about the optimal decision. As a result, under the rich forms of heteroskedasticity considered in this paper, it is unclear how researchers should navigate the planner's dilemma. In particular, researchers cannot be confident in feasible solutions to the problem of trading off selection of high-mobility, high-noise tracts with selection of lower-mobility, lower-noise tracts.

Given the above conclusions on the informativeness of the economic mobility estimates for optimal decision-making, a natural related question asks how well an oracle decision maker would believe feasible rules perform according to their objectives. In particular, while an oracle decision maker can conclude that some tracts are higher mobility than others according to their conditional means, the magnitudes of the conditional mean gains between low- and high-mobility tracts could nevertheless be small. In particular, this would be the case when the distribution of mobility is fairly concentrated, meaning that tracts are qualitatively indistinguishable according to true economic mobility. In the setting of the \citet{Bergmanetal2023} experiment, however, the results above show that this is not the case. Indeed, for the vast majority of pairs of low- and high-mobility tracts, the gains in economic mobility expected by an oracle decision maker are economically large.

The exercises considered in this application suggest two results. First, researchers can be confident that an oracle decision-maker would rank the majority of feasibly selected tracts in the top half, and many in the top third, according to conditional mean mobility. As a result, the movement patterns needed to generate losses in conditional mean economic mobility due to the CMTO experiment are necessarily fairly extreme. Moreover, in comparing low- to high-mobility tracts as selected by a feasible decision maker who approximates the oracle, my exercises suggest that an oracle decision maker would ascribe qualitatively large gains in economic mobility associated with moving from low- to high-mobility neighborhoods. This also suggests wide dispersion in economic mobility across neighborhoods. That said, I cannot reject the possibility that an oracle decision maker makes mistakes according to true mobility at each neighborhood, in line with results from existing work.

\subsection*{Comparison to Analyses of Latent Effects}

Given the above discussion, it is worth comparing the exercises considered above to existing exercises conducting inference on the latent effects $\theta_i$ themselves. While the exercises considered in this paper consider the problem of an oracle decision maker with knowledge of the distribution of mobility across tracts, existing exercises effectively compare to analyzing the problem of a fully-informed oracle decision maker with knowledge of true mobility at each tract. In this sense, existing exercises admit similar economic interpretation to those considered in my application to the \citealt{Bergmanetal2023} selection exercise.

In particular, observe that my first exercise -- obtaining inner and outer confidence sets for the oracle optimal selection -- is analogous to exercises considered in \citealt{Mogstadetal2023}, who consider the problem of obtaining confidence sets for the top-third selection of tracts according to the true $\theta_i$. This latter selection would indeed be the optimal selection for an oracle decision-maker with knowledge of the $\theta_i$. Similarly, my second exercise -- obtaining confidence bounds on the gains in conditional mean mobility between low and high-mobility tracts -- is analogous to exercises considered in \citealt{petrouzeniou_and_shaikh_2025}, who consider the problem of obtaining confidence bounds on the gains in true mobility between low and high-mobility tracts. Again, these latter quantities correspond exactly to the gains in welfare ascribed to moving from low to high-mobility tracts by a fully-informed oracle. Consequently, the exercises considered in this paper admit similar economic interpretations to existing work.

Fundamentally, interest across the exercises considered in this paper and in existing papers is in evaluating and characterizing mistakes relative to an oracle decision maker, as discussed in \citealt{AndrewsandChen2025}. For instance, my first exercise, and consequently the exercises of \citealt{Mogstadetal2023}, allow researchers to evaluate which selections would not be classified as mistakes by an oracle decision maker. My second exercise, and thus the exercises of \citealt{petrouzeniou_and_shaikh_2025}, allow researchers to evaluate the magnitudes of these mistakes in terms of welfare. Therefore, the exercises considered in this paper, and those considered in previous work, differ in the sense of the oracle decision maker defining these mistakes.

The choice of oracle considered in this paper is natural given existing work. \citealt{chen2024} explicitly considers the problem of approximating an oracle decision maker with knowledge of the effect distribution $G$. In the original decision problem of \citealt{Bergmanetal2023}, selections were made using shrunk estimates of economic mobility, thus estimating conditional mean mobility and implicitly approximating the same oracle decision maker with knowledge of the effect distribution $G$. Given the theoretical and applied interest, therefore, in evaluating and approximating the decisions of an oracle decision maker with knowledge of the distribution of mobility across tracts, I view the exercises considered in this paper as complementary to existing exercises considering the fully-informed oracle.

Results due to \citealt{Mogstadetal2023} and \citealt{petrouzeniou_and_shaikh_2025} suggest that evaluating the fully-informed oracle decision maker is challenging in the setting of \citealt{Bergmanetal2023}. Economic mobility estimates are too numerous and too noisy to draw informative conclusions on the problem of a fully-informed oracle decision maker. My exercises complement this existing work by considering an alternative oracle decision-maker, and yield qualitatively different conclusions. I obtain positive results, and can draw informative conclusions on the problem of the oracle decision considered in this paper. Of course, my analyses do not rule out the possibility that an oracle with knowledge of $G$ makes mistakes relative to a fully-informed oracle. In fact, existing work suggests that this possibility cannot be rejected. Nevertheless, the framework developed in this paper, alongside the methodological advances developed herein, allow for new, economically substantive insights on place-based policy targeting using local statistics on mobility.

\section{Conclusion}
\label{sec:conclude}

This paper develops a framework and methods for evaluating uncertainty in a broad range of compound decision problems. I show that such compound decisions are ubiquitous in applied work in economics, and define a notion of an oracle decision maker in these problems. I argue that this oracle decision maker's problem is subject to substantial uncertainty in finite-samples. I provide tools to evaluate this uncertainty. Applying my tools to a canonical compound decision problem in economics -- the problem of differentiating low and high-mobility areas -- shows that local mobility statistics may be more informative for place-based policy targeting than previously understood.

\begin{appendix}

\bibliography{refs}

\clearpage
\startcontents[app]

\section*{Appendix Contents}

\printcontents[app]{}{1}{}

\section{Proofs of Results from Main Text}
\label{app:proofs_of_main_results}

\begin{proof}
        \textbf{Proof of proposition \ref{prop:step_down_finite_sample}:} The argument follows mutatis mutandis from the argument used to prove proposition \ref{prop:step_down_validity_est_cv}, and is thus omitted.
    \end{proof}
\begin{proof}
        \textbf{Proof of proposition \ref{prop:step_down_validity_est_cv}:} First, note that it suffices to show:
        \begin{equation}
            \label{eq:equivalent_coverage_infeasible_alg}
    \limsup_{n\to\infty} \sup_{G\in\calG} P_G\left(
\{\exists \gamma \in \Gamma(G):\gamma\not\in \Gamma_{\hat{k}}\}
    \right) =  \limsup_{n\to\infty} \sup_{G\in\calG} P_G\left(
\Gamma(G) \not\subseteq \Gamma_{\hat k}
    \right) \leq \alpha~.
        \end{equation}
        The event in \eqref{eq:equivalent_coverage_infeasible_alg} corresponds to a false discovery under familywise error rate control.
        Let $\hat\jmath$ denote the first index such that there exists $\gamma$ in $\Gamma(G)$ such that $\gamma$ does not lie in $\Gamma_{\hat\jmath}$. Notice that $\hat\jmath \ge 1$, and that $\Gamma(G)$ is contained in $\Gamma_{\hat\jmath -1}$ by construction. Furthermore, notice that $\Gamma(G) \not\subseteq \Gamma_{\hat{k}}$ if and only if $\hat\jmath \leq \hat k$. It follows that, whenever $\Gamma(G) \not\subseteq \Gamma_{\hat{k}}$, there exist $\gamma \in \Gamma(G)$ such that:
        \begin{equation}
        \label{eq:false_discovery_equiv_event}
            |\hat{L}_{n,\gamma} |> \bar B_\gamma +  \hat{c}_n(1-\alpha,\Gamma_{\hat\jmath - 1}) \frac{\hat\sigma_n(Q_\gamma)}{\sqrt{n}}~.
        \end{equation}
        Notice, however, that part (ii) of assumption \ref{assu:cv_validity_monotonicity} implies:
        \begin{equation*}
            |\hat{L}_{n,\gamma} |> \bar B_\gamma +  \hat{c}_n(1-\alpha,\Gamma(G)) \frac{\hat\sigma_n(Q_\gamma)}{\sqrt{n}}~.
        \end{equation*}
        As a result, 
        \begin{equation*}
            \sup_{\gamma \in \Gamma(G)}
        \left|
        \frac{
                    \sqrt n\left(
                        \widehat L_{n,\gamma}
                        -
                        \mathbb E_G\widehat L_{n,\gamma}
                    \right)
                }{
                    \hat\sigma_n(Q_\gamma)
                }
        \right|> \hat{c}_n(1-\alpha,\Gamma(G))
        \end{equation*}
        is implied by the event \eqref{eq:false_discovery_equiv_event}, since $\bar{B}_\gamma \ge |B_\gamma(G)|$. Applying part (i) of assumption \ref{assu:cv_validity_monotonicity} yields:
        \begin{align*}
            &\limsup_{n\to\infty} \sup_{G\in\calG} P_G\left(
\Gamma(G) \not\subseteq \Gamma_{\hat{k}}
    \right)\\
    & \le \limsup_{n\to\infty}\sup_{G\in\calG} P_G\left(
\sup_{\gamma\in\Gamma(G)}
\left|
\frac{
                    \sqrt n\left(
                        \widehat L_{n,\gamma}
                        -
                        \mathbb E_G\widehat L_{n,\gamma}
                    \right)
                }{
                    \hat\sigma_n(Q_\gamma)
                }
\right|> \hat{c}_n(1-\alpha,\Gamma(G))
    \right)\le \alpha~,
        \end{align*}
        thus proving the proposition.    
    \end{proof}

\section{\texorpdfstring{$F$}{F}-localizations}
\label{app:Floc}

$F$-localizations constitute an alternative approach to inference advanced in the literature by \citealt{IgnatiadisandWager2021}. I generalize these $F$-localizations to the heteroskedastic Gaussian sequence model \eqref{eq:EB_DGP}. The construction of $F$-localization confidence sets proceeds in three steps. First, I construct a confidence set for the marginal, joint cumulative distribution function of the $(\hat\theta_i,\sigma_i)$, which is a linear functional of the latent distribution $G$. Using this confidence set for the joint CDF of the $(\hat\theta_i,\sigma_i)$, one can write obtain a confidence set for the latent distribution $G$. Finally, to obtain a confidence set for a conditional mean estimand $M_b(t,s,G)$ involves computing the smallest and largest values of the conditional mean over the confidence set for the latent distribution $G$. 

I obtain the joint cumulative distribution function:
\begin{equation*}
    F_G(t,s) = \int_{u\le t}\int_{v\le s} f_G(u,v)dvdu~,\qquad f_G(t,s) := \int \frac{1}{s}\varphi\left(
    \frac{t-\theta}{s}
    \right)g_G(s,\theta)d\theta~.
\end{equation*}
Crucially, $F(t,s,G)$ is linear in the distribution $G$ for any $t\in\mathbb{R}$ and any $s\ge 0$. The $F$-localization approach to inference proceeds as follows. First, one obtains a level $1-\alpha$ confidence set for the marginal joint CDF $F(\cdot,\cdot)$, which I denote by $\mathcal{F}_n(1-\alpha)$. From this confidence set, one can obtain a confidence set for $G$ as follows:
\begin{equation*}
    \calG_n(1-\alpha) := \left\{
G\in\calG:F_G(\cdot,\cdot)\in \mathcal{F}_n(1-\alpha)
    \right\}~.
\end{equation*}
One can project this confidence set to the conditional means of interest to obtain:
\begin{align*}
        \underline{\mu}(t,s,1-\alpha) &:= \inf_{G\in\calG_n(1-\alpha)} \mu_G(t,s)\\ \overline{\mu}(t,s,1-\alpha) &:= \sup_{G\in\calG_n(1-\alpha)} \mu_G(t,s)
    \end{align*}
Defining $\widehat{CS}^{floc}(t,s):= [\underline{\mu}(t,s,1-\alpha), \overline{\mu}(t,s,1-\alpha)]$, I obtain the following proposition:
\begin{prop}
\label{prop:F_loc_validity}
    Let $\mathcal{F}_n(1-\alpha)$ satisfy:
    \begin{equation}
    \label{eq:suff_cond_for_floc_validity}
        \liminf_{n\to\infty}\inf_{G\in\calG} P_G\left(
F_G(\cdot,\cdot) \in \mathcal{F}_n(1-\alpha)
        \right)\ge 1-\alpha~.
    \end{equation}
    Then $\widehat{CS}^{floc}(t,s)$ satisfies \eqref{eq:uniform_coverage} and:
    \begin{equation}
    \label{eq:floc_coverage_of_G}
    \liminf_{n\to\infty}\inf_{G\in\calG} P_G(G\in\calG_n(1-\alpha))\geq 1-\alpha~.
    \end{equation}
\end{prop}
Applying the multivariate DKW inequality due to \citealt{Naaman2021} provides one construction of $\mathcal{F}_n(1-\alpha)$ satisfying \eqref{eq:suff_cond_for_floc_validity}. Allowing $\hat{F}_n$ to be the empirical joint CDF of the $\hat\theta_i,\sigma_i$, I define:
\begin{equation*}
    \mathcal{F}_n^{DKW}(1-\alpha) := \left\{
F: \sup_{t\in\mathbb{R}, s\geq 0} |F(t,s) - \hat{F}_n(t,s)| \le \sqrt\frac{
\log(2(n+1)/\alpha)
}{2n}
    \right\}~.
\end{equation*}
This construction $\mathcal{F}_n^{DKW}(1-\alpha)$ is particularly convenient, in that $\calG_n^{DKW}(1-\alpha)$ is defined by a set of linear constraints on $G$. As a result, $\widehat{CS}^{floc}(t,s)$ can be computed via linear fractional programming, as discussed at depth in \cite{IgnatiadisandWager2021}.
\begin{rem}
\label{rem:choice_of_floc}
    By imposing restrictions on $\calG$, further choices of $\mathcal{F}_n(1-\alpha)$ may exist. For example, \citealt{IgnatiadisandWager2021} provide $\mathcal{F}_n(1-\alpha)$ based on a kernel density estimator of the density $f(\cdot,\cdot,G)$, for the homoskedastic gaussian noise model. \citealt{IgnatiadisandWager2021} show that the marginal density of the $\hat\theta$ can be easily estimated using a particular choice of kernel. In particular, the smoothness of mixtures of Gaussians allows for a choice of a shrinking bandwidth such that bias vanishes asymptotically. In the heteroskedastic case, the arguments of \citealt{IgnatiadisandWager2021} do not easily generalize, apart from under a few strong, sufficient conditions.\footnote{In particular, the usual arguments bounding bias via exponential tail bounds on the Fourier transform of $f(\cdot,\cdot,G)$ do not easily generalize. Imposing precision independence, $\theta_i \indep \sigma_i$, one can recover extensions of the results due to \citealt{IgnatiadisandWager2021} by imposing that the marginal distribution of $\sigma$ have support bounded away from zero and be supersmooth of order $2$, in the sense of \citealt{Fan1991}.} Nevertheless, I apply this KDE $F$-localization in appendix \ref{app:KDE_results_mob}, and show that these $F$-localizations are conservative relative to my methods.
\end{rem}

Finally, I notice that $F$-localization confidence sets are fully simultaneous, in that they satisfy simultaneous coverage of all functionals of $G$ via \eqref{eq:floc_coverage_of_G}. However, in the decision-making problems of interest in this paper, it is desirable to focus power on collections of functionals of interest, such as $\mu_G(\cdot,\cdot)$, motivating the methods developed in this paper.

\section{Details on Affine Minimax Estimation Under Heteroskedasticity}
    \label{app:minimax_details}

    In this appendix, I derive a class of affine minimax estimators that allow for (i) rich forms of dependence between $\theta_i$ and $\sigma_i$, and (ii) optimally trade off bias and variance. The remainder of this section is devoted to obtaining these affine minimax estimators, which will constitute my choice of test-statistic for the procedure derived above. The construction builds on recent formalizations of a heuristic due to Charles Stein (see e.g. \citealt{Donoho1994}, \citealt{ArmstrongKolesar2018}, \citealt{IgnatiadisandWager2021}), which takes a hardest one-dimensional subproblem, and solves for a minimax estimator within the subproblem. 

Formally, I restrict attention to estimators $\hat{L}_{n,\gamma}$ of $L_\gamma(G)$ that can be written as sample means, for some function $Q_\gamma:\mathbb{R}^2\to\mathbb{R}$:
\begin{equation}
\label{eq:estimator_Q_sum_form}
    \hat{L}_{n,\gamma} := \frac{1}{n}\sum_{i=1}^n Q_\gamma(\hat\theta_i,\sigma_i)~.
\end{equation}
Our goal is to choose some $Q_\gamma$ that optimally trades off bias and variance. In service of this objective, I define bias and variance as follows:
\begin{gather*}
    B_\gamma(G) := \int Q_\gamma(t,s)f_G(t,s)d(t,s) - L_\gamma(G)~;\\
    \widehat{\var}(Q_\gamma) := \int Q_\gamma(t,s)^2 \hat{f}_n(t,s)d(t,s) - \left(
\int Q_\gamma(t,s) \hat{f}_n(t,s)d(t,s)
    \right)^2~,
\end{gather*}
where $\hat{f}_n$ is the maximum likelihood estimator of $f_G(\cdot,\cdot)$ over $\calG$:\footnote{As noted by \citealt{chen2024}, obtaining this nonparametric maximum likelihood estimate is in general difficult, since researchers typically observe one observation of $\hat\theta_i$ per observed $\sigma_i$. To correct for this \citealt{chen2024} provides a flexible class of joint distributions over $(\theta, \sigma)$ for which nonparametric maximum likelihood can be applied, under a particular transformation. Similarly, in my setting, the additional structure needed to compute the nonparametric maximum likelihood is imposed by $\calG$. I discuss choices of $\calG$ in my application.}
\begin{equation*}
    \hat{f}_n := f_{\hat{G}_n}(\cdot,\cdot), \text{ s.t. }\hat{G}_n := \argmax_{G'\in\calG} \sum_{i=1}^n \log(f_{G'}(\hat\theta_i, \sigma_i))~.
\end{equation*}
For any $\delta$, I consider $Q_\gamma$ which solve the following class of programs indexed by $\gamma$:
\begin{equation}
\label{eq:minimax_program}
    \min_{Q_\gamma}\max_{G\in\calG} B_\gamma(G)^2~\text{s.t.}~\frac{1}{n}\widehat{\var}(Q_\gamma) \leq V_n~,
\end{equation}
where $V_n$ is a tolerance parameter that can be chosen to minimize, for example, the mean squared error of $\hat{L}_{n,\gamma}$. In the spirit of arguments due to \citealt{Donoho1994}, \citealt{ArmstrongKolesar2018}, and \citealt{IgnatiadisandWager2021}, I demonstrate formally in appendix \ref{app:minimax_details} that it suffices to solve the following modulus problem for $\delta > 0$:
\begin{equation}
\label{eq:modulus_problem}
    \omega_{n,\gamma}(\delta) := \sup_{G_1,G_{-1}\in\calG}\left\{
L_\gamma(G_1)-L_\gamma(G_{-1}) : \int \frac{\left(f_{G_1}(t,s) - f_{G_{-1}}(t,s)\right)^2}{\hat{f}_n(t,s)}d(t,s) \leq \frac{\delta^2}{n}
    \right\}~.
\end{equation}
In particular, let solutions $G_1^\delta$ and $G_{-1}^\delta$ to the above problem, and let $G_0^\delta = \frac{1}{2}(G_1^\delta + G_{-1}^\delta)$. For $V_n = \omega_{n,\gamma}'(\delta)^2$, the following choice of $Q_\gamma(\cdot,\cdot)$ solves \eqref{eq:minimax_program}:
\begin{equation*}
    \frac{n\omega_{n,\gamma}'(\delta)}{\delta}\left\{
\frac{f_{G_1^\delta}(\cdot,\cdot) - f_{G_{-1}^\delta}(\cdot,\cdot)}{\hat{f}_n(\cdot,\cdot)}-\int\frac{\left(f_{G_1^\delta}(t,s) - f_{G_{-1}^\delta}(t,s) \right)f_{G_0^\delta}(t,s)}{\hat{f}_n(t,s)} d(t,s)
    \right\} + L_\gamma(G_0^\delta)~.
\end{equation*}
Importantly given the step-down procedure described above, lemmas \ref{lem:max_bias_attained} and \ref{lem:final_minimaxity} provide a closed form for the worst-case bias of the affine minimax $\hat{L}_{n,\gamma}$ constructed as above:
\begin{equation*}
    \sup_{G\in\calG} |B_\gamma(G)| = \frac{1}{2}(\omega_{n,\gamma}(\delta) - \delta\omega_{n,\gamma}'(\delta))~.
\end{equation*}
Finally, note that in what follows, we will define and estimate the variances of $Q_\gamma$ as follows:
\begin{gather*}  
    \var_G(Q_\gamma) = \int Q_\gamma(t,s)^2 f_G(t,s)d(t,s) - \left(
\int Q_\gamma(t,s) f_G(t,s)d(t,s)
    \right)^2~;\\
    \sigma_G(Q_\gamma) = \sqrt{\var_G(Q_\gamma)}~;\\
    \widehat{\var}(Q_\gamma) = n\omega_{n,\gamma}'(\delta)^2~;\\
    \hat\sigma_n(Q_\gamma) = \sqrt{\widehat{\var}(Q_\gamma)}~.
\end{gather*}
\begin{rem}
    \citealt{IgnatiadisandWager2021} use these affine minimax estimators to obtain marginal confidence sets for individual conditional mean estimands, satisfying \eqref{eq:marginal_coverage}. As discussed in section \ref{sec:setup}, this notion of coverage is insufficient to quantify uncertainty around Bayes-optimal decisions themselves. However, because the criterion functions $L_\gamma(G)$ are linear in $G$, the same affine minimax estimators derived in \citealt{IgnatiadisandWager2021} are natural choices of test statistics for the problem \eqref{eq:uniform_coverage}. Finally, I emphasize that my analysis of affine minimax estimation allows for rich forms of heteroskedasticity as in \eqref{eq:EB_DGP}, in contrast to the work of \citealt{IgnatiadisandWager2021}. 
\end{rem}
\begin{rem}
\label{rem:delta_selection}
    The parameter $\delta$ is the key object via which researchers can navigate the bias-variance tradeoff. While for marginal confidence sets, researchers can easily choose $\delta$ to minimize confidence interval length, this exercise is more challenging when targeting simultaneous coverage. Because critical values depend on the distribution of collections of $\hat{L}_{n,\gamma}$ over subsets of $\Gamma$, which I will denote by $\Tilde{\Gamma}$, $\delta$ should be chosen simultaneously for each $\gamma$ over $\Tilde{\Gamma}$, presenting a challenging optimization problem. Moreover, since our step-down procedure obtains critical values for different choices of $\Tilde{\Gamma}$ over each iteration in the procedure, allowing for reoptimization of $\delta$ at each iteration can violate the monotonicity of critical values that underlies the validity result in proposition \ref{prop:step_down_finite_sample}. As a result, for some $\Delta\subset \mathbb{R}_+$ bounded away from zero, I advocate choosing $\delta$ that minimizes mean squared error, given by $\left(\frac{1}{2}(\omega_{n,\gamma}(\delta)-\delta\omega_{n,\gamma}'(\delta))\right)^2 + \omega_{n,\gamma}'(\delta)^2$, in line with the choice of $\delta$ in \citealt{IgnatiadisandWager2021}. 
\end{rem}

    The remainder of this appendix concerns proving the minimaxity of estimators $Q_\gamma(\cdot,\cdot)$ defined as follows:
\begin{equation}
\label{eq:opt_Q}
    \frac{n\omega_{n,\gamma}'(\delta)}{\delta}\left\{
\frac{f_{G_1^\delta}(\cdot,\cdot) - f_{G_{-1}^\delta}(\cdot,\cdot)}{\hat{f}_n(\cdot,\cdot)}-\int\frac{\left(f_{G_1^\delta}(t,s) - f_{G_{-1}^\delta}(t,s) \right)f_{G_0^\delta}(t,s)}{\hat{f}_n(t,s)} d(t,s)
    \right\} + L_\gamma(G_0^\delta)~,
\end{equation}
    given solutions $G_1^\delta$ and $G_{-1}^\delta$ to the modulus problem \eqref{eq:modulus_problem}. In particular, I demonstrate that such $Q_\gamma(\cdot,\cdot)$ solve \eqref{eq:minimax_program}. 

    \subsection{Supplemental results and proofs}

    \begin{prop}
    \label{prop:minimaxity_main_result}
        Let $\calG$ be convex, and let $\sup_{G\in\calG} |L_\gamma(G)| < \infty$. Moreover, let $\hat{L}_{n,\gamma}$ be given as above. $\hat{L}_{n,\gamma}$ solves \eqref{eq:minimax_program} subject to $V_n = \omega_{n,\gamma}'(\delta)^2$.
    \end{prop}

    In service of proving proposition \ref{prop:minimaxity_main_result}, I provide the following lemmas. My proofs largely follow well-known arguments, and closely follow the structure of the proof of proposition S15 of \citealt{IgnatiadisandWager2021}.

    \begin{lem}
        \label{lem:modulus_properties}

        Let $\calG$ be convex, let $\inf_{t\in\mathbb{R},s\ge 0} \hat{f}_n(t,s) > 0$,\footnote{In particular, we consider $\hat{f}_n$ on a truncated domain, rather than on $\bbR\times\bbR_+$} and let $\sup_{G\in\calG} |L_\gamma(G)| < \infty$. Then $\omega_{n,\gamma}(\cdot)$ is non-decreasing, bounded and non-negative, concave, and satisfies $\omega_{n,\gamma}(\delta) \ge \delta\omega_{n,\gamma}'(\delta)$ and $\delta \in \argmax_{\Tilde{\delta}} \omega_{n,\gamma}(\Tilde{\delta}) - \omega_{n,\gamma}'(\delta)\Tilde{\delta}$.
    \end{lem}
    The proof of this lemma is identical to the proof of proposition S14 in \citealt{IgnatiadisandWager2021}, and is thus omitted.

    \begin{lem}
        \label{lem:max_bias_attained}
        Let the conditions of lemma \ref{lem:modulus_properties} hold, and let $\hat{L}_{n,\gamma}$ be given by \eqref{eq:estimator_Q_sum_form}, with $Q_\gamma$ given by \eqref{eq:opt_Q}. $\hat{L}_{n,\gamma}$ satisfies:
        \begin{align*}
    \sup_{G\in\calG} B_\gamma(G) &= B_\gamma(G_{-1}^\delta)\\
& = -B_\gamma(G_1^\delta) = -\inf_{G\in\calG} B_\gamma(G)~.
\end{align*}
In addition, $\sup_{G\in\calG} B_\gamma(G) = \frac{1}{2}(\omega_{n,\gamma}(\delta) - \delta\omega_{n,\gamma}'(\delta))$
    \end{lem}
\begin{proof}
First, note that:
    \begin{align*}
    & B_\gamma(G_1^\delta) \\
    & = \int Q_\gamma(t,s)f_{G_1^\delta}(t,s) d(t,s) - L_\gamma(G_1^\delta)\\
    & = -\frac{1}{2}(L_\gamma(G_1^\delta)-L_\gamma(G_{-1}^\delta))\\
    & + \frac{n\omega_{n,\gamma}'(\delta)}{\delta}\left\{\int
\frac{f_{G_1^\delta}(t,s) - f_{G_{-1}^\delta}(t,s)}{\hat{f}_n(t,s)} f_{G_1^\delta}(t,s)d(t,s) - \int \frac{f_{G_1^\delta}(t,s) - f_{G_{-1}^\delta}(t,s)}{\hat{f}_n(t,s)} f_{G_0^\delta}(t,s)d(t,s)
    \right\}\\
    & = -\frac{1}{2}\omega_{n,\gamma}(\delta) + \frac{\omega_{n,\gamma}'(\delta)}{2\delta}\delta^2 = -\frac{1}{2}(\omega_{n,\gamma}(\delta)-\delta\omega_{n,\gamma}'(\delta))~.
    \end{align*}
    The first equality holds by the definition of $B_\gamma(G_1^\delta)$. The second equality holds by plugging in the expression for $Q_\gamma(\cdot,\cdot)$. The third equality holds by noticing that $\omega_{n,\gamma}(\delta) = L_\gamma(G_1^\delta)-L_\gamma(G_{-1}^\delta)$, and by applying the linearity of $f_G(t,s)$ in $G$ to rewrite the term in the fourth line of the above in terms of the constraint in \eqref{eq:modulus_problem}. Similar reasoning yields $B_\gamma(G_{-1}^\delta) = \frac{1}{2}(\omega_{n,\gamma}(\delta)-\delta\omega_{n,\gamma}'(\delta))$. Now, I define the following for arbitrary $G\in\calG$:
    \begin{align*}
        \Delta(\lambda) &= \left(
n\int \frac{\left(f_{G_1^\delta}(t,s)- ((1-\lambda)f_{G_{-1}^\delta}(t,s) + \lambda f_G(t,s))\right)^2}{\hat{f}_n(t,s)}d(t,s)
        \right)^{1/2}\\
        J(\lambda) &=L_\gamma(G_1^\delta)-L_\gamma((1-\lambda)G_{-1}^\delta+\lambda G) - \omega_{n,\gamma}'(\delta)\Delta(\lambda)
    \end{align*}
    For any $\lambda\in [0,1]$:
    \begin{align*}
        J(\lambda) &\leq \omega_{n,\gamma}(\Delta(\lambda))-\omega_{n,\gamma}'(\delta)\Delta(\lambda)\\
        & \leq \sup_{\Tilde{\delta}\geq 0}\left\{
        \omega_{n,\gamma}(\Tilde{\delta})-\omega_{n,\gamma}'(\delta)\Tilde{\delta}
        \right\}\\
        & \leq \omega_{n,\gamma}(\delta)-\omega_{n,\gamma}'(\delta)\delta\\
        & = J(0)
    \end{align*}
    The first inequality holds by noticing that $L_\gamma(G_1^\delta)-L_\gamma((1-\lambda)G_{-1}^\delta+\lambda G) \le \omega_{n,\gamma}(\Delta(\lambda))$. The third inequality holds by the concavity of $\omega_{n,\gamma}$, established by lemma \ref{lem:modulus_properties}. Applying the chain rule and dominated convergence theorem gives:\footnote{The dominating function, in this case, would be the maximum of $(f_{G_1^\delta}(t,s)-f_G(t,s))^2/\hat{f}_n(t,s)$ and $(f_{G_1^\delta}(t,s) - f_{G_{-1}^\delta}(t,s))^2/ \hat{f}_n(t,s) $. The maximum of two integrable functions is itself integrable.}
    \begin{equation*}
        J'(0) = L_\gamma(G_{-1}^\delta) - L_\gamma(G) +\frac{n\omega_{n,\gamma}'(\delta)}{\delta}\int\left(
        \frac{f_{G_1^\delta}(t,s) - f_{G_{-1}^\delta}(t,s)}{\hat{f}_n(t,s)}(f_G(t,s)-f_{G_{-1}^\delta}(t,s))
        \right)d(t,s)
    \end{equation*}
    Note that, since for any $\lambda\in [0,1]$, $J(\lambda) \le J(0)$, it holds that $J'(0)\le 0$. Now, note that:
    \begin{align*}
        & B_\gamma(G) - B_\gamma(G_{-1}^\delta)\\
        & = \left(\int Q_\gamma(t,s)f_G(t,s)d(t,s) - L_\gamma(G)\right) - \left(\int Q_\gamma(t,s)f_{G_{-1}^\delta}(t,s)d(t,s) - L_\gamma(G_{-1}^\delta)\right)\\
        & = L_\gamma(G_{-1}^\delta)-L_\gamma(G) + \frac{n\omega_{n,\gamma}'(\delta)}{\delta} \int \left(
\frac{f_{G_1^\delta}(t,s)-f_{G_{-1}^\delta}(t,s)}{\hat{f}_n(t,s)}(f_G(t,s)-f_{G_{-1}^\delta}(t,s))d(t,s)
        \right)\\
        & = J'(0) \leq 0
    \end{align*}
   Thus, the claim holds. 
\end{proof}
\begin{lem}
\label{lem:final_minimaxity}
    Let the conditions of lemma \ref{lem:modulus_properties} hold. Recall that $V_n = \widehat{\var}(Q_\gamma)/n$. Allowing $\widetilde{Q}_\gamma$ to be such that, for $\widetilde{L}_{n,\gamma}:=\frac{1}{n}\sum_{i=1}^n\widetilde{Q}_\gamma(\hat\theta_i,\sigma_i)$ and $\widetilde{B}_\gamma(G):=\bbE_G\widetilde{L}_{n,\gamma}-L_\gamma(G)$, $\widehat{\var}(\widetilde{Q}_\gamma)/n \leq V_n$, it holds that:
    \begin{equation*}
    \sup_{G\in\calG} \widetilde{B}_\gamma(G)^2  \geq \sup_{G\in\calG} B_\gamma(G)^2~.
\end{equation*}
In addition, $\widehat{\var}(Q_\gamma)/n = \omega_{n,\gamma}'(\delta)^2$. 
\end{lem}
\begin{proof}
    First, I will write $V_n$ in terms of $\omega_{n,\gamma}'(\delta)$. Note that, when computing $\widehat{\var}$, it is without loss of generality to consider $Q_\gamma(\cdot,\cdot)$ neglecting constants:
    \begin{equation*}
        \frac{n\omega_{n,\gamma}'(\delta)}{\delta}\left\{
        \frac{f_{G_1^\delta}(\cdot,\cdot) - f_{G_{-1}^\delta}(\cdot,\cdot)}{\hat{f}_n(\cdot,\cdot)}
        \right\}
    \end{equation*}  
    Applying linearity of integration, it is clear that for $Q_\gamma$ of this form, $\int Q_\gamma(t,s)\hat{f}_n(t,s)d(t,s)= 0$. Consequently, applying the definition of the modulus constraint in \eqref{eq:modulus_problem}, it follows that $\widehat{\var}(Q_\gamma) = n\omega_{n,\gamma}'(\delta)^2$, giving $V_n = \omega_{n,\gamma}'(\delta)^2$. 

    Any other $\widetilde{Q}_\gamma(\cdot)$ can be rewritten as $\widetilde{Q}_\gamma(\cdot) = \widetilde{Q}_{\gamma,0} + \widetilde{Q}_{\gamma,1}(\cdot)$, for $\widetilde{Q}_{\gamma,0}$ a constant and for $\widetilde{Q}_{\gamma,1}$ such that $\int \widetilde{Q}_{\gamma,1}(t,s)\hat{f}_n(t,s)d(t,s) = 0$. Thus:
    \begin{equation}
    \label{eq:var_Q_tilde_bound}
        \widehat{\var}(\widetilde{Q}_\gamma) = \int \widetilde{Q}_{\gamma,1}(t,s)^2\hat{f}_n(t,s)d(t,s) \leq nV_n = n\omega_{n,\gamma}'(\delta)^2
    \end{equation}
    Finally, I repeat the calculations of \citealt{IgnatiadisandWager2021}:
    \begin{align*}
        & \widetilde{B}_\gamma(G_{-1}^\delta) - \widetilde{B}_\gamma(G_1^\delta)\\
        & = \left(
        \int \widetilde{Q}_{\gamma,1}(t,s)f_{G_{-1}^\delta}(t,s)d(t,s) - L_\gamma(G_{-1}^\delta)
        \right) - \left(
        \int \widetilde{Q}_{\gamma,1}(t,s)f_{G_1^\delta}(t,s)d(t,s) - L_\gamma(G_{1}^\delta)
        \right)\\
        & = L_\gamma(G_1^\delta) - L_\gamma(G_{-1}^\delta) + \int \widetilde{Q}_{\gamma,1}(t,s)(f_{G_{-1}^\delta}(t,s)- f_{G_1^\delta}(t,s))d(t,s)\\
        & = \omega_{n,\gamma}(\delta) + \int \widetilde{Q}_{\gamma,1}(t,s)\hat{f}_n(t,s)^{1/2} \frac{f_{G_{-1}^\delta}(t,s)- f_{G_1^\delta}(t,s)}{\hat{f}_n(t,s)^{1/2}}d(t,s)\\
        & \overset{c.s.}{\geq} \omega_{n,\gamma}(\delta) - \left(\int \widetilde{Q}_{\gamma,1}(t,s)^2\hat{f}_n(t,s)d(ts)\right)^{1/2}\left(\int \left( \frac{f_{G_{-1}^\delta}(t,s)- f_{G_1^\delta}(t,s)}{\hat{f}_n(t,s)}\right)^2\hat{f}_n(t,s)d(ts)\right)^{1/2}\\
        & \geq \omega_{n,\gamma}(\delta) - (n\omega_{n,\gamma}'(\delta)^2)^{1/2}(\delta^2/n)^{1/2}\\
        & = \omega_{n,\gamma}(\delta) - \delta\omega_{n,\gamma}'(\delta)\\
        & = B_\gamma(G_{-1}^\delta) - B_\gamma(G_{1}^\delta)\\
        & = 2\sup_{G\in\calG}\left\{
        |B_\gamma(G)|
        \right\}
    \end{align*}
    The second inequality is a consequence of the bound in \eqref{eq:var_Q_tilde_bound}, and of the constraint in \eqref{eq:modulus_problem}. The final two equalities hold by applying lemma \ref{lem:max_bias_attained}. Finally, I obtain:
    \begin{align*}
        \sup_{G\in\calG}\left\{
\left|
\widetilde{B}_\gamma(G)
\right|
        \right\} &\geq |\widetilde{B}_\gamma(G_{-1})| \vee |\widetilde{B}_\gamma(G_{1})|\\
        & \geq \left(|
\widetilde{B}_\gamma(G_{-1})| + |\widetilde{B}_\gamma(G_{1})|
        \right)/2\\
        & \geq  \left|\widetilde{B}_\gamma(G_{-1}) - \widetilde{B}_\gamma(G_1)\right|/2\\
        & \geq \sup_{G\in\calG}\left\{
        |B_\gamma(G)|
        \right\}~,
    \end{align*}
    thus concluding the proof.
\end{proof}

\section{Bootstrap Proofs}
\label{app:bootstrap_theory}
This appendix contains the proof of theorem \ref{thm:bootstrap_validity}. The proof proceeds in three steps. First, I establish the H\"older smoothness properties of the regularized affine minimax estimators used in the inferential procedure, under a set of sufficient conditions. I use these H\"older smoothness properties to derive entropy and envelope bounds on the class of affine minimax estimators. Second, I use these bounds to obtain a Gaussian coupling to the empirical and multiplier supremum statistics defining the critical values used in the step-down procedure. To do so, I apply results due to \citealt{CCK_2014} and \citealt{CCK_2016}. Finally, I use these couplings to obtain the sufficient critical value validity of assumption \ref{assu:cv_validity_monotonicity}.

Before proceeding with the proof, I will introduce some notation that will be used throughout the remainder of this appendix. Let $\Gamma :=\mathbb{R}^2\times\mathbb{R}_+$. Let $Z_i := (\hat\theta_i,\sigma_i)$. Let $\xi_i$ denote the standard Gaussian multipliers, drawn independently from the $Z_i$. The econometrician observes a sample $(Z_1,\ldots,Z_n)$. The $Z_i$ are drawn iid from the law $P_G$. I denote by $P_{n,G}$ the joint law of the sample $(Z_1,\ldots,Z_n)$, and by $P_n$ the empirical law of the sample. Relatedly, I denote by $P_\xi$ and $P_{n,\xi}$ the marginal and joint laws of the Gaussian multipliers, respectively. For a given function $v$ taking arguments in $\mathbb{R}^2$, let $P_nv := \frac1 n \sum_{i=1}^n v(Z_i)$ and let $P_Gv := \mathbb{E}_G v(Z_i)$. Throughout the remainder of this section, we denote by $H$ the pilot estimator of $G$ given by $\hat{G}_n$. Again, for a given function $v$ taking arguments in $\mathbb{R}^2$, we can take the true and pilot-estimated variances of $v$:
\begin{equation*}
    \sigma_G(v)^2 := P_G( v - P_Gv)^2~,\quad \hat\sigma_n(v)^2 = P_H( v - P_Hv)^2~.
\end{equation*}
For any $\tilde\Gamma \subseteq\Gamma$, define:
\begin{equation*}
    T_n^\xi(\tilde{\Gamma}) = \sup_{\gamma\in\tilde\Gamma}\left|
\frac{1}{\sqrt{n}} \sum_{i=1}^n \xi_i\frac{Q_\gamma(Z_i)-P_nQ_\gamma}{\hat\sigma_n(Q_\gamma)}
    \right|~.
\end{equation*}
Similarly, define:
\begin{equation*}
    T_n(\tilde\Gamma) = \sup_{\gamma\in\tilde\Gamma}\left|
\frac{1}{\sqrt{n}} \sum_{i=1}^n \frac{Q_\gamma(Z_i)-P_GQ_\gamma}{\hat\sigma_n(Q_\gamma)}
    \right|~.
\end{equation*}
We denote by $\hat{c}_n$ the corresponding critical values:
\begin{equation*}
    \hat{c}_n(1-\alpha,\tilde\Gamma) = \inf\left\{
x\ge 0:P_{n,\xi}\left(
T_n^\xi(\tilde\Gamma)\le x
\right)\ge 1-\alpha
\right\}~.
\end{equation*}
Towards proving theorem \ref{thm:bootstrap_validity}, I demonstrate in the following section that the above $\hat{c}_n$ satisfy assumption \ref{assu:cv_validity_monotonicity}. In particular:
\begin{equation*}
    \limsup_{n\to\infty}\sup_{G\in\calG}P_{n,G}\left(
\sup_{\gamma\in\Gamma(G)} \frac{\sqrt{n}(\hat{L}_n(\gamma)-P_GQ_\Gamma)}{\sigma_G(Q_\gamma)}>\hat{c}_n(1-\alpha,\tilde\Gamma)
    \right)\le \alpha~.
\end{equation*}

In what follows, I will assume that $\calG$ is a convex class of distributions. I assume that each $G$ in $\calG$ admits a density with respect to a common dominating measure;\footnote{For ease of exposition, I take this dominating measure to be the Lebesgue measure throughout.} I denote by $g_G(\theta,s)$ this density. I will further assume that the support of $\theta$ and $\sigma$ are contained in subsets $\Theta$ and $\calS$ of $\bbR$, respectively. Let $\varphi_s$ denote the density of a gaussian $\mathcal{N}(0,s^2)$. The joint density of $\hat\theta$ and $\sigma$ can be given as:
\begin{equation*}
    f_G(t,s) =\int_\Theta \varphi_s(t-\theta)g_G(\theta,s)d\theta~,\quad (t,s)\in\bbR\times\calS~.
\end{equation*}
For each $\gamma$ in $\Gamma$, the affine minimax estimator $\hat{L}_n(\gamma)$ can be written as $P_nQ_\gamma$. Recall that each affine minimax estimator is based on the following modulus:
\begin{equation*}
    \omega_{n,\gamma}(\delta)=\sup_{G^+,G^-\in\calG}\left\{ 
    L(\gamma,G^+) - L(\gamma,G^-): \int_{\calS}\int_{\bbR} \frac{(f_{G^+}(t,s)-f_{G^-}(t,s))^2}{f_H(t,s)}dtds \le \frac{\delta^2}{n}
    \right\}~.
\end{equation*}
For each $\gamma$, recall that $\delta$ is chosen to trade off bias and variance appropriately, and let $G_{n,\gamma}^+$ and $G_{n,\gamma}^-$ denote the solutions to this modulus problem. Denote the choice of $\delta$ above by $\delta_{n,\gamma}$. Further recall that $Q_\gamma$ is given as follows, for a constant $d_{n,\gamma}$:
\begin{equation*}
    Q_\gamma(t,s) = \frac{n\omega_{n,\gamma}'(\delta_{n,\gamma})}{\delta_{n,\gamma}} \frac{f_{G_{n,\gamma}^+}(t,s) - f_{G_{n,\gamma}^-}(t,s)}{f_H(t,s)} + d_{n,\gamma}~.
\end{equation*} 
Finally, it will be useful to define $\Delta_{n,\gamma}$ as follows:
\begin{equation*}
    \Delta_{n,\gamma}(t,s) := f_{G_{n,\gamma}^+}(t,s) - f_{G_{n,\gamma}^-}(t,s)~.
\end{equation*}

\subsection{Statement of Assumptions}

I now state the assumptions needed for theorem \ref{thm:bootstrap_validity} to hold. I state four assumptions, which are as follows. Assumption \ref{assu:smoothness} imposes sufficient conditions on the regularity of densities $g_G$ for $G$ in $\calG$, which will be useful to control the entropy of the studentized test statistics of interest. Assumption \ref{assu:modulus_binding} places bounds on $\delta_{n,\gamma}$ which prevents the standard errors of the test statistics to approach zero too quickly, negating the regularity afforded by the first assumption \ref{assu:smoothness}. Assumptions \ref{assu:pilot_independence} and \ref{assu:variance_uniform_estimability} impose conditions on the pilot estimate $\hat{G}_n$. Assumption \ref{assu:pilot_independence} imposes independence of the $\hat{G}_n$ from the sample $(Z_1,\ldots,Z_n)$, which will allow conditioning on $\hat{G}_n$ in the following arguments. Assumption \ref{assu:variance_uniform_estimability} asserts that the pilot estimates of $\hat{\sigma}_n(Q_\gamma)$ are uniformly consistent estimators of $\sigma_G(Q_\gamma)$. 

Before proceeding, it will be instructive to introduce some new notation. First, let $m_G(s)$ denote the marginal density of the $\sigma$:
\begin{equation*}
    m_G(s) = \int_\Theta g_G(\theta,s)d\theta~.
\end{equation*}
Consequently, we may write $g_G$ as a product of the marginal density of the $\sigma$, and the conditional density of the $\theta$ given realizations of $\sigma$:
\begin{equation*}
    g_G(\theta,s) = m_G(s)q_G(\theta|s)~.
\end{equation*}

\begin{assu}
    \label{assu:smoothness}

    The following assumptions will take as given a choice of density $g_0(\theta, s) = m_0(s)q_0(\theta|s)$. First, assume that $\Theta$ is given by the compact interval $[-M_\theta, M)\theta]$ and that $\mathcal{S}$ is given by $[\underline{s},\bar{s}]$ for $0<\underline{s}<\bar{s}<\infty$. For this choice, we assume that there exists some $B> 9$ that can be uniquely written as $k+\zeta$ for some integer $k$, and some $\zeta \in (0,1]$, such that:

    \begin{enumerate}

\item $c \leq \frac{m_G(s)}{m_0(s)} \leq C$ for each $G\in\mathcal G$ and $s\in\mathcal S$. Furthermore,
\begin{equation*}
    \sup_{G\in\mathcal G}\sup_{s\in\mathcal S}
    \left|
    \partial_s^j
    \left(
        \frac{m_G(s)}{m_0(s)}
    \right)
    \right|
    \leq M,
    \qquad
    j=0,\ldots,k,
\end{equation*}
and
\begin{equation*}
    \sup_{G\in\mathcal G}
    \sup_{\substack{s,s'\in\mathcal S\\s\neq s'}}
    \frac{
        \left|
        \partial_s^k
        \left(
            \frac{m_G(s)}{m_0(s)}
        \right)
        -
        \partial_s^k
        \left(
            \frac{m_G(s')}{m_0(s')}
        \right)
        \right|
    }{
        |s-s'|^\zeta
    }
    \leq M.
\end{equation*}

\item $c \leq
    \frac{q_G(\theta\mid s)}{q_0(\theta\mid s)}
    \leq C$ for each $G\in\mathcal G$, $s\in\mathcal S$, and
$\theta\in\Theta$. Furthermore,
\begin{equation*}
    \sup_{G\in\mathcal G}\sup_{s\in\mathcal S}
    \left\|
    \partial_s^j
    \left(
        \frac{q_G(\cdot\mid s)}{q_0(\cdot\mid s)}
    \right)
    \right\|_{L^\infty(\Theta)}
    \leq M,
    \qquad
    j=0,\ldots,k,
\end{equation*}
and
\begin{equation*}
    \sup_{G\in\mathcal G}
    \sup_{\substack{s,s'\in\mathcal S\\s\neq s'}}
    \frac{
        \left\|
        \partial_s^k
        \left(
            \frac{q_G(\cdot\mid s)}{q_0(\cdot\mid s)}
        \right)
        -
        \partial_s^k
        \left(
            \frac{q_G(\cdot\mid s')}{q_0(\cdot\mid s')}
        \right)
        \right\|_{L^\infty(\Theta)}
    }{
        |s-s'|^\zeta
    }
    \leq M.
\end{equation*}

\item $c\leq m_0(s)\leq C,~
    c\leq q_0(\theta\mid s)\leq C$ for all $s\in\mathcal S$ and $\theta\in\Theta$. Moreover,
\begin{equation*}
    \sup_{s\in\mathcal S}
    \left|
        \partial_s^j m_0(s)
    \right|
    \leq M,
    \qquad
    j=0,\ldots,k,
\end{equation*}
and
\begin{equation*}
    \sup_{\substack{s,s'\in\mathcal S\\s\neq s'}}
    \frac{
        \left|
            \partial_s^k m_0(s)
            -
            \partial_s^k m_0(s')
        \right|
    }{
        |s-s'|^\zeta
    }
    \leq M.
\end{equation*}
Finally,
\begin{equation*}
    \sup_{s\in\mathcal S}
    \left\|
        \partial_s^j q_0(\cdot\mid s)
    \right\|_{L^\infty(\Theta)}
    \leq M,
    \qquad
    j=0,\ldots,k,
\end{equation*}
and
\begin{equation*}
    \sup_{\substack{s,s'\in\mathcal S\\s\neq s'}}
    \frac{
        \left\|
            \partial_s^k q_0(\cdot\mid s)
            -
            \partial_s^k q_0(\cdot\mid s')
        \right\|_{L^\infty(\Theta)}
    }{
        |s-s'|^\zeta
    }
    \leq M.
\end{equation*}

\end{enumerate}
    
\end{assu}

\begin{assu}
    \label{assu:modulus_binding} There exist fixed $0 <\underline{\delta} < \bar{\delta}<\infty$ such that for each $n$ and each $\gamma$ in $\Gamma$, $\delta_{n,\gamma}$ lies in the interval $[\underline{\delta},\bar\delta]$. Moreover, the modulus constraint binds at $\delta_{n,\gamma}$, such that:
    \begin{equation*}
        \int_\calS\int_\bbR \frac{\Delta_{n,\gamma}(t,s)^2}{f_H(t,s)}dtds = \frac{\delta_{n,\gamma}^2}{n}~.
    \end{equation*}
\end{assu}

\begin{assu}
    \label{assu:pilot_independence}
    The sample $(Z_1,\ldots,Z_n)$ is independent from $\hat{G}_n$, and therefore from the solutions $G_{n,\gamma}^+$ and $G_{n,\gamma}^-$ to the modulus problem.
\end{assu}

\begin{assu}
    \label{assu:variance_uniform_estimability}
    For some $\vartheta$ in $(2/B,2/9)$, there exists $\varepsilon_n= o(n^{-\vartheta/2})$ such that:
    \begin{equation*}
        \limsup_{n\to\infty}\sup_{G\in\calG}P_{n,G}\left(
\sup_{\gamma\in\Gamma(G)}\left|
\frac{\hat\sigma_n(Q_\gamma)}{\sigma_G(Q_\gamma)}-1
\right|>\varepsilon_n
        \right) = 0~.
    \end{equation*}
\end{assu}

\subsection{H\"older Smoothness Lemmas, Entropy, and Envelope Properties}
\label{app:Holder_smoothness}

This section establishes H\"older smoothness properties of the affine minimax estimators considered in this paper. From there, I obtain entropy bounds on the class of affine minimax estimators, and obtain bounds on the envelope of this class. To this end, I first show that after studentization it suffices to consider ratios of densities. I establish H\"older bounds on these ratios, and then apply the entropy bound of \citealt{kolmogorovtikhomirov1961} as stated in theorem 2.7.1. of \citealt{VdVandWellner1996} to control the entropy of the class of affine minimax estimators. Throughout this subsection, I condition on the pilot sample, and thus on the pilot estimator $\hat{G}_n \equiv H$, thanks to assumption \ref{assu:pilot_independence}. All bounds below will apply uniformly over $H$ in $\calG$.

\subsubsection*{Simplification of Affine Minimax Estimators After Studentization}

Recall that, by lemma \ref{lem:final_minimaxity}, $\hat{\sigma}_n(Q_\gamma)^2 = n\omega_{n,\gamma}'(\delta_{n,\gamma})^2$. Moreover, $P_HQ_\gamma = d_{n,\gamma}$, since the densities $f_{G_{n,\gamma}^+}$ and $f_{G_{n,\gamma}^-}$ both integrate to one. I define the normalized function $\psi_{n,\gamma}$, which by the above simplifies as follows:
\begin{equation*}
    \psi_{n,\gamma}(t,s) := \frac{Q_\gamma(t,s)-P_HQ_\gamma}{\hat\sigma_n(Q_\gamma)} = \frac{\sqrt{n}\Delta_{n,\gamma}(t,s)}{\delta_{n,\gamma}f_H(t,s)}~.
\end{equation*}
Consequently, we may write:
\begin{align*}
    \frac{\sqrt{n}(\hat{L}_n(\gamma) - P_GQ_\gamma)}{\hat\sigma_n(Q_\gamma)} &= \sqrt{n}(P_n-P_G)\psi_{n,\gamma}~,\\
    \frac{1}{\sqrt{n}}\sum_{i=1}^n\xi_i\frac{Q_\gamma(Z_i)-P_nQ_\gamma}{\hat\sigma_n(Q_\gamma)}&= \frac{1}{\sqrt{n}}\sum_{i=1}^n \xi_i(\psi_{n,\gamma}(Z_i)-P_n\psi_{n,\gamma})~.
\end{align*}
Consequently, it suffices to analyze the normalized functions $\psi_{n,\gamma}$, rather than the $Q_\gamma$.

In what follows, I will analyze the enlarged class:
\begin{equation*}
    \Psi_n(H) = \left\{
\frac{\sqrt{n}(f_{G^+}-f_{G^-})}{uf_H}:G^+,G^-\in\calG,u\in[\underline{\delta},\bar\delta],\int_\calS\int_\bbR\frac{(f_{G^+}(t,s)-f_{G^-}(t,s))^2}{f_H(t,s)}dtds = \frac{u^2}{n}
    \right\}~.
\end{equation*}
This class includes all the $\psi_{n,\gamma}$ selected as solutions to the modulus problem. Moreover, for each $\psi$ in $\Psi_n(H)$, $P_H\psi = 0$ and $P_H\psi^2 = 1$.

By assumption \ref{assu:smoothness}, $c^2 g_0 < g_G < C^2 g_0$ for any $G$ in $\calG$. Consequently, integrating against $\varphi_s$ gives $c^2 f_0 < f_G < C^2 f_0$. Since $f_0 >0$, it follows that:
\begin{equation}
\label{eq:density_fraction}
    \frac{c^2}{C^2}\le \frac{f_G(t,s)}{f_H(t,s)}\le \frac{C^2}{c^2}~.
\end{equation}
Moreover, the following bound on elements $\psi$ in $\Psi_n(H)$ will prove useful in controlling the behavior of the affine minimax estimators in the tails:
\begin{equation}
\label{eq:global_psi_bounds}
    |\psi(t,s)|\le \frac{\sqrt{n}}{\underline{\delta}}\left\{
\frac{f_{G^+}(t,s)}{f_H(t,s)}+\frac{f_{G^-}(t,s)}{f_H(t,s)}
    \right\} \le \frac{2C^2}{\underline{\delta}c^2}\sqrt{n}~.
\end{equation}
The above is a consequence of the triangle inequality.

\subsubsection*{H\"older Smoothness}

For $\tau \ge 1$, define the rectangle $D_\tau := [-\tau,\tau]\times\calS$. For each multi-index $a := (a_s,a_t)$, let $|a|=a_t + a_s$, and let $D^a = \partial_t^{a_t}\partial_s^{a_s}$. I define the following H\"older norm:
\begin{equation}
\label{eq:Holder_norm_defn}
    \|v\|_{C^B(D_\tau)} := \sum_{|a| \le k} \|D^a v\|_{L^\infty(D_\tau)} + \sum_{|a| =k} \sup_{\substack{z,z'\in D_\tau\\ z\neq z'}} \frac{|D^av(z) - D^av(z')|}{\|z-z'\|_2^\zeta}
\end{equation}
Before proceeding, it will be useful to state a few facts on H\"older spaces. For $B\le 1$, it holds that (see e.g. \citealt{gilbarg_trudinger_2001}) for any bounded set $E$:
\begin{equation*}
    \|vw\|_{C^B(E)} \le \|v\|_{C^B(E)}\|w\|_{C^B(E)}
\end{equation*}
For arbitrary $B$, the fact that H\"older spaces form a generalized Banach algebra (See example 21.73 of \citealt{zeidler1990nonlinear}) yields:
\begin{equation}
\label{eq:Holder_generalized_banach_algebra}
    \|vw\|_{C^B(E)} \le  C_B\|v\|_{C^B(E)}\|w\|_{C^B(E)}~,
\end{equation}
for constants $C_B$ in $\bbR$ that do not depend on $E$.

Finally, it will be useful to obtain the following interpolation inequalities. Let $E$ be either $D_\tau$ or $\calS$. Convexity and the mean value theorem provide:
\begin{equation*}
    |v(z)-v(z')| \le \min\left\{
2\|v\|_{L^\infty(E)},\|\nabla v\|_{L^\infty(E)}\|z-z'\|_{2}
    \right\}~.
\end{equation*}
Therefore, for $0<\zeta<1$, the following holds:
\begin{equation*}
    \sup_{\substack{z,z'\in E\\ z\neq z'}} \frac{|v(z)-v(z')|}{\|z-z'\|^\zeta}\le 2^{1-\zeta}\|v\|_{L^\infty(E)}^{1-\zeta}\|\nabla v\|_{L^\infty(E)}^\zeta~.
\end{equation*}
Applying further, define the constant $c_{v,E}$ to be:
\begin{equation*}
    c_{v,E} := \max\left\{
\|v\|_{L^\infty(E)},\max_j \|\partial_jv\|_{L^\infty(E)}
    \right\}~.
\end{equation*}
Consequently, in two dimensions, $\|\nabla v\|_{L^\infty(E)} \le \sqrt{2}c_{v,E}$. The above interpolation inequality gives:
\begin{equation}
\label{eq:interpolation_ineq_consequence}
    \|v\|_{C^\zeta(E)} \le 3c_{v,E}~,\quad 0<\zeta\le 1~.
\end{equation}
Apart from the constant $c_{v,E}$, the right hand side of the above does not depend on $v$ or $E$. 

Before stating and proving the main H\"older continuity lemma, I record a few immediate consequences of assumption \ref{assu:smoothness}. First, recall that the density $g_G$ can be written as:
\begin{equation*}
    g_G(\theta,s) = m_0(s)q_0(\theta|s) \frac{m_G(s)}{m_0(s)}\frac{q_G(\theta|s)}{q_0(\theta|s)}~.
\end{equation*}
Since each of these four factors lies between $c$ and $C$, as defined in assumption \ref{assu:smoothness}, it follows that $c^4 \le g_G(\theta,s) \le C^4$. Repeated application of the generalized Banach algebra property of H\"older spaces \eqref{eq:Holder_generalized_banach_algebra} yields:
\begin{equation*}
    \|g_G(\theta,\cdot)\|_{C^B(\calS)} \le C_B^3 \|m_0(\cdot)\|_{C^B(\calS)}\|q_0(\theta|\cdot)\|_{C^B(\calS)}\left\|
    \frac{m_G(\cdot)}{m_0(\cdot)}
    \right\|_{C^B(\calS)}\left\|
    \frac{q_G(\theta|\cdot)}{q_0(\theta|\cdot)}
    \right\|_{C^B(\calS)} \le C_B^3 K^4~,
\end{equation*}
for some constant $K>0$, where the constants above hold uniformly in $G\in\calG$. Consequently, $\|\partial_s^jg_G\|_{L^\infty(\calS)}$ are bounded above by the same constant as in the above equation, for $j\le k$. The interpolation inequality \eqref{eq:interpolation_ineq_consequence} gives:
\begin{equation}
\label{eq:lower_order_Holder_norm_bounds}
    \sup_{G\in\calG}\sup_{\theta\in\Theta}\|\partial_s^j g_G(\theta,\cdot)\|_{C^\zeta(\calS)}\lesssim 1~,
\end{equation}
for $0<\zeta\le 1$.

Finally, to bound the smoothness of elements of $\Psi_n(H)$, it will be useful to factorize the densities $f_G(t,s)$. In particular, write:
\begin{align*}
    \Tilde{f}_G(t,s) &:= \frac{f_G(t,s)}{\varphi_s(t)} = \int_\theta K_\theta(t,s)g_G(\theta,s)d\theta~,\\
    K_\theta(t,s) &:= \exp\left(
\frac{t\theta}{s^2}- \frac{\theta^2}{2s^2}
    \right)~.
\end{align*}
For any $\psi$ in $\Psi_n(H)$, I rewrite:
\begin{equation*}
    \psi(t,s) = \frac{\sqrt{n}}{u} \frac{\Tilde{f}_{G^+}(t,s) - \Tilde{f}_{G^-}(t,s)}{\Tilde{f}_H(t,s)}~.
\end{equation*}
This factorization eliminates a term in $\Tilde{f}_H(t,s)^{-1}$ that is a power of $\exp(\tau^2)$, where we recall that $\tau$ defines the domain of interest via $D_\tau$. When $\tau$ is proportional to $\sqrt{\log n}$, this term eliminated term grows at a rate proportional to a power of $n$. By eliminating this term, I analyze the smoothness of the numerator and denominator separately.

\begin{lem}
    \label{lem:key_smoothness}
    Let assumption \ref{assu:smoothness} hold. Then, there exist constants $C_1 \ge 1$ and $C_2 > 0$ such that, for all $n\ge 2$, $\tau \ge 1$, and $H\in\calG$:
    \begin{equation*}
        \sup_{\psi\in\Psi_n(H)}\|\psi\|_{C^B(D_\tau)}\le C_1\sqrt{n}\exp(C_2\tau)~.
    \end{equation*}
\end{lem}
\begin{proof}
    I first obtain bounds on derivatives of the kernel $K_\theta$. In particular, each derivative of $K_\theta(t,s)$ can be written as a sum of terms taking the form $t^p\theta^qs^{-r}K_\theta(t,s)$, for integers $p$, $q$, and $r$. To see this, I proceed by induction. The product rule yields:
    \begin{equation*}
        \partial_t\!\left(t^p\theta^q s^{-r}K_\theta\right)
        =
        p\,t^{p-1}\theta^q s^{-r}K_\theta
        +
        t^p\theta^{q+1}s^{-(r+2)}K_\theta~,
    \end{equation*}
    and:
    \begin{equation*}
        \partial_s\!\left(t^p\theta^q s^{-r}K_\theta\right)
        =
        -r\,t^p\theta^q s^{-(r+1)}K_\theta
        +
        t^p\theta^q s^{-r}
        \frac{\theta^2-2t\theta}{s^3}K_\theta~.
    \end{equation*}
    Since $|t|\le \tau$, $|\theta|\le M_\theta$, and $s\ge \underline{s}$, the derivatives of $K_\theta$ of order at most $k+1$ are bounded on $D_\tau$ by the product of a polynomial in $\tau$ and $K_\theta$. Consequently, for any multi-index $a$ such that $|a| \le k+1$, there exist constants $m_a$ and $C_0'$ such that:
    \begin{equation*}
        \sup_{\theta\in\Theta}\|D^aK_\theta\|_{L^\infty(D_\tau)} \lesssim (1+\tau)^{m_a}\exp(C_0'\tau)~.
    \end{equation*}
    There exists a further constant $C_0$ such that for any multi-index $a$ with $|a|\le k+1$ and any $\tau > 1$, $(1+\tau)^{m_a} \lesssim \exp (C_0\tau)$. Choosing $C_0 := C_0'\vee C_0''$ yields:
    \begin{equation}
    \label{eq:K_theta_Holder_norm_bounds}
        \sup_{\theta\in\Theta}\sup_{|a|\le k+1}\|D^aK_\theta\|_{L^\infty(D_\tau)} \lesssim \exp(C_0\tau)~, \quad \sup_{\theta\in\Theta}\sup_{|a|\le k}\|D^aK_\theta\|_{C^\zeta(D_\tau)} \lesssim \exp(C_0\tau)~,
    \end{equation}
    where the second inequality holds by \eqref{eq:interpolation_ineq_consequence}. 

    Combining \eqref{eq:lower_order_Holder_norm_bounds} and \eqref{eq:K_theta_Holder_norm_bounds}, it holds that whenever $p+q \le k$:
    \begin{equation*}
        \sup_{G\in\calG}\sup_{j\le q}\sup_{(t,s)\in D_\tau}\sup_{\theta\in\Theta}|\partial_t^p\partial_s^{q-j}K_\theta(t,s)\partial_s^jg_G(\theta,s)|\lesssim \exp(C_0\tau)
    \end{equation*}
    Since $\Theta$ is a compact interval, the right-hand side is integrable in $\theta$, and thus the left-hand side. Therefore, interchanging differentiation and integration via theorem 16.8.ii of \citealt{BillingsleyProbMeas} yields:
    \begin{equation}
    \label{eq:tilde_f_derivatives}
        \partial_t^p\partial_s^q\Tilde{f}_G(t,s) = \sum_{j=0}^q\binom{q}{j}\int_\Theta \partial_t^p\partial_s^{q-j}K_\theta(t,s)\partial_s^jg_G(\theta,s)d\theta~.
    \end{equation}
    Because $|s-s'| \le \|z-z'\|_2$, the H\"older norm of $\partial_s^jg_G(\theta,s)$ as a function of $t,s$ is bounded above the the H\"older norm of $\partial_s^jg_G(\theta,s)$ as a function of $s$. In particular, $\|\partial_s^jg_G(\theta,s)\|_{C^\zeta(D_\tau)} \le \|\partial_s^jg_G(\theta,s)\|_{C^\zeta(\calS)}$. Consequently, each integrand in \eqref{eq:tilde_f_derivatives} can be bounded in $C^\zeta(D_\tau)$. Applying \eqref{eq:Holder_generalized_banach_algebra} to \eqref{eq:K_theta_Holder_norm_bounds} and \eqref{eq:lower_order_Holder_norm_bounds} yields:
    \begin{equation*}
        \sup_{G\in\calG}\sup_{j\le q}\sup_{\theta\in\Theta} \left\|
        \partial_t^p\partial_s^{q-j}K_\theta(t,s)\partial_s^jg_G(\theta,s)
        \right\|_{C^\zeta(D_\tau)} \lesssim \exp(C_0\tau)~,
    \end{equation*}
    whenever $p+q\le k$. Applying Minkowski's inequality yields:\footnote{Let $F(z) := \int_{\Theta} f_{\theta}(z)\,d\theta$. Then
$\frac{|F(z)-F(z')|}{\|z-z'\|^\zeta}
\le
\int_{\Theta}
\frac{|f_{\theta}(z)-f_{\theta}(z')|}
{\|z-z'\|^\zeta}
\,d\theta$, by Minkowski's inequality.
Taking the supremum over $z\neq z'$ gives
$\|F\|_{C^\zeta(E)}
\le
\int_{\Theta}
\|f_{\theta}\|_{C^\zeta(E)}
\,d\theta$.}
    \begin{equation*}
        \left\|
        \int_\Theta \partial_t^p\partial_s^{q-j}K_\theta(t,s)\partial_s^jg_G(\theta,s)d\theta
        \right\|_{C^\zeta(D_\tau)} \le 
        \int_\Theta\left\| \partial_t^p\partial_s^{q-j}K_\theta(t,s)\partial_s^jg_G(\theta,s)\right\|_{C^\zeta(D_\tau)}d\theta
    \end{equation*}
    Because $\Theta$ is compact, it follows that:
    \begin{equation}
    \label{eq:D_A_f_G_holder_bounds}
        \sup_{G\in\calG}\max_{|a|\le k} \|D^a \Tilde{f}_G\|_{L^\infty(D_\tau)} \le \sup_{G\in\calG}\max_{|a|\le k} \|D^a \Tilde{f}_G\|_{C^\zeta(D_\tau)} \lesssim \exp(C_0\tau)~.
    \end{equation}
    Consequently, $\sup_{G\in\calG}\| \Tilde{f}_G\|_{C^B(D_\tau)} \lesssim \exp(C_0\tau)$.

    It remains to bound $\Tilde{f}_H^{-1}$. Since $|t|\le \tau$, $|\theta|\le M_\theta$, and $s\ge \underline{s}$, $K_\theta(t,s) \ge \exp(-M_\theta\tau/\underline{s}^2 - M_\theta^2/(2\underline{s}^2))$ on $D_\tau$. This bound, and the fact that $\int_\Theta g_H(\theta,s)d\theta = m_H(s)\ge c^2$ yields:
    \begin{equation}
    \label{eq:inv_f_H_lower_bound}
        \inf_{(t,s)\in D_\tau} \Tilde{f}_H(t,s)\ge c^2\exp\left(
-\frac{M_\theta\tau}{\underline{s}^2}-\frac{M_\theta^2}{2\underline{s}^2}
        \right)~.
    \end{equation}
    Consequently, $\|\Tilde{f}_H^{-1}\|_{L^\infty(D_\tau)}\lesssim \exp(D_0'\tau)$, where $D_0' := M_\theta / \underline{s}^2$. Moreover:
    \begin{equation*}
        \frac{\left|
\widetilde f_H(z)^{-1}
-
\widetilde f_H(z')^{-1}
\right|}{\|z-z'\|_2}
=
\frac{
\left|
\widetilde f_H(z)-\widetilde f_H(z')
\right|
}{
\widetilde f_H(z)\widetilde f_H(z')\|z-z'\|_2
}~.
    \end{equation*}
    Therefore, by \eqref{eq:D_A_f_G_holder_bounds} and \eqref{eq:inv_f_H_lower_bound}, $\sup_{H\in\calG}\|\Tilde{f}_H^{-1}\|_{C^\zeta(D_\tau)}\lesssim \exp(D_0\tau)$, where $D_0 = D_0'\vee C_0$. It remains to control the H\"older seminorm of higher-order derivatives. For two multi-indices $a$ and $b$, let $b\le a$ denote $b_s\le a_s$ and $b_t \le a_t$. Also let $a-b := (a_s-b_s, a_t-b_t)$. The product rule applied to $\Tilde{f}_H\Tilde{f}_H^{-1}$ yields:
    \begin{equation*}
        D^a(\Tilde{f}_H^{-1}) = -(\Tilde{f}_H^{-1}) \sum_{0<b\le a} \binom{a_s}{b_s}\binom{a_t}{b_t}D^b\Tilde{f}_HD^{a-b}(\Tilde{f}_H^{-1})~,\quad 1\le|a|\le k~.
    \end{equation*}
    To control $\|D^a(\Tilde{f}_H^{-1})\|_{C^\zeta(D_\tau)}$, I proceed by induction. Applying \eqref{eq:Holder_generalized_banach_algebra}, the triangle inequality, \eqref{eq:D_A_f_G_holder_bounds}, and the inductive hypothesis provide that there exists a constant $E_0$ satisfying:
    \begin{equation*}
        \sup_{H\in\calG}\max_{|a|\le k} \|D^a(\Tilde{f}_H^{-1})\|_{C^\zeta(D_\tau)}\lesssim \exp(D_0\tau)~.
    \end{equation*}
    The reasoning above provides $\sup_{H\in\calG}\|\Tilde{f}_H^{-1}\|_{C^B(D_\tau)}\lesssim \exp(D_0\tau)$. 

    Combining results yields, uniformly over $H$, $G^+$, $G^-$, and $u$:
    \begin{equation*}
        \|\psi(t,s)\|_{C^B(D_\tau)} \le \frac{\sqrt{n}}{\delta}C_B\|\Tilde{f}_{G^+}- \Tilde{f}_{G^-}\|_{C^B(D_\tau)}\|\Tilde{f}_H^{-1}\|_{C^B(D_\tau)} \le C_1\sqrt{n}\exp(C_2\tau)~,
    \end{equation*}
    thus concluding the proof.
\end{proof}

\subsubsection*{Entropy Bounds}

Lemma \ref{lem:key_smoothness} points to the regularity of the affine minimax estimators applied in this paper. In this section, I use these regularity results to obtain entropy bounds on the class $\Psi_n(H)$. For any given norm $\|\cdot\|$, set $\calF$, and $\varepsilon>0$, we denote by $N(\varepsilon,\calF,\|\cdot\|)$ the covering number of $\calF$, the minimal number of balls of radius $\varepsilon$ needed to cover $\calF$. As in \citealt{VdVandWellner1996}, the centers of these balls need not lie in $\calF$, but must have finite norm. We will primarily consider the seminorm $\|\cdot\|_{L^2(Q)}$ for Borel probability measures on $\bbR\times\calS$ given by $\|v\|_{L^2(Q)} = (\int v^2 dQ)^{1/2}$.

I provide two entropy lemmata in what follows. In the first lemma \ref{lem:general_holder_entropy}, I provide general results on the entropy of H\"older balls in $C^B(D_\tau)$ by extending theorem 2.7.1 of \citealt{VdVandWellner1996}. In lemma \ref{lem:psi_tr_entropy}, I apply these general results to obtain bounds on the covering numbers of a truncation of the class $\Psi_n(H)$.

\begin{lem}
\label{lem:general_holder_entropy}
    Let $\tau \ge 1$, $L\ge 1$, and let $\calF$ be a class of functions on $D_\tau$ satisfying $\sup_{v\in\calF} \|v\|_{C^B(D_\tau)} \le L$. Then:
    \begin{equation*}
        \log N(\varepsilon,\calF, \|\cdot\|_{L^\infty(D_\tau)})\le C_{B,\calS} (1+\tau)(L/\varepsilon)^{2/B}~,\quad 0\le\varepsilon\le 1~,
    \end{equation*}
    where $C_{B,\calS}$ depends only on $B$ and $\calS$.
\end{lem}
\begin{proof}
    Let $D_\tau^1 :=\{z\in\bbR^2:\inf_{w\in D_\tau} \|z-w\|_2<1\}$ be the 1-enlargement $D_\tau$. Notice that $D_\tau^1$ is bounded and convex in $\bbR^2$ with nonempty interior. Let $\Leb$ denote the Lebesgue measure on $\bbR^2$. Since $D_\tau^1$ is contained $[-\tau-1,\tau+1]\times[\underline{s}-1,\bar{s}+1]$, it follows that:
    \begin{equation*}
        \Leb (D_\tau^1) \le (2\tau +2)(\bar{s}-\underline{s}+2)\lesssim 1+\tau~.
    \end{equation*}
    First, notice that the H\"older norm \eqref{eq:Holder_norm_defn} dominates the H\"older norm considered in \citealt{VdVandWellner1996}. Defining the class $\calF/L$ by $\{v/L:v\in\calF\}$, and applying theorem 2.7.1 of \citealt{VdVandWellner1996} yields:
    \begin{equation*}
        \log N(\varepsilon/(2L), \calF/L,\|\cdot\|_{L^\infty(D_\tau)}) \le K\Leb(D_\tau^1)\left(
\frac{2L}{\varepsilon}
        \right)^{2/B}\lesssim (1+\tau)\left(
\frac{L}{\varepsilon}
        \right)^{2/B}~,
    \end{equation*}
    for some constant $K$ depending only on $B$ and the dimension $D_\tau$, both of which are fixed. Let $B(v_i,\varepsilon/2L)$ denote, for $i=1,\ldots,N(\varepsilon/(2L), \calF/L,\|\cdot\|_{L^\infty(D_\tau)})$, be a collection of balls covering $\calF/L$. For each $i$, choose $v_i'$ to be an element of $\calF/L$, and notice that the collection $B(v_i',\varepsilon/L)$ also covers $\calF/L$. Finally, the collection $B(v_i'L,\varepsilon)$ covers $\calF$, thus proving the lemma.
\end{proof}
To apply this result, let $\tau_n := D\sqrt{\log(n)}$ for a constant $D>0$. I define the truncated class:
\begin{equation*}
    \Psi_n^{tr}(H) := \{\psi(t,s)\mathbb{1}\{|t|\le\tau_n\}:\psi\in\Psi_n(H)\}~.
\end{equation*}
For sufficiently large $n$, $\tau_n \ge 1$. The choice of $\tau_n$ will allow for control of observations outside $D_{\tau_n}$.
\begin{lem}
\label{lem:psi_tr_entropy}
Let assumption \ref{assu:smoothness} hold, and fix $\vartheta \in (2/B,2/9)$. Then, there exist constants $C <\infty$ and $n_0<\infty$ such that for every $n\ge n_0$, and for any $0<\varepsilon\le 1$:
\begin{equation*}
    \sup_{H\in\calG}\sup_{Q} \log N(\varepsilon,\Psi_n^{tr},\|\cdot\|_{L^2(Q)}) \le Cn^{\vartheta/2}\varepsilon^{-\vartheta}~,
\end{equation*}
where the second supremum is taken over all Borel probability measures over $\bbR\times\calS$.
\end{lem}
\begin{proof}
    I first relate the $L^2(Q)$ norm on $\Psi_n^{tr}(H)$ to the $L^\infty(D_\tau)$ norm on $\Psi_n(H)$.
    \begin{align*}
        \left\| 
        (\psi - \Tilde{\psi})\mathbb{1}\{|t|\le\tau_n\}
        \right\|_{L^2(Q)}^2 &= \int_{D_{\tau_n}}|\psi - \Tilde{\psi}|^2 dQ\\
        & \le \left\| 
        (\psi - \Tilde{\psi})
        \right\|_{L^\infty(D_{\tau_n})}^2Q(D_{\tau_n})\\
        & \le \left\| 
        (\psi - \Tilde{\psi})
        \right\|_{L^\infty(D_{\tau_n})}^2~.
    \end{align*}
    Therefore, any cover of $\Psi_n(H)$ using the $L^\infty(D_\tau)$ norm provides a cover of $\Psi_n^{tr}(H)$ using the $L^2(Q)$ norm. 

    By lemma \ref{lem:key_smoothness}, elements of $\Psi_n^{tr}$ have $C^B(D_{\tau_n})$ norm bounded by $C_1\sqrt{n}\exp(C_2\tau_n)$. Combining this result with the above inequality on norms, and with lemma \ref{lem:general_holder_entropy} yields:
    \begin{align}
        \sup_{H\in\calG}\sup_Q \log N(\varepsilon,\Psi_n^{tr},\|\cdot\|_{L^2(Q)}) &\lesssim (1+\tau_n)\left(
\frac{C_1\sqrt{n}\exp(C_2\tau_n)}{\varepsilon}
        \right)^{2/B}\nonumber\\
        & \lesssim (1+\tau_n)\exp(C_3\tau_n)n^{1/B}\varepsilon^{-2/B}\label{eq:entropy_bound_tr_class}~,
    \end{align}
    for $C_3 := C_2\frac{2}{B}$. First, notice that $(1+\tau_n)\exp(C_3\tau_n)$ is $n^{o(1)}$. To see this, notice that:
    \begin{equation*}
        \frac{\log(1+\tau_n) + C_3\tau_n}{\log n}\to 0~.
    \end{equation*}
    Since $\vartheta /2-1/B > 0$, it follows that there exists $n_0 < \infty$ such that for $n\ge n_0$:
    \begin{equation*}
        (1+\tau_n)\exp(C\tau_n) \le n^{\vartheta /2 - 1/B}~.
    \end{equation*}
    Again because $\vartheta > 2/B$, $\varepsilon^{-2/B} \le \varepsilon^{-\vartheta}$ for $0\le \varepsilon\le 1$. Substituting these bounds into \eqref{eq:entropy_bound_tr_class} provides the lemma.
\end{proof}
Before proceeding, I will provide an immediate consequences of lemma \ref{lem:psi_tr_entropy}. Let $\calF \subseteq \Psi_n^{tr}(H)$. Any covering of $\Psi_n^{tr}(H)$ is a covering of $\calF$, meaning that:
\begin{equation}
\label{eq:subclass_bound}
    \log N(\varepsilon,\calF,L^2(Q)) \le \log N(\varepsilon,\Psi_n^{tr}(H),L^2(Q))~.
\end{equation}
To handle centering, fix $G\in\calG$ and define $Q^* := (Q+P_G)/2$. For $v$ and $w$ in $\calF$, applying the triangle inequality and Cauchy-Schwarz yields:
\begin{align*}
    \|(v-P_Gv)-(w-P_Gw)\|_{L^2(Q)} &\le  \|v-w\|_{L^2(Q)} + |P_G(v-w)|\\
    & \le \|v-w\|_{L^2(Q)} + \|v-w\|_{L^2(P_G)}\\
    & \le \sqrt{2\left(
    \|v-w\|_{L^2(Q)}^2 + \|v-w\|_{L^2(P_G)}^2
    \right)}\\
    & \le 2\|v-w\|_{L^2(Q^*)}~.
\end{align*}
Consequently, any cover of $\calF$ at radius $\varepsilon/2$ using the $L^2(Q^*)$ seminorm provides a cover of $\{v-P_Gv:v\in\calF\}$ at radius $\varepsilon$ using the $L^2(Q)$ seminorm. Applying the bound \eqref{eq:subclass_bound} and lemma \ref{lem:psi_tr_entropy} yields, for any $\calF \subseteq \Psi_n^{tr}(H)$:
\begin{equation}
\label{eq:subclass_centered_bound}
    \sup_{H,G\in\calG}\sup_Q \log N(\varepsilon,\{v-P_Gv:v\in\calF\},L^2(Q))\lesssim n^{\vartheta/2}\varepsilon^{-\vartheta}~,\quad 0<\varepsilon\le 1~.
\end{equation} 

\subsubsection*{Envelope Properties}

The normalization $P_H\psi^2$ bounds the $L^2(P_H)$ norm of elements $\psi$ in $\Psi_n(H)$. This section converts these $L^2(P_H)$ bounds into bounds in $L^\infty(D_\tau)$. I leverage the H\"older smoothness assumptions in \ref{assu:smoothness} to obtain bounds in the $s$ coordinate. To obtain bounds in the coordinate $t$, I relate $\psi(t,s)$ to a particular Laplace transform, and then apply a general fact on Laplace transforms which I derive in lemma \ref{lem:laplace_fact}.

\begin{lem}
\label{lem:laplace_fact}
    Fix $L,v>0$ and $0<\epsilon,\rho<1$. For any real integrable function $w$ on $[-L,L]$, let:
    \begin{equation*}
        M_w(t) := \int_{-L}^L e^{\lambda t}w(\lambda)d\lambda~,\quad \|w\|_1:=\int_{-L}^L |w(a)|da~.
    \end{equation*}
    Then, for every $t\in\bbR$:
    \begin{equation*}
        |M_w(t)|\le C\|w\|_1^\epsilon\|M_w\|_{L^2(\varphi_v)}^{1-\epsilon}\exp\left\{
        \frac{\rho t^2}{2v^2(1+\rho)}
        \right\}~,
    \end{equation*}
    where $\varphi_v$ denotes the $\mathcal{N}(0,v^2)$ density.
\end{lem}
\begin{proof}
    First, notice that $|M_w(t)| \le \|w\|_1e^{L|t|}$. As a result, $M_w$ lies in $L^2(\varphi_v)$. We denote by $\He_j$ the probabilist's Hermite polynomials, which are orthogonal in $L^2(\varphi_1)$. Taking the generating function of the Hermite polynomials, and applying Mehler's formula (see equations 18.12.16 and 18.18.28 in \citealt{NIST:DLMF}, respectively):
    \begin{equation}
    \label{eq:gen_function_and_mehler}
        \exp\left(
\lambda x -\lambda^2/2
        \right) = \sum_{j=0}^\infty \frac{\lambda^j}{j!}\He_j(x)~,\quad \sum_{j=0}^\infty \rho^j\frac{\He_j(x)^2}{j!} = \frac{\exp\left(
        \frac{\rho x^2}{1+\rho}
        \right)}{\sqrt{1-\rho^2}}~.
    \end{equation}
    Applying the first identity with $x=t/v$ and replacing $\lambda$ by
    $v\lambda$ gives:
    \begin{equation*}
        e^{\lambda t}
        =
        e^{v^2\lambda^2/2}
        \sum_{j=0}^\infty
        \frac{(v\lambda)^j}{j!}\He_j(t/v)~.
    \end{equation*}
    Substituting this expression into the definition of $M_w$ gives:
\begin{align}
    M_w(t)
    &=
    \int_{-L}^L
    e^{v^2\lambda^2/2}
    \sum_{j=0}^\infty
    \frac{(v\lambda)^j}{j!}\He_j(t/v)
    w(\lambda)d\lambda \nonumber\\
    &=
    \sum_{j=0}^\infty
    b_j\frac{\He_j(t/v)}{j!}~,
\end{align}
where
\begin{equation*}
    b_j
    :=
    \int_{-L}^L
    e^{v^2\lambda^2/2}
    (v\lambda)^j
    w(\lambda)d\lambda~.
\end{equation*}
This defines the Hermite expansion of $M_w(t)$. In $L^2(\varphi_v)$, the functions $\He_j(t/v)/\sqrt{j!}$ are orthogonal. Multiplying the Hermite expansion of $M_w$ by $\He_j(t/v)$ and integrating
against $\varphi_v$ yields:
\begin{equation*}
    \int_{\mathbb R}
    M_w(t)\He_j(t/v)\varphi_v(t)dt
    =
    \sum_{k=0}^\infty
    \frac{b_k}{k!}
    \int_{\mathbb R}
    \He_k(t/v)\He_j(t/v)\varphi_v(t)dt
    =
    b_j~.
\end{equation*}
By Bessel's inequality, $\sum_{j=0}^\infty \frac{b_j^2}{j!} \le \|M_w\|_{L^2(\varphi_v)}^2$. As usual, we may bound $|b_j| \le \|w\|_1 e^{v^2L^2/2} (vL)^j$, implying that:
\begin{align*}
    \sum_{j=0}^\infty
    \frac{b_j^2}{j!}\rho^{-j/\epsilon}
    &\le
    \|w\|_1^2 e^{v^2L^2}
    \sum_{j=0}^\infty
    \frac{(vL)^{2j}\rho^{-j/\epsilon}}{j!}
    \nonumber\\
    &=
    \|w\|_1^2
    \exp\left\{
        v^2L^2\left(1+\rho^{-1/\epsilon}\right)
    \right\}.
\end{align*}
Applying H\"older's inequality to the above with exponents $1/(1-\epsilon)$ and $1/\epsilon$ yields:
\begin{align}
    \sum_{j=0}^\infty
    \frac{b_j^2}{j!}\rho^{-j}
    &=
    \sum_{j=0}^\infty
    \left(
        \frac{b_j^2}{j!}
    \right)^{1-\epsilon}
    \left(
        \frac{b_j^2}{j!}\rho^{-j/\epsilon}
    \right)^\epsilon
    \nonumber\\
    &\le
    \left(
        \sum_{j=0}^\infty
        \frac{b_j^2}{j!}
    \right)^{1-\epsilon}
    \left(
        \sum_{j=0}^\infty
        \frac{b_j^2}{j!}\rho^{-j/\epsilon}
    \right)^\epsilon
    \nonumber\\
    &\lesssim
    \|w\|_1^{2\epsilon}
    \|M_w\|_{L^2(\varphi_v)}^{2(1-\epsilon)}~.
\end{align}
Finally, I apply Cauchy--Schwarz to the Hermite expansion of $M_w$.
First, rewrite each term in the expansion as:
\begin{align}
    |M_w(t)|
    &=
    \left|
    \sum_{j=0}^\infty
    \left(
        \frac{b_j}{\sqrt{j!}}\rho^{-j/2}
    \right)
    \left(
        \rho^{j/2}\frac{\He_j(t/v)}{\sqrt{j!}}
    \right)
    \right|
    \nonumber\\
    &\le
    \left(
        \sum_{j=0}^\infty
        \frac{b_j^2}{j!}\rho^{-j}
    \right)^{1/2}
    \left(
        \sum_{j=0}^\infty
        \rho^j\frac{\He_j(t/v)^2}{j!}
    \right)^{1/2},
\end{align}
where the inequality follows from Cauchy--Schwarz.

The first term is bounded by the preceding H\"older inequality. To analyze the second term, recall Mehler's inequality as written in \eqref{eq:gen_function_and_mehler}. Combining yields:
\begin{equation*}
    |M_w(t)|
    \lesssim
    \|w\|_1^\epsilon
    \|M_w\|_{L^2(\varphi_v)}^{1-\epsilon}
    \exp\left\{
        \frac{\rho t^2}{2v^2(1+\rho)}
    \right\}~,
\end{equation*}
thus concluding the proof.
\end{proof}

\begin{lem}
\label{lem:main_envelope}
Let assumptions \ref{assu:smoothness} and \ref{assu:modulus_binding} hold, and fix $D> 0$ and $\eta > 0$. There exists $n_0 < \infty$ sufficiently large and a constant $C_\eta>0$ depending only on $\eta$ such that for any $n\ge n_0$:
\begin{equation}
\label{eq:main_envelope_main}
\sup_{H\in\calG}\sup_{\psi \in \Psi_n(H)} \|\psi\|_{L^\infty(D_{\tau_n})} \le C_\eta n^{1/(4B+2) + \eta}~.
\end{equation}
Moreover, for any $H,G\in\calG$ and $\psi \in \Psi_{n}^{tr}(H)$, $|\psi-P_G\psi| \le 2C_\eta n^{1/(4B+2) + \eta}$. For $j=2,3,4$:
\begin{equation}
\label{eq:main_envelope_moments}
    \sup_{G,H\in\calG}\sup_{\psi \in \Psi_n^{tr}(H)} P_G|\psi - P_G\psi|^j \le \frac{C^2}{c^2}\left(
C_\eta n^{1/(4B+2) + \eta}
    \right)^{j-2}~.
\end{equation}
\end{lem}
\begin{proof}
    Fix $H\in\calG$ and $\psi \in \Psi_n(H)$. Recall that:
    \begin{equation*}
        \psi(t,s)\Tilde{f}_H(t,s) = \frac{\sqrt{n}}{u}\left\{
\Tilde{f}_{G^+}(t,s) - \Tilde{f}_{G^-}(t,s)
        \right\}~.
    \end{equation*}
    I first bound this product using the preceding lemma \ref{lem:laplace_fact}. The bound for $\Tilde{f}_H^{-1}$ obtained in the proof of lemma \ref{lem:key_smoothness} will provide the desired bound for $\psi$. Throughout the proof, implicit constants are independent of $G^+$, $G^-$, $H$, $u$, and $n$.

    I first consider the problem of obtaining bounds for the product $\psi(t,s)\Tilde{f}_H(t,s)$. I proceed in two steps. First, I obtain bounds on the $L^2(\varphi_{\underline{s}/2})$-norm of $\psi(t,s)\Tilde{f}_H(t,s)$. I apply the Gagliardo-Nirenberg inequalities (see e.g. \citealt{Nirenberg1959}) to use the smoothness assumptions of assumption \ref{assu:smoothness} to translate this $L^2(\varphi_{\underline{s}/2})$ bound into a uniform bound. 

    Towards the first objective, recall the normalization $P_H\psi^2$. For $s\in\calS$, notice that:
    \begin{align}
        \frac{\varphi_{\bar{s}/2}(t) \Tilde{f}_H(t,s)}{\varphi_{s}(t)} &= \frac{2s}{\underline{s}}\Tilde{f}_H(t,s)\exp\left(
\frac{-2t^2}{\underline{s}^2} + \frac{t^2}{2s^2}
        \right)\nonumber\\
        & \le \frac{2s}{\underline{s}} m_H(s)\exp\left(
\frac{-2t^2}{\underline{s}^2} + \frac{t^2}{2s^2} + \frac{M_\theta|t|}{\underline{s}^2}
        \right)\nonumber\\
        & \lesssim \exp\left(
\frac{-3t^2}{2\underline{s}^2} + \frac{M_\theta|t|}{\underline{s}^2}
        \right)\lesssim 1 \label{eq:tilde_f_CoM}~.
    \end{align}
    The first inequality holds since $K_\theta(t,s) \le \frac{M_\theta|t|}{\underline{s}^2}$. The second inequality holds since $\underline{s}<s<\bar{s}$ and since $c<m_H(s) <C$, per assumption \ref{assu:smoothness}. The final inequality holds by completing the square.\footnote{To see this, notice that $-\frac{3t^2}{2s^2}+\frac{M_\theta |t|}{s^2}
=
-\frac{3}{2s^2}\left(|t|-\frac{M_\theta}{3}\right)^2
+\frac{M_\theta^2}{6s^2}
\le
\frac{M_\theta^2}{6\underline{s}^2}$.} As a result, it follows that:
\begin{align}
\int_{\mathcal S}\int_{\mathbb R}
\left|\psi(t,s)\widetilde f_H(t,s)\right|^2
\varphi_{\underline s/2}(t)\,dt\,ds
&\lesssim
\int_{\mathcal S}\int_{\mathbb R}
\psi(t,s)^2 \widetilde f_H(t,s)\varphi_s(t)\,dt\,ds\nonumber
\\
&=\int_{\mathcal S}\int_{\mathbb R}
\psi(t,s)^2 f_H(t,s)\,dt\,ds\nonumber\\
&= P_H\psi^2
= 1~, \label{eq:avg_L2_bound_on_prod}
\end{align}
where the inequality holds by the preceding bound \eqref{eq:tilde_f_CoM}. 

This controls the average $L^2(\varphi_{\underline{s}/2})$ norm over $s$. Our goal is to obtain a bound on the $L^2(\varphi_{\underline{s}/2})$ norm, uniformly over $s$. I achieve this by leveraging the smoothness assumption \ref{assu:smoothness}. Applying \eqref{eq:D_A_f_G_holder_bounds} with $\tau = 1\vee|\tau|$ yields:
\begin{align}
    \|\psi(t,\cdot)\Tilde{f}_H(t,\cdot)\|_{C^B(\calS)} &\le \frac{\sqrt{n}}{u}\left\{
\|\Tilde{f}_{G^+}(t,\cdot)\|_{C^B(\calS)} + \|\Tilde{f}_{G^-}(t,\cdot)\|_{C^B(\calS)}
    \right\}\nonumber\\
    &\lesssim \sqrt{n}\exp(C_0|t|)\label{eq:C_B_bound_num_envelope}~,
\end{align}
where the final inequality uses $u\ge \underline{\delta}$ under assumption \ref{assu:modulus_binding}. An extension of the interpolation lemma in \citealt{Nirenberg1959} yields, for any $0<\ell<\bar{s}- \underline{s} $:
\begin{align}
    |\psi(t,s)\Tilde{f}_H(t,s)| &\lesssim \ell^{-1/2}\|\psi(t,\cdot)\Tilde{f}_H(t,\cdot)\|_{L^2(\calS)} + \ell^B\|\psi(t,\cdot)\Tilde{f}_H(t,\cdot)\|_{C^B(\calS)}\nonumber\\
    & \lesssim \ell^{-1/2}\|\psi(t,\cdot)\Tilde{f}_H(t,\cdot)\|_{L^2(\calS)} + \ell^B\sqrt{n}\exp(C_0|t|)\label{eq:Nirenberg_extension}~,
\end{align}
where the second inequality holds by \eqref{eq:C_B_bound_num_envelope}. As a result:
\begin{align*}
    \|\psi(\cdot,s)\Tilde{f}_H(\cdot,s)\|_{L^2(\varphi_{\underline{s}/2})} &= \left(\int_\bbR (\psi(t,s)\Tilde{f}_H(t,s))^2\varphi_{\underline{s}/2}(t)dt\right)^{1/2}\\
    & \lesssim \left(\int_\bbR\left\{
\ell^{-1/2}\|\psi(t,\cdot)\Tilde{f}_H(t,\cdot)\|_{L^2(\calS)} + \ell^B\sqrt{n}\exp(C_0|t|)
    \right\}^2 \varphi_{\underline{s}/2}(t)dt\right)^{1/2}\\
    & \le \ell^{-1/2}\int_{\bbR}\|\psi(t,\cdot)\Tilde{f}_H(t,\cdot)\|_{L^2(\calS)}^2\varphi_{\underline{s}/2}(t)dt +\ell^B\sqrt{n}\int_{\bbR} \exp(2C_0|t|)\varphi_{\underline{s}/2}(t)dt\\
    & \lesssim \ell^{-1/2} + \sqrt{n}\ell^B~.
\end{align*}
The first inequality holds by \eqref{eq:Nirenberg_extension}. The second inequality follows by the triangle inequality. The final inequality holds by Fubini's theorem and \eqref{eq:avg_L2_bound_on_prod}. The constants implicit above do not depend on $s$. Choosing $\ell = n^{-1/(2B+1)}$, which is less than $\bar{s}-\underline{s}$ for sufficiently large $n$, gives:
\begin{equation}
    \label{eq:envelope_unif_in_s_bounds}
    \sup_{s\in\calS}\|\psi(\cdot,s)\Tilde{f}_H(\cdot,s)\|_{L^2(\varphi_{\underline{s}/2})}\lesssim n^{1/(4B+2)}~.
\end{equation}

To complete the proof, it remains to translate the uniform bounds \eqref{eq:envelope_unif_in_s_bounds}. To this end, I show how the product $\psi\Tilde{f}_H$ can be written a Laplace transform. I then apply the general bounds developed in lemma \ref{lem:laplace_fact} and \eqref{eq:envelope_unif_in_s_bounds} to obtain the desired bounds uniform in $t$. To this end, let $\lambda := \theta/s^2$, and consider the following change of variables:
\begin{align*}
\psi(t,s)\Tilde{f}_H(t,s)
&=
\frac{\sqrt n}{u}
\int_{\Theta}
e^{t\theta/s^2-\theta^2/(2s^2)}
\left\{
g_{G^+}(\theta,s)-g_{G^-}(\theta,s)
\right\}
\,d\theta
\\
&=
\int_{-M_\theta/s^2}^{M_\theta/s^2}
e^{\lambda t}\left\{
\frac{\sqrt n\,s^2}{u}
e^{-s^2\lambda^2/2}
\left\{
g_{G^+}(s^2\lambda,s)-g_{G^-}(s^2\lambda,s)
\right\}\right\}
\,d\lambda~.
\end{align*}
Therefore, $\psi(t,s)\Tilde{f}_H(t,s)$ can be viewed as a Laplace transform of a particular function. Before applying lemma \ref{lem:laplace_fact}, it remains to bound the $L^1$ norm of this function:
\begin{align*}
\frac{\sqrt n}{u}
\int_{\Theta}
e^{-\theta^2/(2s^2)}
\left|
g_{G^+}(\theta,s)-g_{G^-}(\theta,s)
\right|
\,d\theta
&\le
\frac{\sqrt n}{u}
\int_{\Theta}
e^{-\theta^2/(2s^2)}
\left\{
g_{G^+}(\theta,s)+g_{G^-}(\theta,s)
\right\}
\,d\theta
\\
&\qquad=
\frac{\sqrt n}{u}
\left\{
m_{G^+}(s)+m_{G^-}(s)
\right\}
\\
&\qquad\le
\frac{2C^2}{\underline\delta}\sqrt n~.
\end{align*}
Therefore, lemma \ref{lem:laplace_fact} and \eqref{eq:envelope_unif_in_s_bounds}, for fixed $0<\epsilon,\rho<1$:
\begin{align*}
|\psi(t,s)\widetilde f_H(t,s)|
&\lesssim
(\sqrt n)^\epsilon
\left(n^{1/(4B+2)}\right)^{1-\epsilon}
\exp\left\{
\frac{2\rho t^2}
{\underline s^2(1+\rho)}
\right\}
\\
&=
n^{(1-\epsilon)/(4B+2)+\epsilon/2}
\exp\left\{
\frac{2\rho t^2}
{\underline s^2(1+\rho)}
\right\}.
\end{align*}
Since $|t| \le \tau_n = D\sqrt{\log n}$ on $D_{\tau_n}$, we can further bound the above as follows:
\begin{equation*}
    |\psi(t,s)\widetilde f_H(t,s)|
\lesssim
n^{(1-\epsilon)/(4B+2)+\epsilon/2+2\rho D^2/[\underline{s}^2(1+\rho)]}
\end{equation*}
Choosing $\epsilon$ and $\rho$ sufficiently small such that $\frac{\epsilon}{2}
+
\frac{2\rho D^2}{\underline{s}^2(1+\rho)}
<
\frac{\eta}{2}$ yields:
\begin{equation*}
    |\psi(t,s)\widetilde f_H(t,s)|
\lesssim
n^{(1-\epsilon)/(4B+2)+\epsilon/2+2\rho D^2/[\underline{s}^2(1+\rho)]}
\le
n^{1/(4B+2)+\eta/2}
\end{equation*}
By \eqref{eq:inv_f_H_lower_bound}, $\|\Tilde{f}_H^{-1}\|_{L^\infty(D_{\tau_n})}\lesssim \exp(M_\theta\tau_n/\underline{s}^2)$. Since $\frac{M_\theta\tau_n}{\underline{s}^2\sqrt{\log n}} = o(1)$, $\exp(M_\theta\tau_n/\underline{s}^2) \le n^{\eta/2}$ for sufficiently large $n$. Multiplying these bounds for $\psi \Tilde{f}_H$ and $\Tilde{f}_H^{-1}$ yields \eqref{eq:main_envelope_main}. 

It remains to prove \eqref{eq:main_envelope_moments}, which is a corollary to \eqref{eq:main_envelope_main}. Truncation can only reduce the $L^2(P_H)$, meaning that, for $\psi$ in $\Psi_n^{tr}(H)$ and thus $P_H\psi^2 \le 1$. By the triangle inequality $\|\psi -P_G\psi\|_{L^\infty(D_{\tau_n})} \le 2\|\psi\|_{L^\infty(D_{\tau_n})} \le 2C_\eta n^{1/(4B+2)+\eta}$. \eqref{eq:density_fraction} gives:
\begin{equation*}
    P_G(\psi - P_G\psi)^2 \le P_G\psi^2 \le \frac{C^2}{c^2}P_H\psi^2 \le \frac{C^2}{c^2}~.
\end{equation*}
For $j=3,4$, multiply this bound by the bound $\|\psi -P_G\psi\|_{L^\infty(D_{\tau_n})}$, thus obtaining \eqref{eq:main_envelope_moments}.
\end{proof}
\begin{rem}
    \label{rem:recentered_bounding}
    Since $\|u-P_Gu\|_{L^\infty(D_{\tau_n})} \le 2 \|u\|_{L^\infty(D_{\tau_n})}$, we can choose $C_\eta$ such that, without loss of generality:
    \begin{equation*}
        \sup_{H,G\in\calG}\sup_{\psi \in \Psi_n(H)} \|\psi - P_G\psi\|_{L^\infty(D_{\tau_n})} \le C_\eta n^{1/(4B+2) + \eta}~.
    \end{equation*}
    Consequently, the bounds established in the preceding lemma \ref{lem:main_envelope} apply to the class containing centered elements of $\Psi_n(H)$, without loss of generality.
\end{rem}

\begin{lem}
    \label{lem:variance_bounds}
    Let assumption \ref{assu:smoothness} hold. For every $n\ge 2$, any $G,H$ in $\calG$, and any $\psi$ in $\Psi_n(H)$:
    \begin{equation}
    \label{eq:variance_bounds}
        \frac{c^2}{C^2} \le \var_G(\psi) \le P_G\psi^2 \le \frac{C^2}{c^2}~.
    \end{equation}
    Fix $D>0$ such that $D^2/(8\bar{s}^2) > 1$. For any $\varepsilon> 0$, there exists $n_0<\infty$ sufficiently large such that, for any $n\ge n_0$ and for any $G,H$ in $\calG$ and $\psi$ in $\Psi_n(H)$:
    \begin{equation}
    \label{eq:variance_tail_bounds}
        (1-\varepsilon)\frac{c^2}{C^2}\le \var_G (\psi\mathbb{1}(|t|\le\tau_n)) \le \frac{C^2}{c^2}~.
    \end{equation}
\end{lem}
\begin{proof}
    The statement \eqref{eq:variance_bounds} holds by \eqref{eq:density_fraction}, and the fact that $\var_G(\psi) = \inf_a P_G(\psi-a)^2$. 

    Proving \eqref{eq:variance_tail_bounds} is more involved. Since $\sigma \ge \underline{s}$ and since $|\theta| \le M_\theta$, it holds that for $x \ge 2M_\theta$:
    \begin{equation*}
        \sup_{G\in\calG}P_G(|t|>x) \le 2\exp\left(
-\frac{x^2}{8\bar{s}^2}
        \right)~.
    \end{equation*}
    By \eqref{eq:global_psi_bounds}, it holds that:
    \begin{align}
        \sup_{G,H\in\calG}\sup_{\psi\in\Psi_n(H)} P_G(\psi^2\mathbb{1}(|t|>\tau_n)) &\lesssim n\sup_{G\in\calG}P_G(|t|>\tau_n)\nonumber\\
        & \lesssim n \exp\left(
-\frac{\tau_n^2}{8\bar{s}^2}
        \right)\nonumber\\
        & = n^{1-D^2/(8\underline{s}^2)} = o(1)\label{eq:L2_bounds_tail_psi}
    \end{align}
    The final equality holds by the assumption that $D^2/(8\bar{s}^2) > 1$. The triangle inequality for the $L^2(P_G)$ norm yields:
    \begin{align*}
        \left|
\sqrt{\var_G(\psi)}
-
\sqrt{\var_G\!\left(\psi\mathbbm{1}\{|t|\le\tau_n\}\right)}
\right|&\le
\left\|
(\psi-P_G\psi)
-
\left(
\psi\mathbbm{1}\{|t|\le\tau_n\}
-
P_G\!\left[\psi\mathbbm{1}\{|t|\le\tau_n\}\right]
\right)
\right\|_{L^2(P_G)}
\\
&=
\sqrt{
\var_G\!\left(
\psi\mathbbm{1}\{|t|>\tau_n\}
\right)
}
\\
&\le
\left\|
\psi\mathbbm{1}\{|t|>\tau_n\}
\right\|_{L^2(P_G)}~.
    \end{align*}
    By \eqref{eq:L2_bounds_tail_psi}, this final term is $o(1)$ uniformly in $H,G$ in $\calG$ and in $\psi$ in $\Psi_n(H)$, yielding the lower bound in \eqref{eq:variance_tail_bounds}. To prove the upper bound, notice that:
    \begin{equation*}
        \var_G\!\left(\psi\mathbbm{1}\{|t|\le\tau_n\}\right)
\le
P_G\!\left[\psi^2\mathbbm{1}\{|t|\le\tau_n\}\right]
\le
P_G\psi^2
\le
\frac{C^2}{c^2}~,
    \end{equation*}
    thus concluding the proof.
\end{proof}

\subsection{Gaussian and Multiplier Bootstrap Coupling}

In this section, I use the results developed in the previous section \ref{app:Holder_smoothness} to obtain a coupling of the statistic $T_n(\Gamma(G))$ with the supremum of a Gaussian process by applying results due to \citealt{CCK_2014}. Extending arguments due to \citealt{CCK_2016} yields a coupling of the multiplier statistic $T_n^\xi(\Gamma(G))$ with the same Gaussian process supremum.

Before proceeding, first define the empirical and multiplier processes by:
\begin{equation*}
    \bbG_{n,G} v:= \sqrt{n}(P_n - P_G)v~,\quad \bbG_n^\xi:= \frac{1}{\sqrt{n}}\sum_{i=1}^n\xi_i (v(Z_i)-P_nv)~.
\end{equation*}
In what follows, I use the centered class:
\begin{equation*}
    \calF_{n,G} := \{u-P_Gu: u = \psi_{n,\gamma}\mathbb{1}(|t|\le\tau_n),~\gamma\in\Gamma(G)\}~.
\end{equation*}
Since subtracting $P_Gu$ does not change $\bbG_{n,G}u$ of $\bbG_n^\xi u$, the truncated versions of the statistics $T_n(\Gamma(G))$ and $T_n^\xi(\Gamma(G))$ are as follows:
\begin{equation*}
    T_{n,G}^{tr} := \sup_{f\in\calF}|\bbG_n f| \equiv \|\bbG_n\|_\calF~,\quad T_{n,G}^{\xi,tr} := \sup_{f\in\calF} |\bbG_n^\xi f|\equiv \|\bbG_n^\xi\|_\calF~.
\end{equation*}
Notice that $\Gamma(G)$, and thus $\calF_{n,G}$, is nonempty for any $G$ in $\calG$.

Fix $\vartheta \in (2/B , 2/9)$ as in lemma \ref{lem:psi_tr_entropy}. Furthermore, choose $\eta$ in lemma \ref{lem:main_envelope} such that $1/(4B+2) + \eta \le \vartheta/4$. Such $\eta$ exists, since $1/(4B+2) < 1/(2B) < \vartheta/4$. Lemma \ref{lem:main_envelope} yields an envelope $F_n := C_\eta n^{1/(4B+2)+\eta}$ for $\calF$. Lemma \ref{lem:variance_bounds} and \ref{lem:psi_tr_entropy} yield, respectively:
\begin{equation*}
    \sup_{f\in\calF} P_Gf^2 \le \frac{C^2}{c^2}~,\quad \sup_{G\in\calG}\sup_Q\log N(\varepsilon,\calF_{n,G},L^2(Q))\lesssim n^{\vartheta/2}\varepsilon^{-\vartheta}~.
\end{equation*}
Lemma \ref{lem:CCK_meas_verification} verifies the pointwise measurability of $\calF_{n,G}$. In what follows, we will hold the pilot fixed, and suppress dependence.

\subsubsection*{Measurability}

In order to demonstrate that $\calF_{n,G}$ is pointwise measurable in the sense of assumption A1 of \citealt{CCK_2014}, I demonstrate that $\Psi_n$ is separable with respect to a choice of metric $d_{loc}$ which metrizes the topology of uniform convergence on compacta. To this end, define $K_k := [-k,k]\times \mathcal{S}$ for $k\in\mathbb{N}$. For any two function $h_1$ and $h_2$ from $\mathbb{R}\times\mathcal{S}\mapsto\mathbb{R}$, define:
\begin{equation*}
    d_{loc}(h_1,h_2) := \sum_{k=1}^\infty 2^{-k}(1\wedge \|h_1 - h_2\|_{\infty, K_k})
\end{equation*}
\begin{lem}
\label{lem:CCK_meas_verification}
    For any $n$, $G$, and pilot estimate $\hat{G}_n$, the class $\Psi_n(H)$ is separable with respect to the metric $d_{loc}$, and $\calF_{n,G}$ is pointwise measurable in the sense of assumption A1 of \citealt{CCK_2014}.
\end{lem}
\begin{proof}
    First, notice that lemma A5.1 in \citealt{Kallenberg2021} provides that the space of continuous functions from $\mathbb{R}\times\mathcal{S}\to \mathbb{R}$, which I denote $\mathcal{C}(\mathbb{R}\times\mathcal{S},\mathbb{R})$ is separable under $d_{loc}$, which metrizes the topology of uniform convergence on compacta (See exercise 9 in \citealt{Lang1993}). By lemma \ref{lem:key_smoothness}, the class $\Psi_n(H)$ is contained in $\mathcal{C}(\mathbb{R}\times\mathcal{S},\mathbb{R})$, and is thus separable under $d_{loc}$ (see chapter 2 in \citealt{Dudley2002}). Consequently, $\{\psi_{n,\gamma}:\gamma \in \Gamma(G)\}\subseteq\Psi_n(H)$ is also separable under $d_{loc}$. Because uniform convergence of sequences of $\psi_{n,\gamma}$ on compacta implies the uniform convergence of the $\psi_{n,\gamma}\mathbb{1}(|t|\le\tau_n)$ on compacta, implying that  $\{\psi_{n,\gamma}\mathbb{1}(|t|\le\tau_n):\gamma \in \Gamma(G)\}\subseteq\Psi_n(H)$ is separable under $d_{loc}$. Finally, applying \eqref{eq:global_psi_bounds} and thus dominated convergence provides that $\calF_{n,G}$ is separable under $d_{loc}$.  Finally, uniform convergence on compacta implies pointwise convergence, concluding the proof.
\end{proof}

\begin{rem}
    The pointwise measurability of $\calF_{n,G}$ implies that the supremum statistic $\sup_{\gamma\in\Gamma(G)}\mathbb{G}_n\psi_{n,\gamma}$ is measurable, allowing me to neglect the measurability concerns discussed in \citealt{VdVandWellner1996}. I can therefore proceed using standard probabilities and expectations, rather than their outer versions, in the analysis that follows.
\end{rem}

\subsubsection*{Constructing the Couplings}

Before proceeding, I define the uniform entropy integral as considered in \citealt{CCK_2014}:
\begin{equation*}
    J_n(\delta) := \int_0^\delta \sup_Q \sqrt{1+\log N(\varepsilon C_\eta c_n,\calF_{n,G},L^2(Q))}d\varepsilon~,\quad c_n := n^{1/(4B+2)+\eta}.
\end{equation*}
A change of variables yields:
\begin{align}
    C_\eta c_n J_n(\delta)
&=
\int_0^{C_\eta c_n\delta}
\sup_Q
\sqrt{
1+\log N(\varepsilon,\mathcal F_{n,G},L^2(Q))
}\,d\varepsilon \nonumber\\
&\lesssim
n^{\vartheta/4}
\int_0^{C_\eta c_n\delta}
\varepsilon^{-\vartheta/2}\,d\varepsilon\nonumber \\
&\lesssim
n^{\vartheta/4}
(C_\eta c_n\delta)^{1-\vartheta/2} \asymp n^{\frac{1+B\vartheta}{4B+2}
+\eta\left(1-\frac{\vartheta}{2}\right)}
\delta^{1-\vartheta/2}~.
\label{eq:unif_entropy_integral_bounds}
\end{align}
Lemma 2.2 of \citealt{CCK_2014} provides a tight centered Gaussian process $\bbX_{n,G}$ on $\calF$ with the covariance function:
\begin{equation*}
    \bbE_G \bbX_{n,G}f\bbX_{n,G}g = P_Gfg~.
\end{equation*}
Allowing $Z_{n,G}^{tr}$ to be the corresponding Gaussian process supremum to $T_{n,G}^{tr}$, I take $Z_{n,G}^{tr} := \sup_{f\in\calF}|\bbX_{n,G}f|$. The following application of lemma 2.1 and theorem 2.1 of \citealt{CCK_2014} relates $Z_{n,G}^{tr}$ to $T_{n,G}^{tr}$.
\begin{lem}
    \label{lem:Gaussian_coupling}
    There exists a sequence $d_n = o(n^{-\vartheta /4})$ such that, for any $G,H$ in $\calG$, there exists a random variable $\Tilde{T}_{n,G}^{tr}$ satisfying:
    \begin{equation*}
        \Tilde{T}_{n,G}^{tr} \overset{d}{=} Z_{n,G}^{tr}~,\quad \sup_{G\in\calG}P_G\left(
|T_{n,G}^{tr} - \Tilde{T}_{n,G}^{tr}| > d_n
        \right) = o(1)~.
    \end{equation*}
\end{lem}
\begin{proof}
    The proof of this lemma proceeds by application of lemma 2.2 and theorem 2.1 of \citealt{CCK_2014}. For $0 < r \leq 1$, define:
\[
\calF_{n,G}(\varepsilon)
:=
\left\{
f-g:
f,g\in\calF_{n,G},\ 
\|f-g\|_{L^2(P_G)}<\varepsilon
\right\},
\qquad
\calF_{n,G}\cdot\calF_{n,G}
:=
\{fg:f,g\in\calF_{n,G}\}.
\]
The three quantities I will bound are:
\begin{gather*}
    \phi_n(\varepsilon)
:=
\mathbb E_{G}\|\mathbb G_n\|_{\calF_{n,G}(\varepsilon)}
\vee
\mathbb E_{G}\|X_{n,G}\|_{\calF_{n,G}(\varepsilon)}~,\\
A_n
:=
\mathbb E_{G}\|\mathbb G_n\|_{\calF_{n,G}\cdot\calF_{n,G}}~,
\qquad
B_n
:=
\mathbb E_{G}
\sup_{f\in\calF_{n,G}} P_n|f|^3~.
\end{gather*}
To do so, I apply lemma 2.2 of \citealt{CCK_2014}. To this end, I first verify assumption A4 of \citealt{CCK_2014} by applying their lemma 2.1. In particular, I first show that $J_n(1) <\infty$ for each $n$. I first consider values of $\varepsilon$
such that $\varepsilon C_\eta c_n\leq 1$. By the entropy bound in
\eqref{eq:entropy_bound_tr_class},
\begin{equation}
\label{eq:uniform_entropy_bounds_for_integral}
\sup_Q
\sqrt{
1+\log N(\varepsilon C_\eta c_n,\mathcal F_{n,G},L^2(Q))
}
\lesssim
n^{\vartheta/4}
(C_\eta c_n)^{-\vartheta/2}
\varepsilon^{-\vartheta/2}.
\end{equation}
The right-hand side is integrable near zero since $\vartheta/2<1$.
For values of $\varepsilon$ such that $\varepsilon C_\eta c_n>1$,
monotonicity of covering numbers gives
\begin{equation}
\label{eq:covering_num_monotonicity}
N(\varepsilon C_\eta c_n,\mathcal F_{n,G},L^2(Q))
\leq
N(1,\mathcal F_{n,G},L^2(Q)),
\end{equation}
and the entropy bound at radius one implies that the integrand is
bounded over the remainder of $[0,1]$. $J_n(1) \lesssim n^{\vartheta/4}$, since:
\begin{align*}
    J_n(1)
&=
\int_0^{1/(C_\eta c_n)}
\sup_Q
\sqrt{
1+\log N(\varepsilon C_\eta c_n,\mathcal F_{n,G},L^2(Q))
}\,d\varepsilon
\\
&\qquad
+
\int_{1/(C_\eta c_n)}^1
\sup_Q
\sqrt{
1+\log N(\varepsilon C_\eta c_n,\mathcal F_{n,G},L^2(Q))
}\,d\varepsilon
\\
&\lesssim
n^{\vartheta/4}(C_\eta c_n)^{-\vartheta/2}
\int_0^{1/(C_\eta c_n)}
\varepsilon^{-\vartheta/2}\,d\varepsilon
+
n^{\vartheta/4}
\\
&\lesssim
\frac{n^{\vartheta/4}}{C_\eta c_n}
+
n^{\vartheta/4}
\lesssim
n^{\vartheta/4}~.
\end{align*}
The first inequality holds by \eqref{eq:uniform_entropy_bounds_for_integral}, the monotonicity in \eqref{eq:covering_num_monotonicity}, and taking  \eqref{eq:uniform_entropy_bounds_for_integral} with $\varepsilon C_\eta c_n=1$. The second inequality follows by evaluating the integral in the preceding line. Thus $J_n(1)<\infty$ for each
$n$, meaning that assumption A4 of \citealt{CCK_2014} is satisfied. 

I first focus on $\phi_n(\varepsilon)$. Lemma 2.2 of \citealt{CCK_2014} yields:
\begin{align}
&\mathbb E_G\|\mathbb G_{n,G}\|_{\mathcal F_{n,G}(\varepsilon)}
\vee
\mathbb E_G\|X_{n,G}\|_{\mathcal F_{n,G}(\varepsilon)}
\nonumber\\
&\qquad
\lesssim
J_n(\varepsilon)\|C_\eta c_n\|_{P_G,2}
+
n^{-1/2}\varepsilon^{-2}J_n(\varepsilon)^2
\left\|\max_{1\leq i\leq n} C_\eta c_n\right\|_{P_G,2}\nonumber\\
&\qquad\lesssim
C_\eta c_n
J_n\left(\frac{\varepsilon}{C_\eta c_n}\right)
+
\frac{C_\eta c_n}{\sqrt n}
\left(\frac{\varepsilon}{C_\eta c_n}\right)^{-2}
J_n\left(\frac{\varepsilon}{C_\eta c_n}\right)^2
\nonumber\\
&\qquad\lesssim
n^{\vartheta/4}\varepsilon^{1-\vartheta/2}
+
\frac{C_\eta c_n}{\sqrt n}
\left(\frac{\varepsilon}{C_\eta c_n}\right)^{-2}
\left(
\frac{
n^{\vartheta/4}\varepsilon^{1-\vartheta/2}
}{
C_\eta c_n
}
\right)^2
\nonumber\\
&\qquad\lesssim
n^{\vartheta/4}\varepsilon^{1-\vartheta/2}
+
n^{-1/2+\vartheta/2}c_n\varepsilon^{-\vartheta}\label{eq:phi_n_bounds_CCK}~.
\end{align}
The first inequality holds by application of lemma 2.2 of \citealt{CCK_2014}. The second inequality holds by application of lemma \ref{lem:main_envelope}. The third inequality holds by \eqref{eq:unif_entropy_integral_bounds}. The final inequality holds by simplifying terms, and by absorbing $C_\eta$ into the implicit constant.

I now bound $A_n$. To this end, notice that lemma \ref{lem:main_envelope} implies:
\begin{equation*}
    \sup_{f\in\calF_{n,G}} P_Gf^4 \le C_\eta^2 c_n^2\sup_{f\in\calF_{n,G}} P_Gf^2 \lesssim c_n^2~.
\end{equation*}
Similarly, $\|f\|_{L^4(P_G)}\le F_n^{1/2}\|f\|_{L^2(P_G)}^{1/2}  \lesssim c_n^{1/2}$. Since $\|F_n\|_{L^4(P_G)} = C_\eta c_n$, it holds that $\sup_{f\in\calF_{n,G}}\frac{\|f\|_{L^4(P_G)}}{\|F_n\|_{L^4(P_G)}}\lesssim c_n^{-1/2}$. Similarly, $\sup_{f\in\calF_{n,G}}\frac{\|f\|_{L^4(P_G)}}{\|F_n\|_{L^4(P_G)}}\lesssim c_n^{-2/3}$. As a result, I choose $\delta_3 \asymp c_n^{-2/3}$ and $\delta_4\asymp c_n^{-1/2}$. I next bound $J_n(\delta_4^2)$. First consider values of $\varepsilon$
such that $\varepsilon C_\eta c_n\leq 1$. By
\eqref{eq:entropy_bound_tr_class},
\begin{equation*}
\sup_Q
\sqrt{
1+\log N(\varepsilon C_\eta c_n,\mathcal F_{n,G},L^2(Q))
}
\lesssim
n^{\vartheta/4}
(C_\eta c_n)^{-\vartheta/2}
\varepsilon^{-\vartheta/2}~.
\end{equation*}
Therefore,
\begin{align*}
\int_0^{1/(C_\eta c_n)}
\sup_Q
\sqrt{
1+\log N(\varepsilon C_\eta c_n,\mathcal F_{n,G},L^2(Q))
}
\,d\varepsilon
&\lesssim
n^{\vartheta/4}
(C_\eta c_n)^{-\vartheta/2}
\int_0^{1/(C_\eta c_n)}
\varepsilon^{-\vartheta/2}\,d\varepsilon \\
&\lesssim
\frac{n^{\vartheta/4}}{c_n}~.
\end{align*}
For values of $\varepsilon$ such that $\varepsilon C_\eta c_n>1$,
monotonicity of covering numbers gives
\begin{equation*}
\log N(\varepsilon C_\eta c_n,\mathcal F_{n,G},L^2(Q))
\leq
\log N(1,\mathcal F_{n,G},L^2(Q))\lesssim n^{\vartheta/2}~.
\end{equation*}
Therefore, the integrand
is bounded by a constant multiple of $n^{\vartheta/4}$. Since
$\delta_4^2\asymp c_n^{-1}$, it follows that
\begin{equation*}
    \int_{1/(C_\eta c_n)}^{\delta_4^2}
\sup_Q
\sqrt{
1+\log N(\varepsilon C_\eta c_n,\mathcal F_{n,G},L^2(Q))
}
\,d\varepsilon
\lesssim
\frac{n^{\vartheta/4}}{c_n}~.
\end{equation*}
Combining the two bounds yields
\begin{equation}
\label{eq:J_n_d_4_rate}
    J_n(\delta_4^2)
\lesssim
\frac{n^{\vartheta/4}}{c_n}~.
\end{equation}
Applying lemma 2.2 of \citealt{CCK_2014} yields:
\begin{align*}
\mathbb E_G
\|\mathbb G_{n,G}\|_{\mathcal F_{n,G}\cdot\mathcal F_{n,G}}
&\lesssim
J_n(\delta_4^2)
\|C_\eta c_n\|_{L^4(P_G)}^2
+
\frac{
J_n(\delta_4^2)^2
}{
\sqrt n\,\delta_4^4
}
\left\|
\max_{1\leq i\leq n} C_\eta c_n
\right\|_{L^4(P_G)}^2
\\
&\lesssim
c_n^2J_n(\delta_4^2)
+
\frac{
c_n^2J_n(\delta_4^2)^2
}{
\sqrt n\,\delta_4^4
}
\\
&\lesssim
n^{\vartheta/4}c_n
+
n^{-1/2+\vartheta/2}c_n^2
\\
&\lesssim
n^{\vartheta/2}~.
\end{align*}
The first inequality is an application of lemma 2.2 of \citealt{CCK_2014}. The second inequality holds simply by noting that $C_\eta c_n$ is a constant, and thus $\|C_\eta c_n\|_{L^4(P_G)}$ is a constant. The third inequality holds by \eqref{eq:J_n_d_4_rate}. The final inequality holds since $1/(4B+2) + \eta < \vartheta/4$, and thus $c_n \lesssim n^{\vartheta/4}$. 

It finally remains to bound $B_n$. My bounds on $B_n$ follow from the bound on $A_n$ established above. In particular, notice that $|f|^3 \le C_\eta c_nf^2$ and $f^2$ lies in $\calF_{n,G}\cdot\calF_{n,G}$.
\begin{align}
\mathbb E_G
\sup_{f\in\mathcal F_{n,G}} P_n|f|^3
&\leq
C_\eta c_n
\mathbb E_G
\sup_{f\in\mathcal F_{n,G}} P_nf^2\nonumber
\\
&\leq
C_\eta c_n
\left\{
\sup_{f\in\mathcal F_{n,G}} P_Gf^2
+
n^{-1/2}
\mathbb E_G
\|\mathbb G_{n,G}\|_{\mathcal F_{n,G}\cdot\mathcal F_{n,G}}
\right\}\nonumber
\\
&\lesssim
c_n\left(1+n^{-1/2+\vartheta/2}\right)
\lesssim
c_n\label{eq:kappa_useful}~.
\end{align}
The first inequality holds by the above bound on $|f|^3$. The second inequality holds, since $P_nf^2 = P_Gf^2 + n^{-1/2}\bbG_{n,G}f^2$, and by the fact that $f^2$ lies in $\calF_{n,G}\cdot\calF_{n,G}$. The third inequality holds by applying \eqref{eq:density_fraction}, and the above bounds on $\mathbb E_G
\|\mathbb G_{n,G}\|_{\mathcal F_{n,G}\cdot\mathcal F_{n,G}}$. The final inequality holds since $\vartheta < 1$. 

With the above bounds established, it remains to apply theorem 2.1 and corollary 2.1 of \citealt{CCK_2014}. To this end, choose some $\rho$ satisfying:
\begin{equation*}
    \frac{\vartheta}{2-\vartheta} < \rho < \frac{1 - 4\vartheta}{4\vartheta} \wedge \frac{1-\vartheta}{4}~.
\end{equation*}
Because $0 < \vartheta < 2/9 $, such a $\rho$ exists. Let $\varepsilon$ and $\gamma$, as used in the statement of theorem 2.1 and corollary 2.1 of \citealt{CCK_2014}, be given by $\frac{n^{-\rho}}{C_\eta C_n}$ and $\frac{1}{\log n}$, respectively. Finally, let $\kappa$ be such that $\mathbb E_G
\sup_{f\in\mathcal F_{n,G}} P_n|f|^3 \le \kappa^3$. By \eqref{eq:kappa_useful}, one can take $\kappa^3 \asymp c_n$. Following \citealt{CCK_2014}, I define the following:
\[
\begin{aligned}
\Delta_n
&=
\phi_n\left(n^{-\rho}\right)
+
(\log n)^{1/4}n^{-\rho}
+
n^{-1/2}(\log n)^{1/4}C_\eta c_n
+
n^{-1/2}(\log n)^{1/2}C_\eta c_n
\\
&\quad
+
n^{-1/4}(\log n)^{1/2}
\left(
\mathbb E_G
\|\mathbb G_{n,G}\|_{\mathcal F_{n,G}\cdot\mathcal F_{n,G}}
\right)^{1/2}
\left[
\log\left\{
N\left(
n^{-\rho},
\mathcal F_{n,G},
L^2(P_G)
\right)
\vee n
\right\}
\right]^{1/2}
\\
&\quad
+
n^{-1/6}(\log n)^{1/3}
\left(
\mathbb E_G
\sup_{f\in\mathcal F_{n,G}}
P_n|f|^3
\right)^{1/3}
\left[
\log\left\{
N\left(
n^{-\rho},
\mathcal F_{n,G},
L^2(P_G)
\right)
\vee n
\right\}
\right]^{2/3}.
\end{aligned}
\]
\[
\begin{aligned}
\delta_n
=
\frac{1}{4}P_G
\Bigg[
&
\left(
\frac{C_\eta c_n}{\kappa}
\right)^3
\\
&\times
\mathbf 1\Bigg\{
\frac{C_\eta c_n}{\kappa}
>
c(\log n)^{1/3}n^{1/3}
\\
&\hspace{4cm}\times
\left[
\log\left\{
N\left(
n^{-\rho},
\mathcal F_{n,G},
L^2(P_G)
\right)
\vee n
\right\}
\right]^{-1/3}
\Bigg\}
\Bigg].
\end{aligned}
\]
I bound these objects using the previously-established bounds on $\phi_n$, $A_n$ and $B_n$. I first analyze each term in the above expression for $\Delta_n$. 
\begin{align*}
\phi_n(n^{-\rho})
&\lesssim
n^{\vartheta/4}(n^{-\rho})^{1-\vartheta/2}
+
n^{-1/2+\vartheta/2}c_n(n^{-\rho})^{-\vartheta}
\\
&=
n^{\vartheta/4-\rho(1-\vartheta/2)}
+
n^{-1/2+\vartheta/2+\rho\vartheta}c_n
\\
&\lesssim
n^{\vartheta/4-\rho(1-\vartheta/2)}
+
n^{-1/2+3\vartheta/4+\rho\vartheta}
\\
&=
o(n^{-\vartheta/4})~.
\end{align*}
The first inequality holds by \eqref{eq:phi_n_bounds_CCK}. The second inequality holds since $c_n\lesssim n^{\vartheta/4}$, since $1/(4B+2) + \eta \le \vartheta / 4$. The final equality follows since $\rho>\vartheta/(2-\vartheta)$ implies $\vartheta/2-\rho(1-\vartheta/2)<0$, while $\rho<(1-4\vartheta)/(4\vartheta)$ implies $-1/2+\vartheta+\rho\vartheta<-1/4<0$.
\begin{align*}
(\log n)^{1/4}n^{-\rho}
&=
n^{-\vartheta/4}
n^{-(\rho-\vartheta/4)}
(\log n)^{1/4}
\\
&=
o(n^{-\vartheta/4})~.
\end{align*}
The first equality follows by factoring out $n^{-\vartheta/4}$. The final equality follows since $\rho>\vartheta/(2-\vartheta)>\vartheta/4$.
\begin{align*}
n^{-1/2}(\log n)^{1/4}C_\eta c_n
&\lesssim
n^{-1/2+\vartheta/4}(\log n)^{1/4}
\\
&=
n^{-\vartheta/4}
n^{-1/2+\vartheta/2}
(\log n)^{1/4}
\\
&=
o(n^{-\vartheta/4})~.
\end{align*}
The first inequality holds since $c_n\lesssim n^{\vartheta/4}$. The first equality follows by factoring out $n^{-\vartheta/4}$. The final equality follows since $\vartheta<2/9$ as in lemma \ref{lem:psi_tr_entropy}, and therefore $-1/2+\vartheta/2<0$.
\begin{align*}
n^{-1/2}(\log n)^{1/2}C_\eta c_n
&\lesssim
n^{-1/2+\vartheta/4}(\log n)^{1/2}
\\
&=
n^{-\vartheta/4}
n^{-1/2+\vartheta/2}
(\log n)^{1/2}
\\
&=
o(n^{-\vartheta/4})~.
\end{align*}
The first inequality holds since $c_n\lesssim n^{\vartheta/4}$. The first equality follows by factoring out $n^{-\vartheta/4}$. The final equality follows since $\vartheta<2/9$ as in lemma \ref{lem:psi_tr_entropy}, and therefore $-1/2+\vartheta/2<0$.
\begin{align*}
&n^{-1/4}(\log n)^{1/2}
\left(
\mathbb E_G
\|\mathbb G_{n,G}\|_{\mathcal F_{n,G}\cdot\mathcal F_{n,G}}
\right)^{1/2}
\left[
\log\left\{
N\left(
n^{-\rho},
\mathcal F_{n,G},
L^2(P_G)
\right)
\vee n
\right\}
\right]^{1/2}
\\
&\qquad\lesssim
n^{-1/4+\vartheta/4}(\log n)^{1/2}
\left[
\log\left\{
N\left(
n^{-\rho},
\mathcal F_{n,G},
L^2(P_G)
\right)
\vee n
\right\}
\right]^{1/2}
\\
&\qquad\lesssim
n^{-1/4+\vartheta/4+\frac{\vartheta}{2}(\rho+1/2)}
(\log n)^{1/2}
\\
&\qquad=
n^{-1/4+\vartheta/2+\rho\vartheta/2}
(\log n)^{1/2}
\\
&\qquad=
n^{-\vartheta/4}
n^{-1/4+3\vartheta/4+\rho\vartheta/2}
(\log n)^{1/2}
\\
&\qquad=
o(n^{-\vartheta/4})~.
\end{align*}
The first inequality holds by the bound on $\mathbb E_G\|\mathbb G_{n,G}\|_{\mathcal F_{n,G}\cdot\mathcal F_{n,G}}$ established above. The second inequality holds by \eqref{eq:subclass_centered_bound}. The first equality follows by simplifying powers of $n$. The second equality follows by factoring out $n^{-\vartheta/4}$. The final equality follows since $\rho<(1-4\vartheta)/(4\vartheta)$ implies $\rho\vartheta/2<1/8-\vartheta/2$, and hence $-1/4+3\vartheta/4+\rho\vartheta/2<-1/8+\vartheta/4<0$, where the final inequality uses $\vartheta<2/9<1/2$.
\begin{align*}
&n^{-1/6}(\log n)^{1/3}
\left(
\mathbb E_G
\sup_{f\in\mathcal F_{n,G}}P_n|f|^3
\right)^{1/3}
\left[
\log\left\{
N\left(
n^{-\rho},
\mathcal F_{n,G},
L^2(P_G)
\right)
\vee n
\right\}
\right]^{2/3}
\\
&\qquad\lesssim
n^{-1/6}(\log n)^{1/3}
c_n^{1/3}
\left[
\log\left\{
N\left(
n^{-\rho},
\mathcal F_{n,G},
L^2(P_G)
\right)
\vee n
\right\}
\right]^{2/3}
\\
&\qquad\lesssim
n^{-1/6+\vartheta/12}(\log n)^{1/3}
\left[
\log\left\{
N\left(
n^{-\rho},
\mathcal F_{n,G},
L^2(P_G)
\right)
\vee n
\right\}
\right]^{2/3}
\\
&\qquad\lesssim
n^{-1/6+\vartheta/12+\frac{2\vartheta}{3}(\rho+1/2)}
(\log n)^{1/3}
\\
&\qquad=
n^{-1/6+5\vartheta/12+2\rho\vartheta/3}
(\log n)^{1/3}
\\
&\qquad=
n^{-\vartheta/4}
n^{-1/6+\frac{2\vartheta}{3}(\rho+1)}
(\log n)^{1/3}
\\
&\qquad=
o(n^{-\vartheta/4})~.
\end{align*}
The first inequality holds by \eqref{eq:kappa_useful}. The second inequality holds since $c_n^{1/3}\lesssim n^{\vartheta/12}$. The third inequality holds by \eqref{eq:subclass_centered_bound}. The first equality follows by simplifying powers of $n$. The second equality follows by factoring out $n^{-\vartheta/4}$. The final equality follows since $\rho<(1-4\vartheta)/(4\vartheta)=1/(4\vartheta)-1$ implies $\vartheta(\rho+1)<1/4$, and therefore $-1/6+(2\vartheta/3)(\rho+1)<0$.

I next analyze $\delta_n$.
\begin{align*}
\delta_n
&=
\frac{1}{4}P_G\Bigg[
\left(
\frac{C_\eta c_n}{\kappa}
\right)^3
\mathbf 1\Bigg\{
\frac{C_\eta c_n}{\kappa}
>
c(\log n)^{1/3}n^{1/3}
\left[
\log\left\{
N\left(
n^{-\rho},
\mathcal F_{n,G},
L^2(P_G)
\right)
\vee n
\right\}
\right]^{-1/3}
\Bigg\}
\Bigg]
\\
&\lesssim
c_n^2
\mathbf 1\Bigg\{
c_n^{2/3}
>
c(\log n)^{1/3}n^{1/3}
\left[
\log\left\{
N\left(
n^{-\rho},
\mathcal F_{n,G},
L^2(P_G)
\right)
\vee n
\right\}
\right]^{-1/3}
\Bigg\}
\\
&\lesssim
c_n^2
\mathbf 1\left\{
c_n^{2/3}
>
c(\log n)^{1/3}
n^{\{1-\vartheta(\rho+1/2)\}/3}
\right\}
\\
&\lesssim
c_n^2
\mathbf 1\left\{
n^{\vartheta/6}
>
c(\log n)^{1/3}
n^{\{1-\vartheta(\rho+1/2)\}/3}
\right\}
\\
&=
0
\end{align*}
The first inequality holds since $\kappa^3\asymp C_\kappa c_n$, so that
$(C_\eta c_n/\kappa)^3\lesssim c_n^2$ and
$C_\eta c_n/\kappa\lesssim c_n^{2/3}$.
The second inequality holds by \eqref{eq:subclass_centered_bound}, which
implies
$\log\{N(n^{-\rho},\mathcal F_{n,G},L^2(P_G))\vee n\}
\lesssim n^{\vartheta(\rho+1/2)}$.
The third inequality holds since $c_n\lesssim n^{\vartheta/4}$, and
therefore $c_n^{2/3}\lesssim n^{\vartheta/6}$.
The final equality holds for all sufficiently large $n$, since
$\rho<(1-4\vartheta)/(4\vartheta)$ implies
$\vartheta(\rho+1)<1/4<1$, and hence
$\vartheta/6<\{1-\vartheta(\rho+1/2)\}/3$.

Corollary 2.1 of \citet{CCK_2014}, applied with $q=4$, implies that
there exists a random variable $\widetilde T_{n,G}^{tr}$ satisfying
$\widetilde T_{n,G}^{tr}\overset{d}{=}Z_{n,G}^{tr}$ such that:
\begin{equation*}
P_G\left(
\left|
T_{n,G}^{tr}-\widetilde T_{n,G}^{tr}
\right|
>
K(4)\Delta_n
\right)
\leq
\gamma\{1+\delta_n\}
+
\frac{C\log n}{n}~,
\end{equation*}
where $K(4)$ is a constant that depends only on the space $L^q(P_G)$ containing $F_n$. The result follows since $\gamma = o(1)$, $1+\delta = 1 + o(1)$, and $\frac{C\log n}{n} = o(1)$.
\end{proof}

To show how our test statistic of interest, $T_{n,G}^{tr}$ compares to the multiplier bootstrap test statistic, $T_{n,G}^{\xi,tr}$, the following lemma \ref{lem:multiplier_coupling} establishes a further coupling of $T_{n,G}^{\xi,tr}$ to a Gaussian process supremum equal in distribution to $Z_{n,G}^{tr}$ considered in the previous lemma \ref{lem:Gaussian_coupling}. I extend the proof of theorem 2.2 of \citealt{CCK_2016} to account for the fact that $\calF_{n,G}$ is not VC-type, thus relaxing assumption B in \citealt{CCK_2016} to allow for H\"older-type entropy bounds. 

\begin{lem}
    \label{lem:multiplier_coupling}
    Let assumptions \ref{assu:smoothness}, \ref{assu:modulus_binding}, and \ref{assu:pilot_independence} hold. Let $\vartheta$ and $\rho$ be as defined in the proof of lemma \ref{lem:Gaussian_coupling}. There exist positive constants $K_1,K_2 >0$ and $n_0<\infty$ such that for every $G,H \in \calG$ and every $\beta \in (0,1)$, there exists for each fixed $H$ a random variable $\Tilde{T}_{n,G}^{\xi,tr}$ satisfying:
    \begin{equation*}
        \Tilde{T}_{n,G}^{\xi,tr} | Z_1,\ldots,Z_n \overset{d}{=} Z_{n,G}^{tr}~a.s.,
    \end{equation*}
    and:
    \begin{equation*}
        P_{n,G,\xi}\left(
|T_{n,G}^{\xi,tr} - \Tilde{T}_{n,G}^{\xi,tr}| > K_1\Delta_n^\xi(\beta)
        \right) \le K_2(\beta + n^{-1})~,
    \end{equation*}
    where $\Delta_n^\xi(\beta)$ is given by:
    \begin{equation}
    \label{eq:delta_n_xi}
        \Delta_n^\xi(\beta)
:=
\frac{
n^{\vartheta/4-\rho(1-\vartheta/2)}
+
n^{-\rho}\sqrt{\log n}
}{
\beta^{1/4}
}
+
\frac{
n^{-1/4+\vartheta/2+\rho\vartheta/2}
}{
\beta^{5/4}
}~.
    \end{equation}

\end{lem}
\begin{proof}
    Fix $G$ and a realization of $H$. The proof will follow the arguments behind the proof of theorem 2.2 of \citealt{CCK_2016}, though I will apply the entropy bounds obtained in lemma \ref{lem:psi_tr_entropy} in lieu of using the entropy bounds for VC-type function classes applied in \citealt{CCK_2016}. I first consider a discretization of the empirical and multiplier processes, and show how to bound the corresponding discretization errors of each by analyzing the equicontinuity properties of the empirical and multiplier processes. I then obtain a coupling for the discretized class. Combining results and applying the conditional Strassen theorem provides the desired coupling. 

    \textit{Step 1: Discretization of Empirical and Multiplier Processes.} Let $\mathcal{N}_n$ be a finite $n^{-\rho}$ net of $\calF_{n,G}$ such that each element in $\calF_{n,G}$ lies in $L^2(P_G)$ distance strictly less than $n^{-\rho}$ from an element in $\mathcal{N}_n$. By lemma \ref{lem:psi_tr_entropy}:
    \begin{equation*}
        \log (2\#\mathcal{N}_n) \lesssim n^{\vartheta/2}n^{\rho\vartheta} = n^{\vartheta(\rho + 1/2)}~.
    \end{equation*}
    My bounds on $\rho$ imply that $\vartheta(\rho + 1/2) < 1/4$. As in the proof of theorem 2.1 of \citealt{CCK_2016}, I obtain the following inequalities:
    \begin{align*}
        0 &\le T_{n,G}^{\xi,tr} - \sup_{f\in\mathcal{N}_n} |\bbG_n^\xi f|\le \|\bbG_n^\xi\|_{\calF_{n,G}(n^{-\rho})}~,\\
        0 &\le Z_{n,G}^{tr} - \sup_{f\in\mathcal{N}_n}| \mathbb{X}_{n,G}f| \le \|\mathbb{X}_{n,G}\|_{\calF_{n,G}(n^{-\rho})}~.
    \end{align*}
    It therefore remains to bound the objects $\|\bbG_n^\xi\|_{\calF_{n,G}(n^{-\rho})}$ and $\|\mathbb{X}_{n,G}\|_{\calF_{n,G}(n^{-\rho})}$.
    
    \textit{Step 2: Gaussian Equicontinuity.} For any probability measure $Q$:
    \begin{equation*}
    \label{eq:bootstrap_useful_entropy_facts}
        \log N(\varepsilon, \calF_{n,G}(n^{-\rho}), L^2(Q)) \le 2\log N(\varepsilon/4, \calF_{n,G}, L^2(Q)) \lesssim n^{\vartheta / 2} \varepsilon^{-\vartheta}~.
    \end{equation*}
    The first inequality holds by applying the triangle inequality to obtain balls over pairwise differences. For sufficiently large $n$, $2n^{-\rho}<1$, and thus applying Dudley's inequality (corollary 2.2.8 in \citealt{VdVandWellner1996}) implies:
    \begin{equation*}
        \bbE \|\mathbb{X}_{n,G}\|_{\calF_{n,G}(n^{-\rho})}\lesssim n^{\vartheta/4}n^{-\rho(1-\vartheta/2)}~.
    \end{equation*}
    By the Borell-Tsirelson-Ibragimov-Sudakov inequality (see proposition A.2.1 of \citealt{VdVandWellner1996}), for some constant $C>0$:
    \begin{equation}
    \label{eq:step_2_equic}
        P\left(
\|\mathbb{X}_{n,G}\|_{\calF_{n,G}(n^{-\rho})} > C\left\{
n^{\vartheta/4}n^{-\rho(1-\vartheta/2)} + n^{-\rho}\sqrt{\log n}
\right\}
        \right)\le \frac{2}{n}~.
    \end{equation}

    \textit{Step 3: Multiplier Process Equicontinuity.} The canonical semimetric for $\mathbb{G}_n^\xi$, conditional on $Z_1,\ldots,Z_n$, satisfies:
    \begin{equation*}
       d_{\xi,n}(h,h')^2:=\bbE_\xi((\bbG_n^\xi h-\bbG_n^\xi h')^2)= P_n(h-h')^2- (P_n(h-h'))^2\le  P_n(h-h')^2~.
    \end{equation*}
    Define:
    \begin{equation*}
        \sigma_n^2 := \sup_{h\in\calF_{n,G}(n^{-\rho})} P_nh^2~.
    \end{equation*}
    Since $h \in \calF_{n,G}(n^{-\rho})$ can be written as $f-g$ for $f$ and $g$ in $\calF_{n,G}$ satisfying $\|f-g\|_{L^2(P_G)} \le n^{-\rho}$, I write:
    \begin{equation*}
        \sigma_n^2 \le n^{-2\rho} +\frac{4}{\sqrt{n}}\|\bbG_n\|_{\calF_{n,G}\cdot\calF_{n,G}}~.
    \end{equation*}
    By the arguments of lemma \ref{lem:Gaussian_coupling}, $\bbE_G\|\bbG_n\|_{\calF_{n,G}\cdot\calF_{n,G}} \lesssim n^{\vartheta/2}$. Applying lemma 6.1 of \citealt{CCK_2016} yields:
    \begin{align*}
\|\bbG_{n,G}\|_{\calF_{n,G}\cdot\calF_{n,G}}
&\le
(1+\beta^{-1/2})
\bbE_G\|\bbG_{n,G}\|_{\calF_{n,G}\cdot\calF_{n,G}}
\\
&\qquad
+
K_2\left[
\left(
\frac{F_nC}{c}
+
\frac{F_n^2}{\sqrt n}
\right)\beta^{-1/2}
+
\beta^{1/2}
\frac{F_n^2}{\sqrt n}\beta^{-1}
\right]
\\
&\lesssim
\beta^{-1/2}
\left\{
\bbE_G\|\bbG_{n,G}\|_{\calF_{n,G}\cdot\calF_{n,G}}
+
F_n
+
\frac{F_n^2}{\sqrt n}
\right\}
\\
&\lesssim
\beta^{-1/2}
\left\{
n^{\vartheta/2}
+
n^{\vartheta/4}
+
n^{-1/2+\vartheta/2}
\right\}
\\
&\lesssim
\beta^{-1/2}n^{\vartheta/2}~,
\end{align*}
with probability at least $1-\beta$. The first inequality holds by application of \citealt{CCK_2016}. The second inequality is algebraic, noting that $1+\beta^{-1/2} \le 2\beta^{-1/2}$. The third inequality holds by the bounds on $A_n$ stated in the proof of lemma \ref{lem:Gaussian_coupling} and by the bounds on $F_n$ provided in lemma \ref{lem:main_envelope}. The final inequality is algebraic. Defining the event $\mathcal{E}_n$ as follows, it holds that:
    \begin{equation*}
        \mathcal{E}_n :=\left\{
\sigma_n^2 \le s_n^2
        \right\}~,\quad s_n^2 := n^{-2\rho} + C\beta^{-1/2}n^{-(1-\vartheta)/2}~, \quad \sup_{G,H} P_G(\mathcal{E}_n^c) \le \beta~.
    \end{equation*}
    Since $\rho < (1-\vartheta)/4$ and $\beta <1$, $s_n \lesssim \beta^{-1/4}n^{-\rho}$. 
    
    On $\mathcal{E}_n$, and conditional on $Z_1,\ldots,Z_n$, $\bbG_n^\xi$ is a centered Gaussian process, with $d_{\xi,n}(h,h') \le 2s_n$ for any $h$ and $h'$ in $\calF_{n,G}(n^{-\rho})$. Applying Dudley's inequality in the usual way used throughout the proof provides:
    \begin{align*}
\mathbb E_\xi\!\left(
\|\mathbb G_n^\xi\|_{\mathcal F_{n,G}(n^{-\rho})}
\right)
&\lesssim
n^{\vartheta/4}
\left\{
\int_0^{(2s_n)\wedge1}
u^{-\vartheta/2}\,du
+
(2s_n-1)_+
\right\}
\\
&\lesssim
n^{\vartheta/4}
\left(
s_n^{1-\vartheta/2}
+
s_n
\right)
\\
&\lesssim
\beta^{-1/4}
n^{\vartheta/4}
n^{-\rho(1-\vartheta/2)}.
\end{align*}
The Borell inequality (proposition A.2.1 in \citealt{VdVandWellner1996}) yields:
    \begin{equation}
    \label{eq:step_3_equic}
        P_\xi\left(
\|\bbG_n^\xi\|_{\calF_{n,G}(n^{-\rho})} > C\beta^{-1/4}\left\{
n^{\vartheta/4} n^{-\rho(1-\vartheta/2)} + n^{-\rho} \sqrt{\log n}
\right\}
        \right) \le \frac{2}{n}~.
    \end{equation}

    \textit{Step 4: Coupling for the Discretized Class.} I apply theorem 2.2 of \citealt{CCK_2016} to $\mathcal N_n \cup \mathcal N_n$. By lemma \ref{lem:main_envelope}, there exists a constant envelope $C_\eta n^{\vartheta /8}$ for this class, since a choice of $\eta >0$ can be taken to satisfy $1/(4B+2)+\eta \le\vartheta/8$. 

    For any probability measure $Q$ and any $0<\varepsilon \le 1$, notice that:
    \begin{equation*}
        N(\varepsilon C_\eta n^{\vartheta /8},\mathcal N_n\cup (-\mathcal N_n),L^2(Q)) \le 2\#\mathcal N_n\le \frac{e \vee (2\#\mathcal N_n)}{\varepsilon}~.
    \end{equation*}
    Therefore, $\mathcal N_n\cup (-\mathcal N_n)$ is VC-type, as defined in \citealt{CCK_2016}. Therefore, assumption B in \citealt{CCK_2016} holds. Similarly, assumption A in \citealt{CCK_2016} is immediate. Assumption C follows from lemma \ref{lem:main_envelope}, and with $q=8$ noting that $\|C_\eta n^{\vartheta/8}\|_{L^8(P_G)}$. 

    It remains to apply theorem 2.2 of \citealt{CCK_2016}. To this end, using the notation from \citealt{CCK_2016}, set $\gamma = \beta$, and let:
    \begin{equation*}
        \zeta_n := n^{-2}~,\quad K_n
:=
\log n
\vee
\log\!\left(
\frac{(e\vee 2\#\mathcal N_n)cC_\eta n^{\vartheta/8}}{C}
\right)
\lesssim
n^{\vartheta(\rho+1/2)}~.
    \end{equation*}
    The drift functional in \citealt{CCK_2016} is zero, and thus $N_B(\zeta_n) =1$. Moreover, $K_n\le n$ for sufficiently large $n$. Let $C_3$ and $C_4$ be the constants in theorem 2.2 of \citealt{CCK_2016}. Application of this theorem that there exists $\Tilde Z_{n,G}^{\mathcal N}$ such that:
    \[
    \Tilde Z_{n,G}^{\mathcal N}
    \overset{d}{=}
    \sup_{f\in\mathcal N_n}
    |\bbX_{n,G}f|
\]
    satisfying:
    \begin{equation}
    \label{eq:finite_net_approximation}
        P_{n,G,\xi}\!\left(
\left|
\sup_{f\in\mathcal N_n}|\mathbb G_n^\xi f|
-
\Tilde Z_{n,G}^{\mathcal N}
\right|
>
C_3\!\left[
n^{-2}
+
\frac{C_\eta n^{\vartheta/8}K_n}{\beta^{1+1/8}n^{1/2-1/8}}
+
\frac{\left\{C_\eta n^{\vartheta/8}\frac{C}{c}K_n^{3/2}\right\}^{1/2}}{\beta^{1+1/8}n^{1/4}}
\right]
\right)
\le
C_4(\beta+n^{-1})~.
    \end{equation}
    It remains to simplify the rates on the left hand side of the above. To this end, write:
    \begin{align*}
&n^{-2}
+
\frac{
C_\eta n^{\vartheta/8}K_n
}{
\beta^{1+1/8}n^{1/2-1/8}
}
+
\frac{
\left\{
C_\eta n^{\vartheta/8}
\frac{C}{c}
K_n^{3/2}
\right\}^{1/2}
}{
\beta^{1+1/8}n^{1/4}
}
\\
&\qquad\lesssim
n^{-2}
+
\beta^{-9/8}
\left\{
n^{-3/8+\vartheta/8+\vartheta(\rho+1/2)}
+
n^{-1/4+\vartheta/16+
\frac34\vartheta(\rho+1/2)}
\right\}
\\
&\qquad=
n^{-2}
+
\beta^{-9/8}
\left\{
n^{-3/8+5\vartheta/8+\rho\vartheta}
+
n^{-1/4+7\vartheta/16+3\rho\vartheta/4}
\right\}
\\
&\qquad\lesssim
\beta^{-5/4}
n^{-1/4+\vartheta/2+\rho\vartheta/2}~.
\end{align*}

\textit{Step 5: Combining steps.} It remains to combine the above multiplier process coupling over $\mathcal N_n$ with the equicontinuity derived in steps 2 and 3 of the proof. For any Borel set $A\subseteq \bbR$, let $A^r$ denote the $r$-enlargement of $A$:
\begin{equation*}
    A^r := \left\{
x\in\bbR : \inf_{a\in A} |x-a| \le r
    \right\}~.
\end{equation*}
Since $\beta < 1$, the bound in \eqref{eq:step_2_equic} is less than that in \eqref{eq:step_3_equic}. Moreover, $\Delta_n^\xi(\beta)$ is given by the sum of the bounds in \eqref{eq:step_3_equic} and \eqref{eq:finite_net_approximation}. Consequently, for some choice of $K_1$ sufficiently large, each of these three bounds can be bounded above $K_1\Delta_n^\xi(\beta)/3$. 

I proceed by applying the conditional Strassen theorem. Conditional on $Z_1,\ldots, Z_n$, I compare $T_{n,G}^{\xi,tr}$ first to the corresponding supremum over $\mathcal N_n \cup (-\mathcal N_n)$. First, notice that:
\begin{align*}
    & \bbE_G\left[
        \sup_{A\in\mathcal B(\bbR)}
        \left\{
        P_{n,\xi}\left(
            T_{n,G}^\xi \in A
        \right)
        -
        P_{n,G}\left(
            Z_{n,G}^{tr}
            \in A^{K_1\Delta_n^\xi(\beta)}
        \right)
        \right\}
    \right]
    \\
    & \qquad \le
    P_{n,G,\xi}\left(
        \left\|
        \bbG_n^\xi
        \right\|_{\calF_{n,G}(n^{-\rho})}
        >
        \frac{K_1\Delta_n^\xi(\beta)}{3}
    \right)
    \\
    & \qquad\quad
    +
    P_{n,G,\xi}\left(
        \left|
        \sup_{f\in\calN_n}
        \left|
        \bbG_n^\xi f
        \right|
        -
        \widetilde Z_{n,G}^{\calN}
        \right|
        >
        \frac{K_1\Delta_n^\xi(\beta)}{3}
    \right)
    \\
    & \qquad\quad
    +
    P\left(
        \left\|
        X_{n,G}
        \right\|_{\calF_{n,G}(n^{-\rho})}
        >
        \frac{K_1\Delta_n^\xi(\beta)}{3}
    \right)
    \\
    & \qquad \le
    P_{n,G}\left(\mathcal E_n^c\right)
    +
    P_{n,G,\xi}\left(
        \left\|
        \bbG_n^\xi
        \right\|_{\calF_{n,G}(n^{-\rho})}
        >
        \frac{K_1\Delta_n^\xi(\beta)}{3}\ \cap
        \ \mathcal E_n
    \right)
    \\
    & \qquad\quad
    +
    P_{n,G,\xi}\left(
        \left|
        \sup_{f\in\calN_n}
        \left|
        \bbG_n^\xi f
        \right|
        -
        \widetilde Z_{n,G}^{\calN}
        \right|
        >
        \frac{K_1\Delta_n^\xi(\beta)}{3}
    \right)
    \\
    & \qquad\quad
    +
    P\left(
        \left\|
        X_{n,G}
        \right\|_{\calF_{n,G}(n^{-\rho})}
        >
        \frac{K_1\Delta_n^\xi(\beta)}{3}
    \right)
    \\
    & \qquad \le
    \beta
    +
    \frac{2}{n}
    +
    C_4\left(\beta+n^{-1}\right)
    +
    \frac{2}{n}
    \\
    & \qquad \le
    K_2\left(\beta+n^{-1}\right).
\end{align*}
The first inequality holds by step 1. The third inequality holds by \eqref{eq:step_3_equic} and by \eqref{eq:finite_net_approximation}. The final inequality holds by taking $K_2 := C_4 +5$. 

Since $\sigma(Z_1,\ldots,Z_n)$ is countable generated, lemma 4.2 of \citealt{CCK_2016} provides $\Tilde{T}_{n,G}^{\xi,tr}$ satisfying:
\begin{equation*}
    \Tilde{T}_{n,G}^{\xi,tr}|Z_1,\ldots,Z_n\overset{d}{=} Z_{n,G}^{tr}\quad a.s.~,
\end{equation*}
and:
\begin{equation*}
    P_{n,G,\xi}\left(
|\Tilde{T}_{n,G}^{\xi,tr} - T_{n,G}^{\xi,tr}|>K_1\Delta_n^\xi(\beta)
    \right)\le K_2(\beta + n^{-1})~,
\end{equation*}
thus concluding the proof.
\end{proof} 

In what follows, we take $\beta_n = \frac{1}{\log n}$, and define $\eta_n = K_1\Delta_n^\xi(\beta_n)$. The following reasoning demonstrates that $\eta_n = o(n^{-\vartheta/4})$:
\begin{align*}
    \eta_n
    &= K_1\Delta_n^\xi(\beta_n)
    \\
    &= K_1\left\{
        n^{\vartheta/4-\rho(1-\vartheta/2)}
        +
        n^{-\rho}(\log n)^{3/4}
        +
        n^{-1/4+\vartheta/2+\rho\vartheta/2}
        (\log n)^{5/4}
    \right\}
    \\
    &= K_1 n^{-\vartheta/4}
    \left\{
        n^{\vartheta/2-\rho(1-\vartheta/2)}
        +
        n^{-(\rho-\vartheta/4)}(\log n)^{3/4}
        +
        n^{-1/4+3\vartheta/4+\rho\vartheta/2}
        (\log n)^{5/4}
    \right\}
    \\
    &= o\left(n^{-\vartheta/4}\right)~.
\end{align*}
The first equality follows by the definition of $\eta_n$. The second equality substitutes $\beta_n=(\log n)^{-1}$ into the definition of $\Delta_n^\xi(\beta)$ in \eqref{eq:delta_n_xi}. The third equality follows by factoring out $n^{-\vartheta/4}$. For the first term in braces, $\vartheta/2-\rho(1-\vartheta/2)<0$ since $\rho>\vartheta/(2-\vartheta)$. For the second term, $\rho>\vartheta/(2-\vartheta)>\vartheta/4$, and therefore $n^{-(\rho-\vartheta/4)}(\log n)^{3/4}=o(1)$. Finally, since $\rho<(1-4\vartheta)/(4\vartheta)$, it follows that $-1/4+3\vartheta/4+\rho\vartheta/2<-1/8+\vartheta/4<0$, where the final inequality uses $\vartheta<2/9<1/2$. Hence each term in braces is $o(1)$, giving that $\eta_n=o(n^{-\vartheta/4})$.

\subsection{Critical Value Validity}

The preceding lemmata establish that both the empirical and multiplier statistics can be strongly approximated by the supremum of the same Gaussian process. To obtain critical value validity, I translate this approximation into a comparison of distribution functions. This requires the distribution function of the limiting Gaussian process supremum to be sufficiently continuous. Lemma \ref{lem:gaussian_AC} supplies the appropriate anti-concentration result. Lemma \ref{lem:CDF_cv} uses the couplings established in lemmata \ref{lem:Gaussian_coupling} and \ref{lem:multiplier_coupling} and this anti-concentration result to establish convergence of CDFs. Finally, an application of assumption \ref{assu:variance_uniform_estimability} accounts for estimation error in the $\hat\sigma_n$. An application of lemma A.1 in \citealt{RomanoShaikh_2012} yields the desired critical value validity.

\begin{lem}
    \label{lem:gaussian_AC}
    Let assumptions \ref{assu:smoothness}, \ref{assu:modulus_binding}, and \ref{assu:pilot_independence} hold. There exist constants $C_0<\infty$ and $n_0<\infty$ such that for any $n\ge n_0$ and $ 0<r\le 1$:
    \begin{equation*}
        \sup_{G,H\in\calG}\sup_{x\in\bbR} P_{n,G}\left(
|Z_{n,G}-x|\le r
        \right) \lesssim r\left(
n^{\vartheta/4} + \sqrt{1\vee \log (C_0/r)}
        \right)~.
    \end{equation*}
    Consequently, the left hand side is $o(1)$ for any $r = o(n^{-\vartheta /4})$.
\end{lem}
\begin{proof}
    Taking $\varepsilon = 1/2$ in lemma \ref{lem:variance_bounds} gives that for sufficiently large $n$, it holds that for any $G,H$ in $\calG$ and any $f$ in $\calF_{n,G}$:
    \begin{equation*}
        \frac{c^2}{2C^2}\le P_Gf^2 \le \frac{C^2}{c^2}~.
    \end{equation*}
    It remains to bound the expected supremum of $Z_{n,G}$. To take suprema of absolute values, consider $\calF_{n,G}\cup \{0\}\cup(-\calF_{n,G})$. The covering number of this augmented class is bounded above by twice the covering number of $\calF_{n,G}$ plus one. Applying Dudley's entropy inequality (corollary 2.2.8 of \citealt{VdVandWellner1996}) and applying the logic used to bound $J_n(1)$ in the proof of \ref{lem:Gaussian_coupling} yields:
    \begin{equation}
        \bbE_G Z_{n,G}^{tr} \lesssim n^{\vartheta /4}\label{eq:Z_n_G_tr_bound}~.
    \end{equation}
    Lemma A.1 of \citealt{CCK_2014} gives, for $0<r\le 1$:
    \begin{align*}
        \sup_x P_{n,G}\!\left(\left|Z^{tr}_{n,G}-x\right|\le r\right)
&\le
Kr\left\{
\mathbb E_G Z^{tr}_{n,G}
+
\sqrt{
1\vee\log\left(\frac{C}{cr}\right)
}
\right\}
\\
&\le
Kr\left\{
n^{\vartheta/4}
+
\sqrt{
1\vee\log\left(\frac{C}{cr}\right)
}
\right\}~,
    \end{align*}
    thus proving the result with $C_0 := C/c$.
\end{proof}

To obtain critical value validity, I compare distribution functions. In particular, define:
\begin{equation*}
    F_{n,G}(x)
:=
P_{n,G}\!\left(T_n(\Gamma(G))\le x\right),
\qquad
\widehat F_{n,G}(x)
:=
P_{n,\xi}\!\left(
T_n^\xi(\Gamma(G))\le x
\,\middle|\,
Z_1,\ldots,Z_n
\right)~.
\end{equation*}
The following lemma \ref{lem:CDF_cv} verifies the convergence of these two distribution functions.

\begin{lem}
    \label{lem:CDF_cv}
    Let assumptions \ref{assu:smoothness}, \ref{assu:modulus_binding}, and \ref{assu:pilot_independence} hold. Let $D$ be such that $D^2 / (8\underline{s}^2) > 1+\frac{\vartheta}{4}$. Then the empirical and multiplier distribution functions satisfy, for any $\varepsilon>0$:
    \begin{equation}
    \label{eq:CDF_cv_1}
       \sup_{G,H\in\calG}
P_{n,G}\left(
\sup_x
\left|
\widehat F_{n,G}(x)
-
P_{n,G}\left(Z_{n,G}^{tr}\le x\right)
\right|
>
\varepsilon
\right)
\to 0~,
    \end{equation}
    and:
    \begin{equation}
    \label{eq:CDF_cv_2}
        \sup_{G,H\in\calG} \sup_{x\in\bbR} \left|
        F_{n,G}(x) - P_{n,G}(Z_{n,G}^{tr}\le x)
       \right|\to0~.
    \end{equation}
\end{lem}
\begin{proof}
    I apply lemma 2.1 of \citealt{CCK_2016} which states, for random variables $V$ and $W$ and for $r>0$:
    \begin{equation*}
        \sup_x| P(V\le x) - P(W\le x)| \le P(|V-W|> r) + \sup_x P(|W-x|\le r)~.
    \end{equation*}
    The couplings in lemmata \ref{lem:Gaussian_coupling} and \ref{lem:multiplier_coupling}, and the anti-concentration result in lemma \ref{lem:gaussian_AC} control both terms on the right hand side of the above.

    By \eqref{eq:L2_bounds_tail_psi}, and since $D$ is such that $D^2 / (8\underline{s}^2) > 1+\frac{\vartheta}{4}$, it follows that:
    \begin{equation}
    \label{eq:tail_bounds_CDF_cv}
        nP_G(|\hat\theta| > \tau_n)\lesssim n^{1-D^2/(8\underline{s}^2)} = o(n^{-\vartheta / 4})~.
    \end{equation}
    By \eqref{eq:tail_bounds_CDF_cv} and by a union bound, $P_{n,G}(\max_i |\hat\theta_i|>\tau_n) = o(n^{-\vartheta / 4})$. On the event $\max_i |\hat\theta_i|\le\tau_n$:
    \begin{align*}
\left|T_n(\Gamma(G))-T_{n,G}^{tr}\right|
&=
\left|
\sup_{\gamma\in\Gamma(G)}
\left|\mathbb G_{n,G}\psi_{n,\gamma}\right|
-
\sup_{\gamma\in\Gamma(G)}
\left|
\mathbb G_{n,G}
\left(
\psi_{n,\gamma}\mathbbm{1}\{|t|\le\tau_n\}
\right)
\right|
\right|
\\
&\le
\sup_{\gamma\in\Gamma(G)}
\left|
\mathbb G_{n,G}
\left(
\psi_{n,\gamma}\mathbbm{1}\{|t|>\tau_n\}
\right)
\right|
\\
&=
\sup_{\gamma\in\Gamma(G)}
\left|
\sqrt n(P_n-P_G)
\left(
\psi_{n,\gamma}\mathbbm{1}\{|t|>\tau_n\}
\right)
\right|
\\
&=
\sqrt n\sup_{\gamma\in\Gamma(G)}
\left|
P_G\!\left(
\psi_{n,\gamma}\mathbbm{1}\{|t|>\tau_n\}
\right)
\right|
\\
&\lesssim
nP_G(|\hat\theta|>\tau_n)
=
o(n^{-\vartheta/4})~.
\end{align*}
The final inequality holds by the bound in \eqref{eq:global_psi_bounds}.

I use these rates to apply lemma 2.1 of \citealt{CCK_2014} and obtain \eqref{eq:CDF_cv_2}. To this end, let $d_n = o(n^{-\vartheta/4})$ as in lemma \ref{lem:Gaussian_coupling}. It follows that:
\begin{align*}
&P_{n,G}\left(
\left|T_n(\Gamma(G))-\Tilde{T}_{n,G}^{tr}\right|
>
d_n+nP_G\left(|\hat{\theta}|>\tau_n\right)
\right)
\\
&\qquad\le
P_{n,G}\left(
\left|T_n(\Gamma(G))-T_{n,G}^{tr}\right|
>
nP_G\left(|\hat{\theta}|>\tau_n\right)
\right)
+
P_{n,G}\left(
\left|T_{n,G}^{tr}-\Tilde{T}_{n,G}^{tr}\right|
>d_n
\right)
\\
&\qquad \le
P_{n,G}\left(
\max_i |\hat\theta_i| >\tau_n
\right)
+
P_{n,G}\left(
\left|T_{n,G}^{tr}-\Tilde{T}_{n,G}^{tr}\right|
>d_n
\right) =o(1)~.
\end{align*}
The first inequality holds by the triangle inequality and a union bound. The second inequality holds since $\left|T_{n,G}^{tr}-\Tilde{T}_{n,G}^{tr}\right| \ne 0$ if and only if $\max_i |\hat\theta_i| >\tau_n$. The final equality holds by the above application of \eqref{eq:tail_bounds_CDF_cv} and by lemma \ref{lem:Gaussian_coupling}. Applying the above, lemma \ref{lem:gaussian_AC}, and lemma 2.1 of \citealt{CCK_2016} yields \ref{eq:CDF_cv_2}.

It remains to establish \eqref{eq:CDF_cv_1}. To this end, notice that:
\begin{align*}
\sup_x
\left|
\widehat F_{n,G}(x)
-
P_{n,G}\left(Z_{n,G}^{tr}\le x\right)
\right|
&\le
P_{n,G}\left(
\left|
T_{n,G}^{\xi,tr}
-
\widetilde T_{n,G}^{\xi,tr}
\right|
>
\eta_n
\,\middle|\,
Z_1,\ldots,Z_n
\right)
\\
&\quad+
\mathbbm{1}
\left\{
\max_{i\le n}|\hat\theta_i|>\tau_n
\right\}
\\
&\quad+
\sup_x
P_{n,G}\left(
\left|Z_{n,G}^{tr}-x\right|
\le
\eta_n
\right).
\end{align*}
This inequality holds by applying lemma 2.1 of \citealt{CCK_2016} on the event that $\Tilde{T}_{n,G}^{\xi,tr} = \Tilde{T}_{n,G}^{\xi}$, and adding an additional term to indicate the event $\Tilde{T}_{n,G}^{\xi,tr} \ne \Tilde{T}_{n,G}^{\xi}$. I control the first term on the right hand side above.
\begin{align*}
\bbE_G\left[
P_{n,G}\left(
\left|
T_{n,G}^{\xi,tr}
-
\widetilde T_{n,G}^{\xi,tr}
\right|
>
\eta_n
\,\middle|\,
Z_1,\ldots,Z_n
\right)
\right]
&=
P_{n,G}\left(
\left|
T_{n,G}^{\xi,tr}
-
\widetilde T_{n,G}^{\xi,tr}
\right|
>
\eta_n
\right)
\\
&=o(1)~.
\end{align*}
Therefore, by Markov's inequality, for any $\varepsilon>0$,
\begin{align*}
& \sup_{G,H\in\calG}
P_{n,G}\Bigg(
P_{n,G}\left(
\left|
T_{n,G}^{\xi,tr}
-
\widetilde T_{n,G}^{\xi,tr}
\right|
>
\eta_n
\,\middle|\,
Z_1,\ldots,Z_n
\right)
>
\varepsilon
\Bigg)\\
&\qquad \le
\frac{1}{\varepsilon}
\sup_{G,H\in\calG}
\bbE_G\left[
P_{n,G}\left(
\left|
T_{n,G}^{\xi,tr}
-
\widetilde T_{n,G}^{\xi,tr}
\right|
>
\eta_n
\,\middle|\,
Z_1,\ldots,Z_n
\right)
\right]
\\
&\qquad =o(1)~.
\end{align*}
Moreover, by \eqref{eq:tail_bounds_CDF_cv}:
\begin{equation*}
    \sup_{G,H\in\calG}
P_{n,G}\left(
\max_{i\le n}|\hat\theta_i|>\tau_n
\right)
=
o(n^{-\vartheta/4})
\to 0~.
\end{equation*}
Finally, combining the above two bounds and applying lemma \ref{lem:gaussian_AC} yields:
\begin{equation*}
    \sup_{G,H\in\calG}
P_{n,G}\left(
\sup_x
\left|
\widehat F_{n,G}(x)
-
P_{n,G}\left(Z_{n,G}^{tr}\le x\right)
\right|
>
\varepsilon
\right)
\to 0~,
\end{equation*}
Thus proving the result.
\end{proof}
\begin{proof}
    \textbf{Proof of theorem \ref{thm:bootstrap_validity}:} By the triangle inequality, and by applying \eqref{eq:CDF_cv_1} and \eqref{eq:CDF_cv_2}:
    \begin{equation}
        \sup_{G,H\in\calG} P_{n,G}\left(
        \sup_x \left|
\hat F_{n,G}(x) - F_{n,G}(x)
        \right|>\varepsilon
        \right)\to 0~,
    \end{equation}
    as $n\to\infty$. Applying lemma A.1.vi of \citealt{RomanoShaikh_2012} provides:
         \begin{align*}
\sup_{G,H\in\mathcal G}P_{n,G}\!\left(
T_n(\Gamma(G))
>
\hat c_n(1-\alpha,\Gamma(G))
\right)
&\le
\alpha+\varepsilon
+
\sup_{G,H\in\mathcal G}P_{n,G}\!\left(
\sup_x
\left\{
\widehat F_{n,G}(x)
-
F_{n,G}(x)
\right\}
>
\varepsilon
\right)
\\
&=
\alpha+\varepsilon+o(1)~.
\end{align*}
Since $\varepsilon$ was arbitrary, $\hat{c}_n$ satisfy condition (i) of assumption \ref{assu:cv_validity_monotonicity}. Condition (ii) of assumption \ref{assu:cv_validity_monotonicity} holds by the monotonicity of the supremum multiplier statistic. 
\end{proof}
It remains to remark on estimation error in the standard errors. I apply the arguments contained in the proof of theorem 3.1 of \citealt{CCK_ac_2014} and apply assumption \ref{assu:variance_uniform_estimability}. Define:
\begin{equation*}
    S_{n,G} := \sup_{\gamma \in \Gamma(G)}\left|
    \frac{\sqrt{n}(\hat{L}_n(\gamma)- P_GQ_\gamma)}{\sigma_G(Q_\gamma)}
    \right|~.
\end{equation*}
Let $\calH_n(G)$ denote the set of $H$ such that:
\begin{equation*}
    \sup_{\gamma \in \Gamma(G)}\left|
\frac{\hat\sigma_n(Q_\gamma)}{\sigma_G(Q_\gamma)}-1
    \right|\le\varepsilon_n~,
\end{equation*}
where $\varepsilon_n = o(n^{-\vartheta/2})$ is as supplied in assumption \ref{assu:variance_uniform_estimability}. On the event $H\in\calH_{n,G}$:
\begin{equation*}
    |S_{n,G}-T_n(\Gamma(G))|\le\varepsilon_nT_n(\Gamma(G))~.
\end{equation*}
On the further event $T_n(\Gamma(G)) \le \varepsilon_n^{-1/2}$, this difference can be further bounded by $\varepsilon_n^{1/2}$. On this event, $T_n(\Gamma(G)) \le x-\varepsilon_n^{1/2}$ implies that $S_{n,G}\le x$ which further implies that $T_n(\Gamma(G)) \le x+\varepsilon_n^{1/2}$. Consequently:
\begin{align*}
    &\{T_n(\Gamma(G))\le x-\varepsilon_n^{1/2}\}\cap \{T_n(\Gamma(G))\le \varepsilon_n^{-1/2}\}\\
    &\qquad\subseteq \{S_{n,G}\le x\} \\
    &\qquad \subseteq \{T_n(\Gamma(G))\le x+\varepsilon_n^{1/2}\}\cup \{T_n(\Gamma(G))> \varepsilon_n^{-1/2}\}\}
\end{align*}
This further implies:
\begin{align}
    &F_{n,G}(x-\varepsilon_n^{1/2}) - P_{n,G}(T_n(\Gamma(G))> \varepsilon_n^{-1/2})\nonumber\\
    &\qquad P_{n,G}(S_{n,G}\le x)\nonumber\\
    & \qquad F_{n,G}(x+\varepsilon_n^{1/2}) + P_{n,G}(T_n(\Gamma(G))> \varepsilon_n^{-1/2})\label{eq:prob_bounds}
\end{align}
Therefore:
\begin{align}
\sup_x
\left|
P_{n,G}(S_{n,G}\le x)
-
P_{n,G}(Z_{n,G}^{tr}\le x)
\right|
&\le
\sup_x
\max\Bigg\{
\left|
F_{n,G}\!\left(x+\varepsilon_n^{1/2}\right)
-
P_{n,G}(Z_{n,G}^{tr}\le x)
\right|,\nonumber
\\
&\hspace{45mm}
\left|
F_{n,G}\!\left(x-\varepsilon_n^{1/2}\right)
-
P_{n,G}(Z_{n,G}^{tr}\le x)
\right|
\Bigg\}
\nonumber\\
&\qquad
+
P_{n,G}\!\left(
T_n(\Gamma(G))
>
\varepsilon_n^{-1/2}
\right)
\nonumber\\
&\le
\sup_x
\left|
F_{n,G}(x)
-
P_{n,G}(Z_{n,G}^{tr}\le x)
\right|
\nonumber\\
&\qquad
+
P_{n,G}\!\left(
T_n(\Gamma(G))
>
\varepsilon_n^{-1/2}
\right)
\nonumber\\
&\qquad
+
\sup_x
P_{n,G}\!\left(
\left|
Z_{n,G}^{tr}-x
\right|
\le
\varepsilon_n^{1/2}
\right)
\nonumber\\
&\le
2
\sup_x
\left|
F_{n,G}(x)
-
P_{n,G}(Z_{n,G}^{tr}\le x)
\right|
+
\varepsilon_n^{1/2}\,
\mathbb E_G Z_{n,G}^{tr}
\nonumber\\
&\qquad
+
\sup_x
P_{n,G}\!\left(
\left|
Z_{n,G}^{tr}-x
\right|
\le
\varepsilon_n^{1/2}
\right)
=
o(1)\label{eq:CDF_bounds_true_sigma}~.
 \end{align}
 The first inequality holds by \eqref{eq:prob_bounds} and the convexity of absolute values. The second inequality holds by a union bound. The third inequality holds by Markov's inequality. The final equality holds by application of lemma \ref{lem:CDF_cv}, \eqref{eq:Z_n_G_tr_bound}, the assumption that $\varepsilon_n = o(n^{-\vartheta/2})$, and by lemma \ref{lem:gaussian_AC}. Applying \eqref{eq:CDF_bounds_true_sigma}, \eqref{eq:CDF_cv_1}, and the triangle inequality yields:
 \begin{equation*}
     \label{eq:CDF_cv_true_sigma}
     \sup_{\substack{G,H\in\mathcal G\\ H\in\calH_n(G)}}
P_{n,G}\!\left(
\sup_x
\left|
\widehat F_{n,G}(x)
-
P_{n,G}(S_{n,G}\le x)
\right|
>
\varepsilon
\right)
\to 0~.
 \end{equation*}
 \begin{theorem}
 \label{thm:cv_assumptions_verified_var_est}
      Let assumptions \ref{assu:smoothness}, \ref{assu:modulus_binding}, \ref{assu:pilot_independence}, and \ref{assu:variance_uniform_estimability} hold. Then, for any $\alpha\in(0,1)$, $\hat{c}_n$ satisfy:
      \begin{equation}
          \limsup_{n\to\infty}\sup_{G\in\calG}P_{n,G}\!\left(
S_{n,G}
>
\hat c_n(1-\alpha,\Gamma(G))
\right)\le\alpha
      \end{equation}
 \end{theorem}
 \begin{proof}
     Fix $0<\varepsilon < 1-\alpha$. Lemma A.1.vi of \citealt{RomanoShaikh_2012} provides:
     \begin{align*}
\sup_{\substack{G,H\in\mathcal G\\ H\in\calH_n(G)}}P_{n,G}\!\left(
S_{n,G}
>
\hat c_n(1-\alpha,\Gamma(G))
\right)
&\le
\alpha+\varepsilon
+
\sup_{\substack{G,H\in\mathcal G\\ H\in\calH_n(G)}}P_{n,G}\!\left(
\sup_x
\left\{
\widehat F_{n,G}(x)
-
P_{n,G}(S_{n,G}\le x)
\right\}
>
\varepsilon
\right)
\\
&=
\alpha+\varepsilon+o(1)~.
\end{align*}
By assumption \ref{assu:variance_uniform_estimability}, $H\in\calH_{n,G}$ with probability approaching one. Therefore:
\begin{equation*}
    \sup_{G\in\calG}P_{n,G}\!\left(
S_{n,G}
>
\hat c_n(1-\alpha,\Gamma(G))
\right) \le \alpha+\varepsilon~.
\end{equation*}
Since $\varepsilon$ was arbitrary, part (i) of \ref{assu:cv_validity_monotonicity} holds. Since $\widetilde{\Gamma}_1 \subseteq \widetilde{\Gamma}_2
\text{ implies }
T_n^\xi(\widetilde{\Gamma}_1)
\le
T_n^\xi(\widetilde{\Gamma}_2)$, part (ii) of \ref{assu:cv_validity_monotonicity} holds, thus concluding the proof.
 \end{proof}

\section{Additional Results from Economic Mobility Application}
\label{app:additional_results_mob}

This appendix presents supplementary results for the economic mobility application in section \ref{sec:economic-mobility-application}. In appendix \ref{app:KDE_results_mob}, I present results including an extension of the KDE $F$-localizations of \citealt{IgnatiadisandWager2021} to the heteroskedastic hierarchical model in \eqref{eq:EB_DGP}. In appendix \ref{app:additional_figs_mob}, I provide additional figures for the economic mobility application referenced in the main text.

% \subsection{Implementation Details from Economic Mobility Application}
% \label{app:Implementation_details_mob}

\subsection{Results using KDE \texorpdfstring{$F$}{F}-localization}
\label{app:KDE_results_mob}

This section presents results using an alternative $F$-localization procedure, discussed in remark \ref{rem:choice_of_floc}.

\begin{figure}[H]
    
    \begin{subfigure}{0.45\textwidth}
        \centering
        \includegraphics[width=\linewidth]{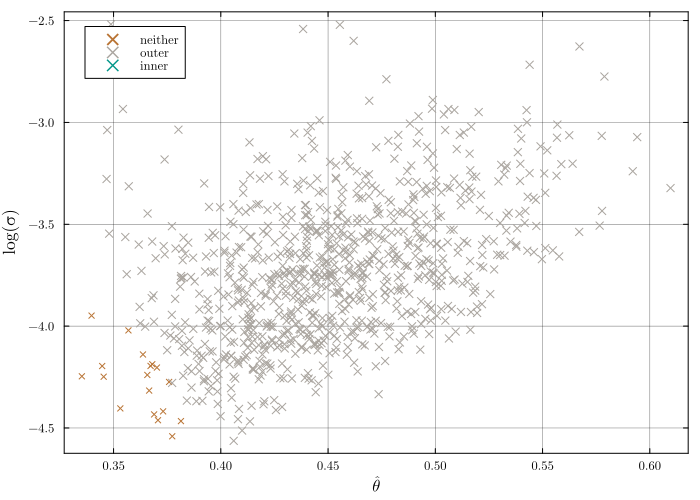}
        \caption{KDE $F$-localization}
    \end{subfigure}
    \hfill
    \begin{subfigure}{0.45\textwidth}
        \centering
        \includegraphics[width=\linewidth]{figs/unif_boot_inner_outer_theta_logsigma.png}
        \caption{Step-Down}
    \end{subfigure}
    \caption{Inner and outer confidence sets for top-third Bayes-optimal selection, in $(\hat\theta,\log(\sigma))$ space.}
    \label{fig:confidence_sets_inner_outer_pts_KDE}
    \end{figure}

\begin{figure}[H]
\begin{subfigure}{0.45\textwidth}
        \centering
        \includegraphics[width=\linewidth]{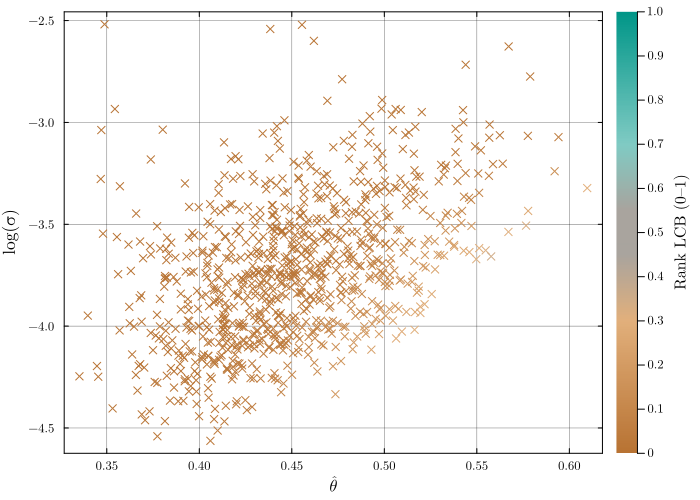}
        \caption{KDE $F$-localization}
    \end{subfigure}
    \hfill
    \begin{subfigure}{0.45\textwidth}
        \centering
        \includegraphics[width=\linewidth]{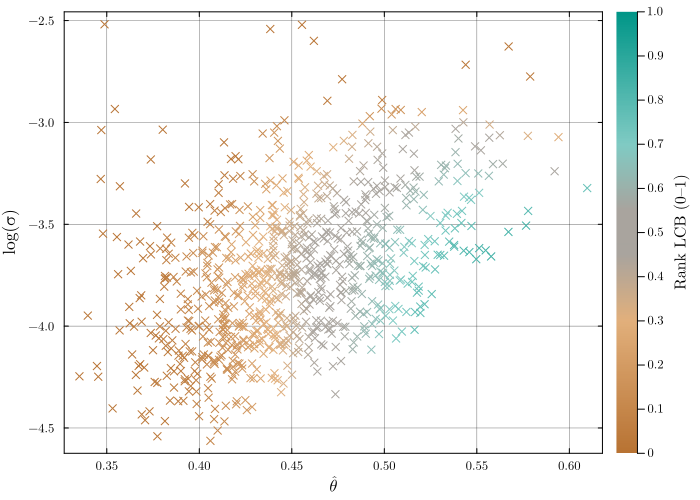}
        \caption{Step-Down}
    \end{subfigure}
    \caption{Lower confidence band for the ranks of Seattle census tracts.}
    \label{fig:rank_LCBs_scatter_KDE}
    \end{figure}

\begin{figure}[H]
\begin{subfigure}{0.45\textwidth}
        \centering
        \includegraphics[width=\linewidth]{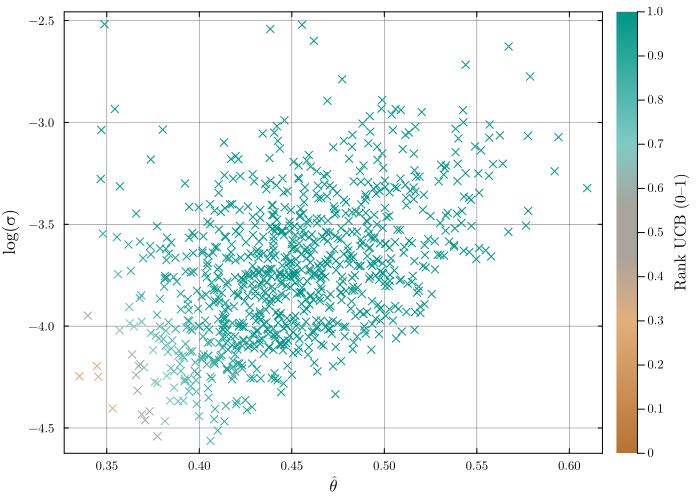}
        \caption{$F$-localization}
    \end{subfigure}
    \hfill
    \begin{subfigure}{0.45\textwidth}
        \centering
        \includegraphics[width=\linewidth]{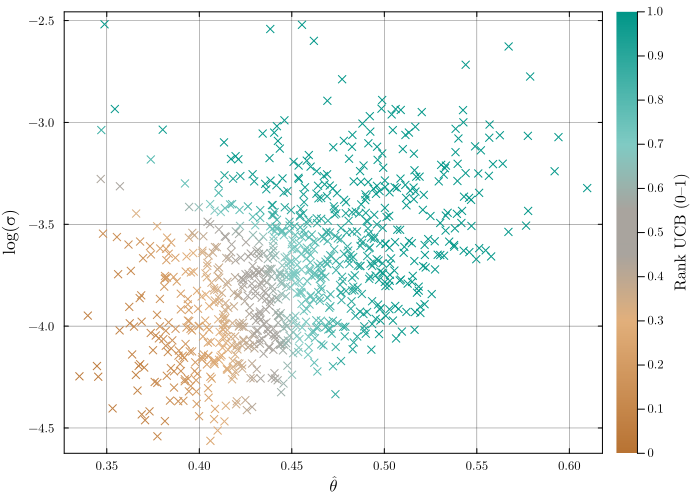}
        \caption{Step-Down}
    \end{subfigure}
    \caption{Upper confidence band for the ranks of Seattle census tracts.}
    \label{fig:rank_UCBs_scatter_KDE}
    \end{figure}

\subsection{Supplemental figures}
\label{app:additional_figs_mob}
This section presents supplemental figures referenced in the main text.

\begin{figure}[H]
    \begin{subfigure}{0.45\textwidth}
        \centering
        \includegraphics[width=\linewidth]{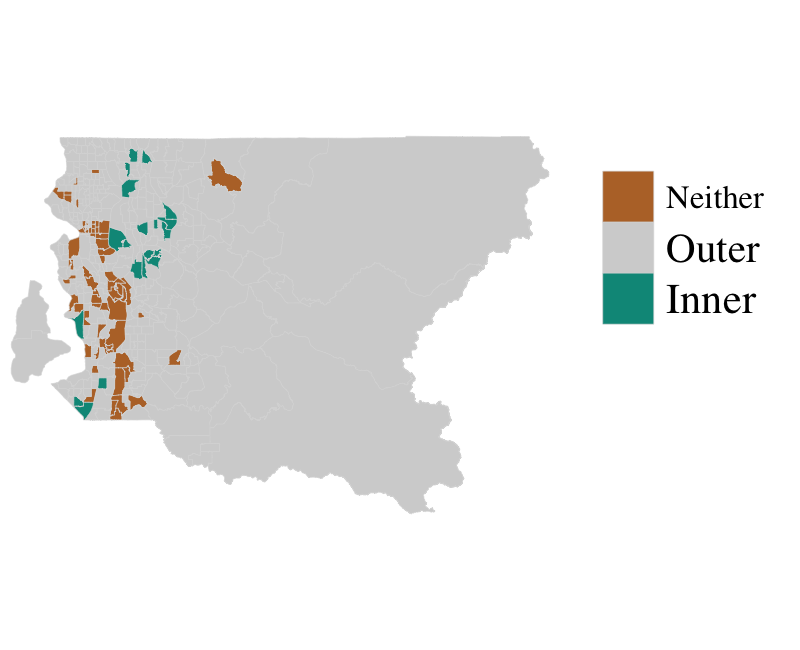}
        \caption{$F$-localization}
    \end{subfigure}
    \hfill
    \begin{subfigure}{0.45\textwidth}
        \centering
        \includegraphics[width=\linewidth]{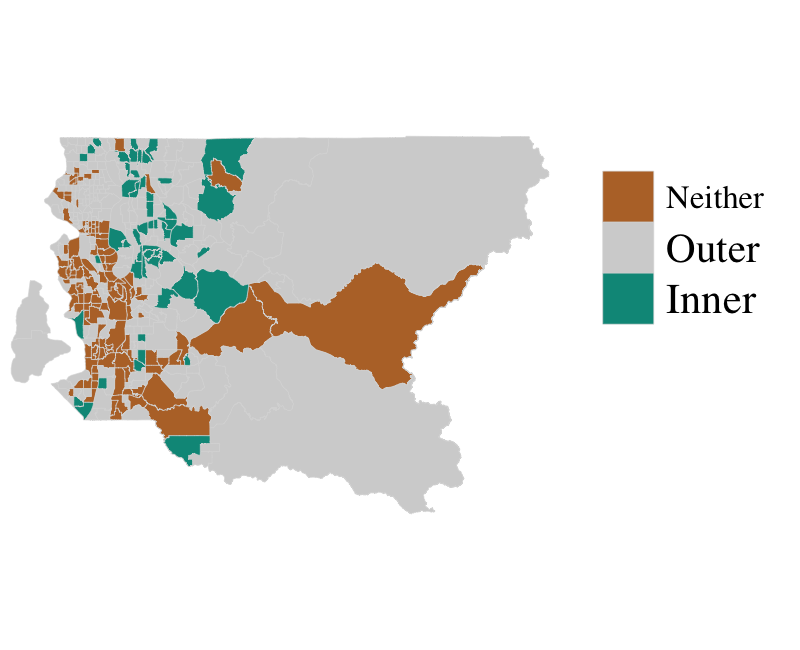}
        \caption{Step-Down}
    \end{subfigure}
    \caption{Inner and outer confidence sets for top-third Bayes-optimal selection, plotted on a map of Seattle.}
    \label{fig:sea_maps}
    \end{figure}

\begin{figure}[H]
\begin{subfigure}{0.45\textwidth}
        \centering
        \includegraphics[width=\linewidth]{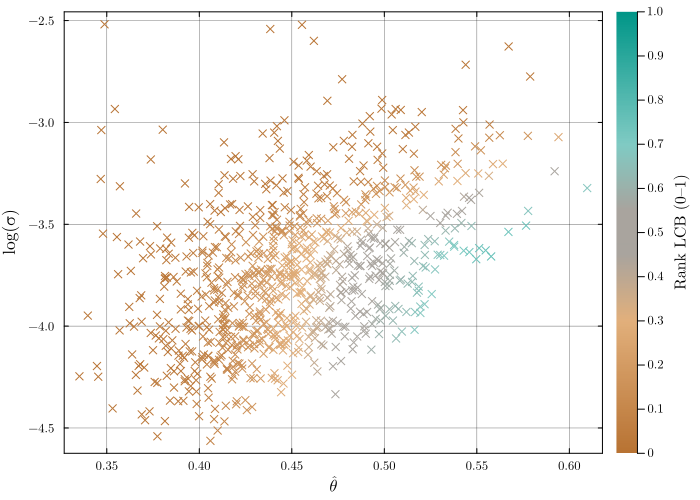}
        \caption{$F$-localization}
    \end{subfigure}
    \hfill
    \begin{subfigure}{0.45\textwidth}
        \centering
        \includegraphics[width=\linewidth]{figs/unif_boot_theta_logsigma_Rank_LCB_gradient.png}
        \caption{Step-Down}
    \end{subfigure}
    \caption{Lower confidence band for the ranks of Seattle census tracts.}
    \label{fig:rank_LCBs_scatter}
    \end{figure}

\begin{figure}[H]
\begin{subfigure}{0.45\textwidth}
        \centering
        \includegraphics[width=\linewidth]{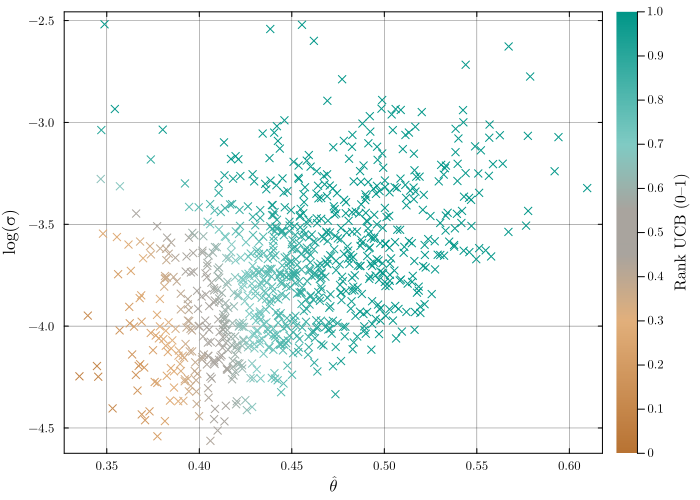}
        \caption{$F$-localization}
    \end{subfigure}
    \hfill
    \begin{subfigure}{0.45\textwidth}
        \centering
        \includegraphics[width=\linewidth]{figs/unif_boot_theta_logsigma_Rank_UCB_gradient.png}
        \caption{Step-Down}
    \end{subfigure}
    \caption{Upper confidence band for the ranks of Seattle census tracts.}
    \label{fig:rank_UCBs_scatter}
    \end{figure}

\newpage

\begin{figure}[H]
\begin{subfigure}{0.45\textwidth}
        \centering
        \includegraphics[width=\linewidth]{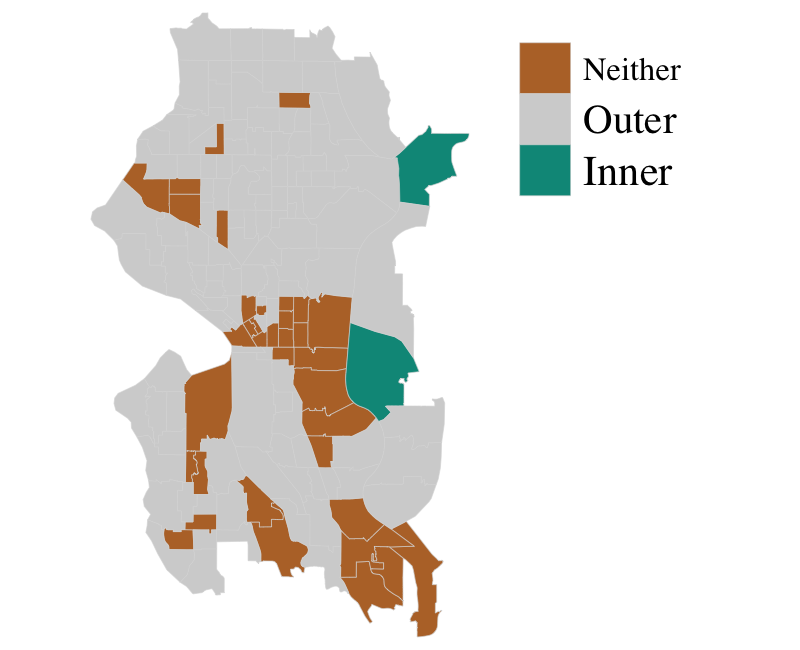}
        \caption{$F$-localization}
    \end{subfigure}
    \hfill
    \begin{subfigure}{0.45\textwidth}
        \centering
        \includegraphics[width=\linewidth]{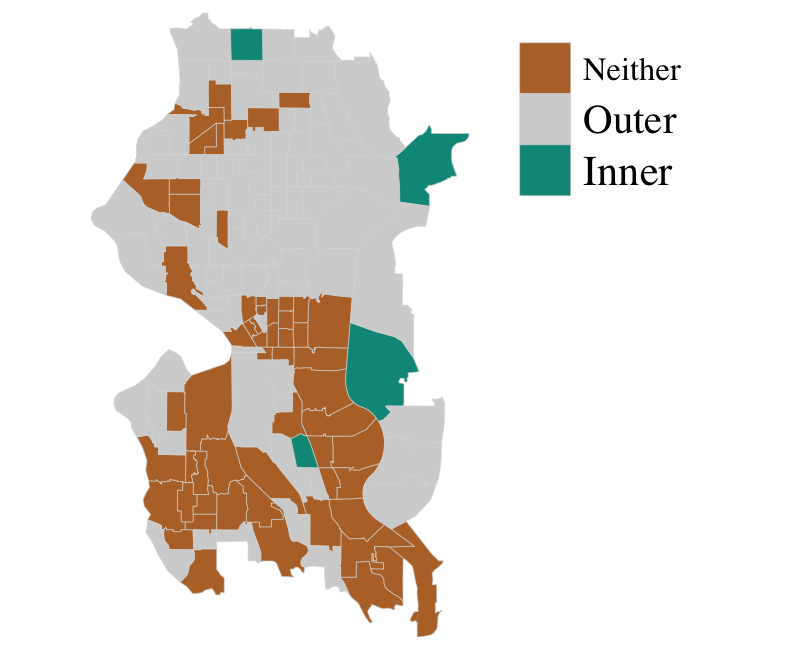}
        \caption{Step-Down}
    \end{subfigure}
    \caption{Inner and outer confidence sets for top-third Bayes-optimal selection, plotted on a map of urban Seattle.}
    \label{fig:sea_maps_urban}
    \end{figure}

    \newpage

    \begin{figure}[H]
\begin{subfigure}{0.45\textwidth}
        \centering
        \includegraphics[width=\linewidth]{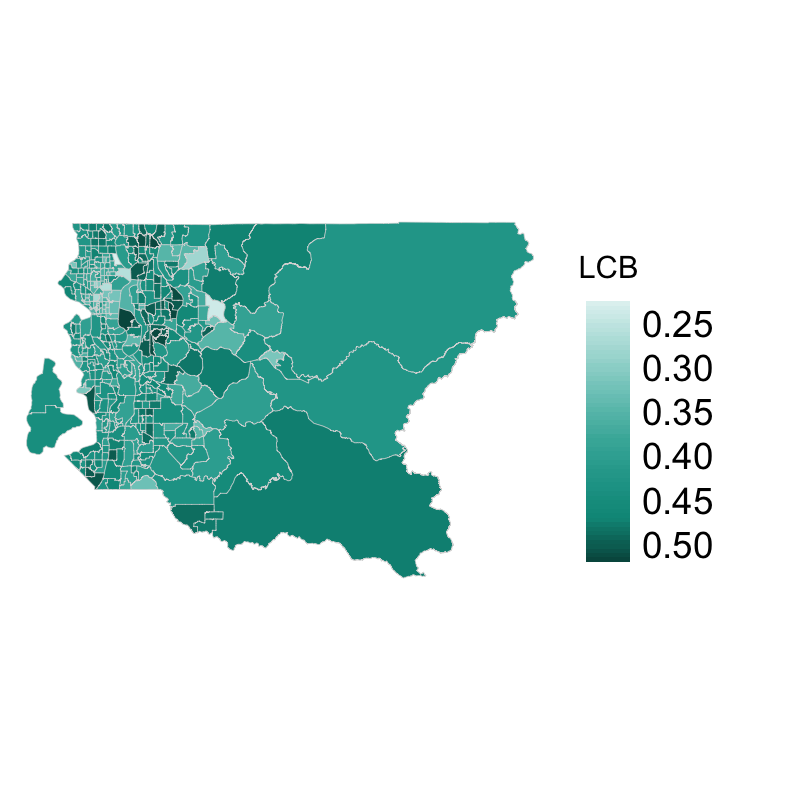}
        \caption{$F$-localization}
    \end{subfigure}
    \hfill
    \begin{subfigure}{0.45\textwidth}
        \centering
        \includegraphics[width=\linewidth]{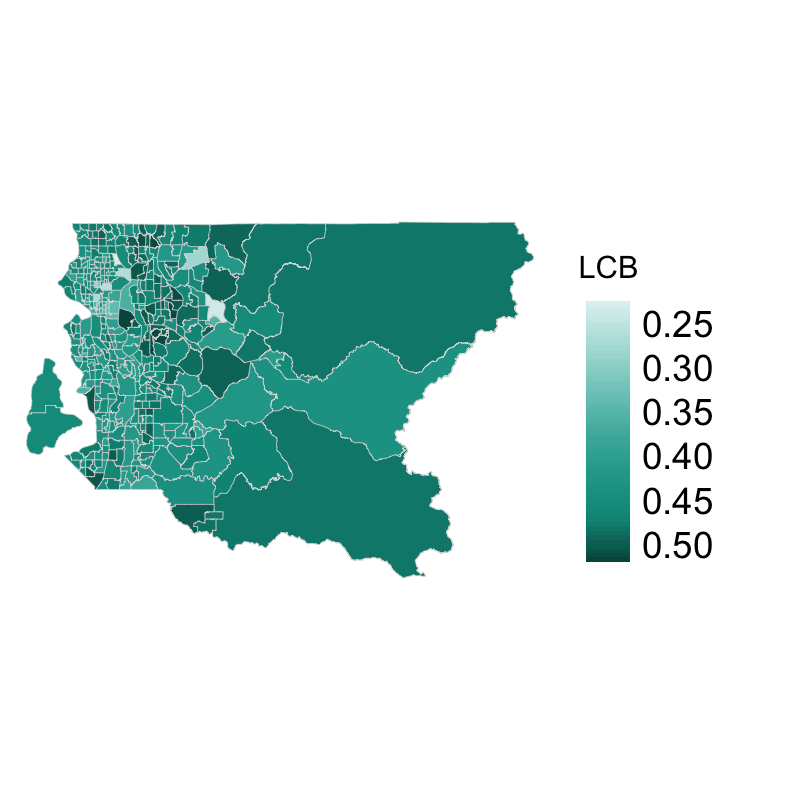}
        \caption{Step-Down}
    \end{subfigure}
    \caption{Lower confidence band for the ranks of Seattle census tracts.}
    \label{fig:rank_LCBs}
    \end{figure}

\begin{figure}[H]
\begin{subfigure}{0.45\textwidth}
        \centering
        \includegraphics[width=\linewidth]{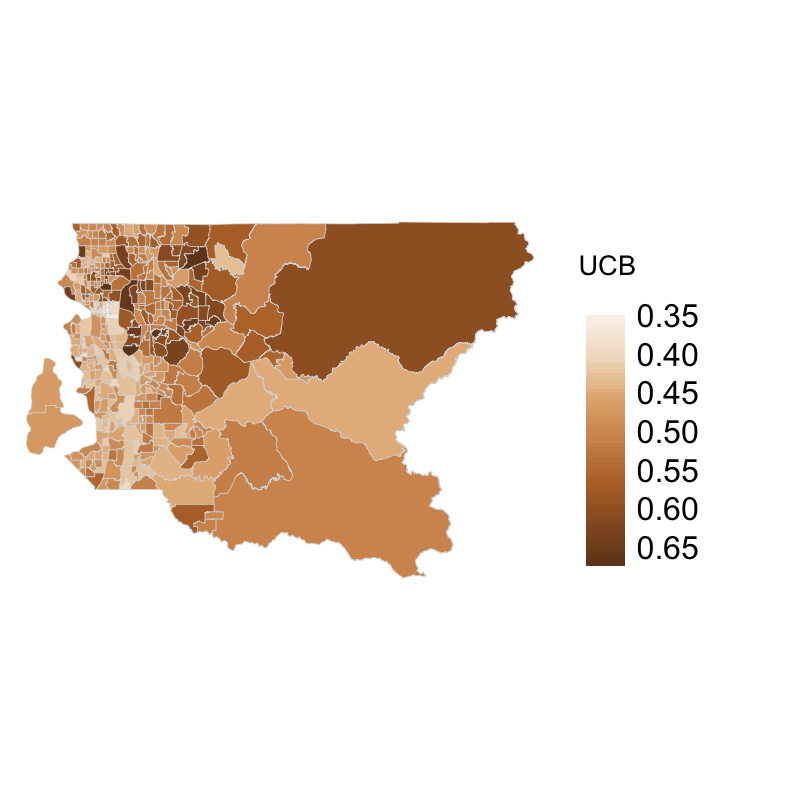}
        \caption{$F$-localization}
    \end{subfigure}
    \hfill
    \begin{subfigure}{0.45\textwidth}
        \centering
        \includegraphics[width=\linewidth]{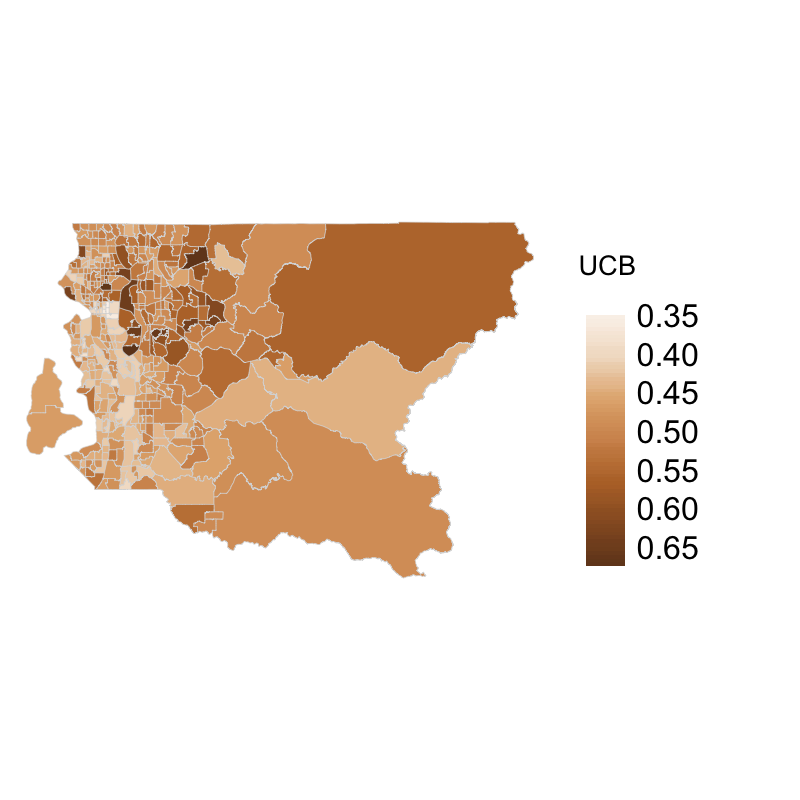}
        \caption{Step-Down}
    \end{subfigure}
    \caption{Upper confidence band for the ranks of Seattle census tracts.}
    \label{fig:rank_UCBs}
    \end{figure}

    \newpage

    \begin{figure}[H]
\begin{subfigure}{0.45\textwidth}
        \centering
        \includegraphics[width=\linewidth]{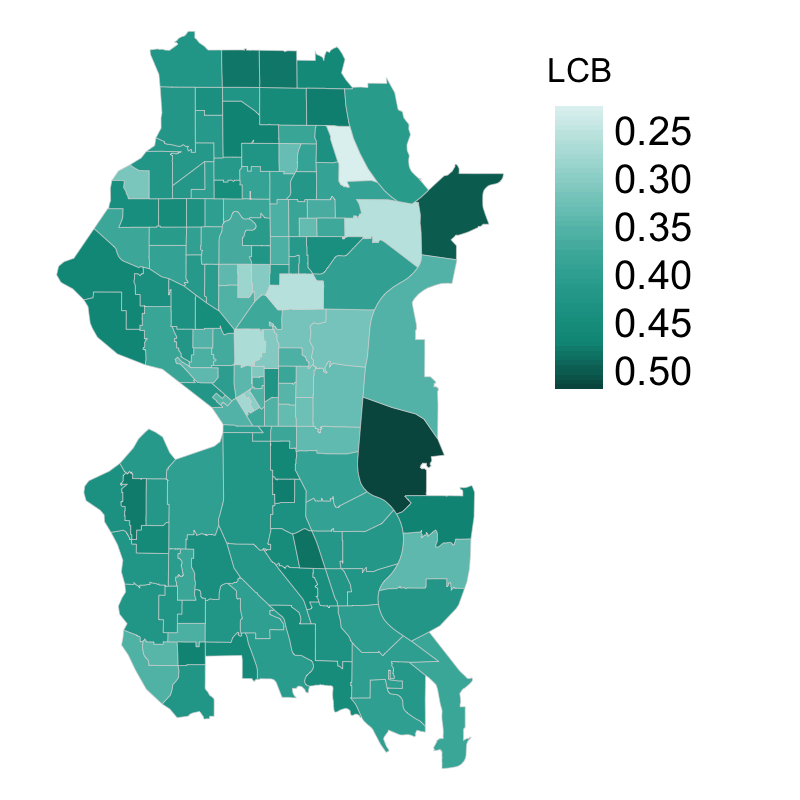}
        \caption{$F$-localization}
    \end{subfigure}
    \hfill
    \begin{subfigure}{0.45\textwidth}
        \centering
        \includegraphics[width=\linewidth]{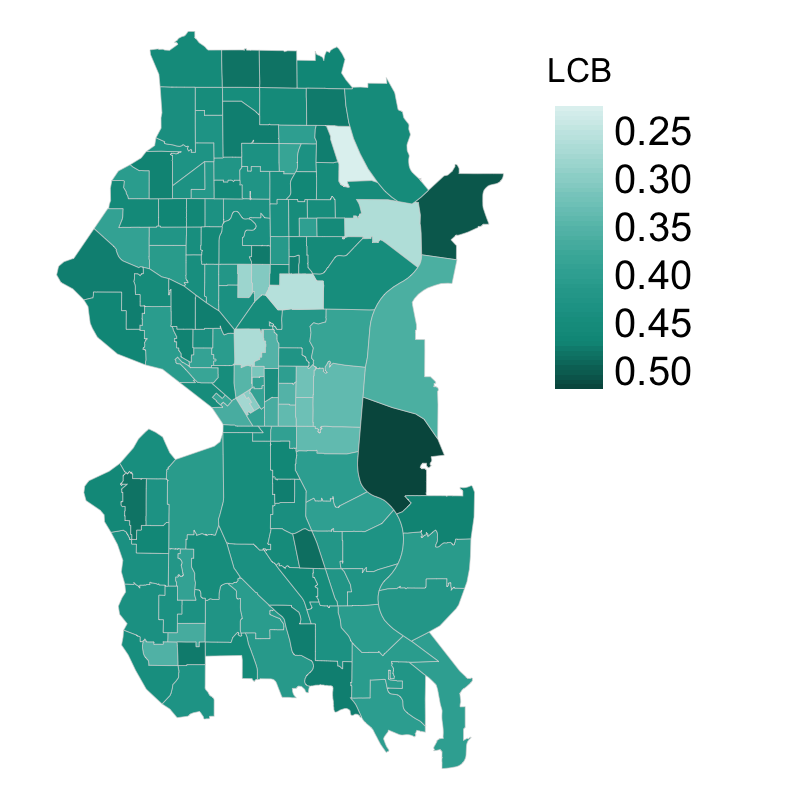}
        \caption{Step-Down}
    \end{subfigure}
    \caption{Lower confidence band for the ranks of Seattle census tracts.}
    \label{fig:urban_rank_LCBs}
    \end{figure}

\begin{figure}[H]
\begin{subfigure}{0.45\textwidth}
        \centering
        \includegraphics[width=\linewidth]{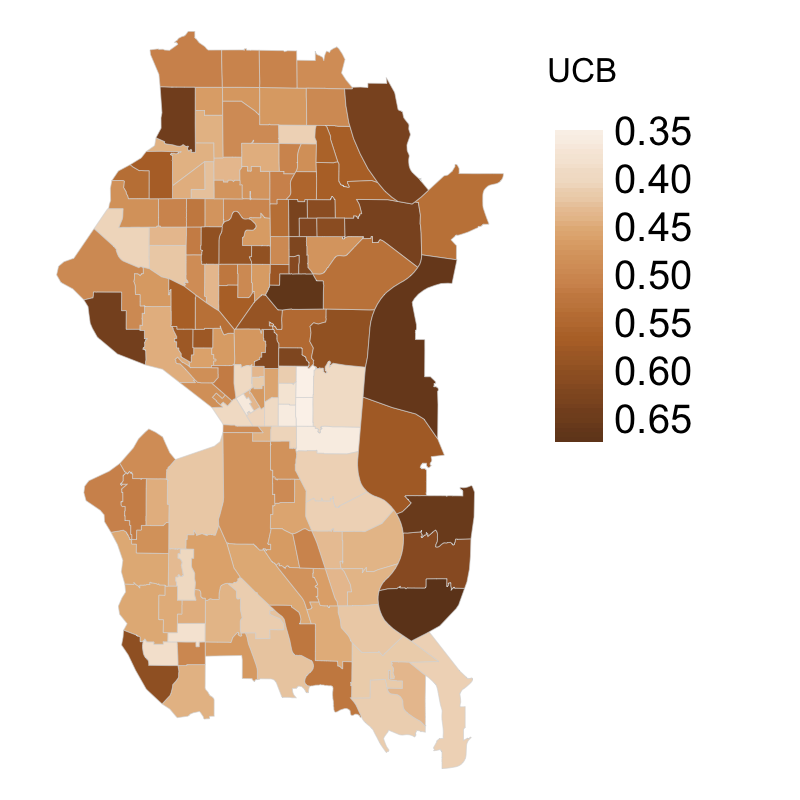}
        \caption{$F$-localization}
    \end{subfigure}
    \hfill
    \begin{subfigure}{0.45\textwidth}
        \centering
        \includegraphics[width=\linewidth]{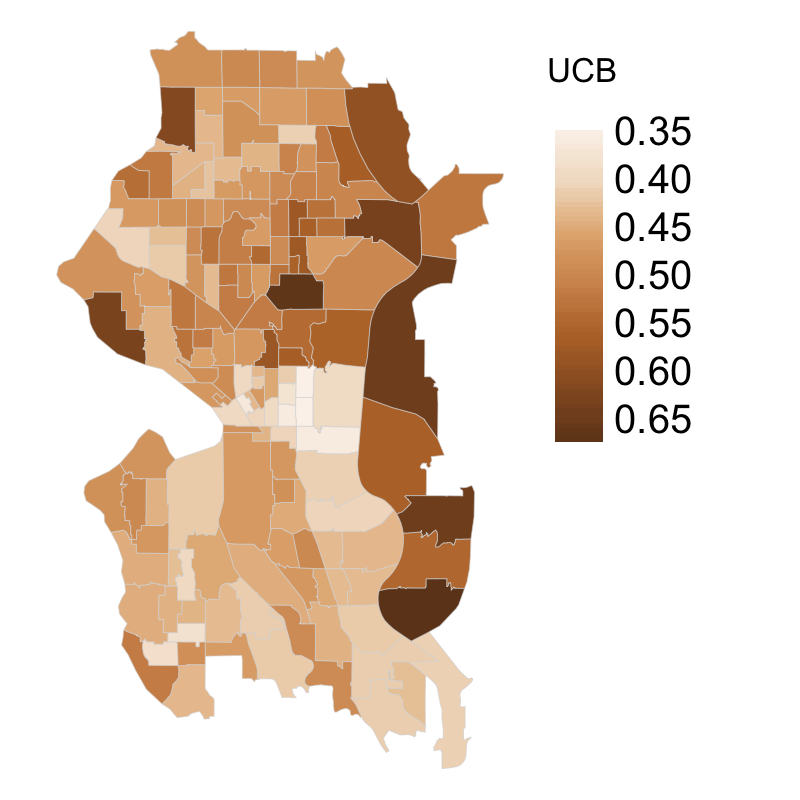}
        \caption{Step-Down}
    \end{subfigure}
    \caption{Upper confidence band for the ranks of Seattle census tracts.}
    \label{fig:urban_rank_UCBs}
    \end{figure}

    % Auto-generated by oa_post_prob_seattle_interpolate.jl
% Requires: \usepackage{booktabs}

\begin{table}[H]
\centering
\caption{Mean LCB of $\Pr(\theta_i > \theta_j)$ by selection threshold~$\tau$: $\hat\theta$ ranking (Seattle)}
\label{tab:tau_sweep_thetahat}
\begin{tabular}{rrrrrr}
\toprule
$\tau$ (\%) & $n_\tau$ & Pairs & Mean LCB & Median LCB & Frac.\ $>0.5$ \\
\midrule
5 & 45 & 2025 & 0.645 & 0.728 & 0.723 \\
10 & 90 & 8100 & 0.710 & 0.794 & 0.784 \\
15 & 136 & 18496 & 0.694 & 0.782 & 0.755 \\
20 & 181 & 32761 & 0.696 & 0.801 & 0.752 \\
25 & 226 & 51076 & 0.682 & 0.782 & 0.736 \\
30 & 272 & 73984 & 0.662 & 0.750 & 0.720 \\
33 & 299 & 89401 & 0.655 & 0.732 & 0.719 \\
\bottomrule
\end{tabular}
\end{table}

\begin{table}[H]
\centering
\caption{Mean LCB of $\Pr(\theta_i > \theta_j)$ by selection threshold~$\tau$: EB conditional mean ranking (Seattle)}
\label{tab:tau_sweep_eb}
\begin{tabular}{rrrrrr}
\toprule
$\tau$ (\%) & $n_\tau$ & Pairs & Mean LCB & Median LCB & Frac.\ $>0.5$ \\
\midrule
5 & 45 & 2025 & 0.796 & 0.900 & 0.878 \\
10 & 90 & 8100 & 0.770 & 0.873 & 0.846 \\
15 & 136 & 18496 & 0.766 & 0.890 & 0.840 \\
20 & 181 & 32761 & 0.730 & 0.855 & 0.793 \\
25 & 226 & 51076 & 0.701 & 0.811 & 0.763 \\
30 & 272 & 73984 & 0.679 & 0.771 & 0.747 \\
33 & 299 & 89401 & 0.660 & 0.737 & 0.729 \\
\bottomrule
\end{tabular}
\end{table}

\section{Supplementary Application to Optimal Auditing of Discriminatory Firms}
\label{app:disc_application}

    In this section, I consider the problem of optimal auditing in the setting of \citealt{Klineetal2022}, in the spirit of examples \ref{ex:util_max_1} and \ref{ex:high_dim_test_1}. In particular, I consider the decision problem of a regulator seeking to identify a set of discriminatory firms. This decision problem is common in practice; the Office of Federal Contract Compliance (OFCCP) annually investigates federal contractors for differential treatment in employment across protected characteristics (\citealt{DOL2013}), and the U.S. Equal Employment Opportunity Commission (EEOC) audits U.S. firms for systemic discrimination (\citealt{EEOC1996}). 

    A recent, large-scale correspondence experiment due to \citealt{Klineetal2022} derives firm-level estimates of the callback differential $\theta_j$ between applicants with distinctly white and black names at 108 U.S. firms, which I will index by $j$. \citealt{Klineetal2022} uncover substantial discrimination against black job applicants, concluding:
    \begin{quote}
        Applying large-scale multiple testing techniques introduced by
Storey (2002, 2003), we find that 23 of the firms in our study discriminate against Black
applicants with at least 95\% posterior certainty.
    \end{quote}
    Motivated by these conclusions, I consider a set of complementary exercises. 
    
    In particular, I consider the auditing problem of an ``intensive margin" policymaker, as in \citealt{Chenetal2025}. That is, I consider a decision maker making a binary classification $\delta_j\in\{0,1\}$ given estimates of firm-level callback differentials $\hat\theta_j$ measured with noise $\sigma_j$ for firms $j=1,\ldots,108$. The policymaker faces the intensive-margin welfare $W_n^{(util)}$. In contrast to the ``extensive margin" decision maker facing a loss which trades off false discoveries against false non-discoveries, the utility-maximizing decision maker's preferences over auditing firm $j$ depend on the magnitude of discrimination at firm $j$. In line with the analysis of \citealt{Chenetal2025}, I choose $K_j = K$ for all firms $j$, meaning that audit costs are identical across firms. I consider a two-sided analogue to the intensive margin loss of \citealt{Chenetal2025} to allow for discrimination against both black and white applicants. As noted in example \ref{ex:util_max_1}, the regulator's Bayes-optimal auditing decision takes the form of a thresholding rule $\delta_j = \mathbb{1}(|\bbE_G(\theta_j|\hat\theta_j,\sigma_j)|\geq K)$.

    I consider the problem of obtaining confidence sets for the Bayes-optimal selection of firms for auditing, and the conditional mean callback differential net auditing costs for firms selected for auditing by an EB decision-maker.\footnote{Given the monotonicity of conditional means and the structure of the Bayes-optimal thresholding rule, the coverage requirement \eqref{eq:uniform_coverage} can be weakened while still obtaining coverage of the Bayes-optimal auditing decision. In particular, it suffices to cover $\mu_G(\cdot,\cdot)$. However, the problem of obtaining such a confidence set is still a multiple testing problem, so the tools developed in this method extend easily.} 
    
    Table \ref{tab:krw_eb_decision_vs_K} speaks to this first exercise. The methods developed in this paper partition the set of firms threefold: a set of confident inclusions and exclusions from the Bayes-optimal audit decision, and a third set of firms for which neither inclusion nor exclusion can be accepted. The methods developed in this paper yield a 183\% larger set of included firms (17 firms) than $F$-localizations (6 firms). Table \ref{tab:krw_eb_decision_vs_K} summarizes these findings. The EB thresholding decision selects 35 firms to audit, 17 of which can be confidently included in the Bayes-optimal selection.

% Generated by krw_hfloc_interpolate_confidence_sets.jl — paste after \usepackage{booktabs,makecell}
% Two-sided EB decision counts vs threshold K on common valid subset (108 firms).
% D+: E[θ|θ̂] > K  (incl.: LCB≥K, excl.: UCB≤K).  D-: E[θ|θ̂] < -K  (incl.: UCB≤-K, excl.: LCB≥-K).
\begin{table}[htbp]
    \centering
    \caption{Confident inclusions / exclusions and EB inclusions vs threshold $K$ (108 firms).}
    \label{tab:krw_eb_decision_vs_K}
    \begin{tabular}{lccc|ccc}
    \toprule
    & \multicolumn{3}{c}{$E[\theta\mid\hat\theta] > K$} & \multicolumn{3}{c}{$E[\theta\mid\hat\theta] < -K$} \\
    \cmidrule(lr){2-4}\cmidrule(lr){5-7}
    $K$ & \makecell{conf.  incl.} & \makecell{conf. excl.} & \makecell{EB incl.} & \makecell{conf. incl.} & \makecell{conf. excl.} & \makecell{EB incl.} \\
    \midrule
    \multicolumn{7}{l}{\textit{DKW $F$-localization}} \\
    0.0125 & 25 & 14 & 63 & 0 & 84 & 2 \\
    0.025 & 6 & 41 & 34 & 0 & 100 & 1 \\
    0.0375 & 3 & 64 & 18 & 0 & 107 & 0 \\
    0.05 & 1 & 80 & 3 & 0 & 107 & 0 \\
    0.0625 & 0 & 94 & 1 & 0 & 108 & 0 \\
    0.075 & 0 & 102 & 1 & 0 & 108 & 0 \\
    \midrule
    \multicolumn{7}{l}{\textit{step-down}} \\
    0.0125 & 38 & 15 & 63 & 0 & 95 & 2 \\
    0.025 & 17 & 42 & 34 & 0 & 106 & 1 \\
    0.0375 & 3 & 72 & 18 & 0 & 107 & 0 \\
    0.05 & 1 & 85 & 3 & 0 & 107 & 0 \\
    0.0625 & 0 & 97 & 1 & 0 & 108 & 0 \\
    0.075 & 0 & 105 & 1 & 0 & 108 & 0 \\
    \bottomrule
    \end{tabular}
\end{table}
At $K=0.0125$ and $K=0.025$, the step-down procedure increases confident inclusions from 25 to 38 and from six to 17, respectively. Figure \ref{fig:krw_eb_decision_vs_K_vis} visualizes these findings.

\begin{figure}
    \centering
    \includegraphics[width=0.75\linewidth]{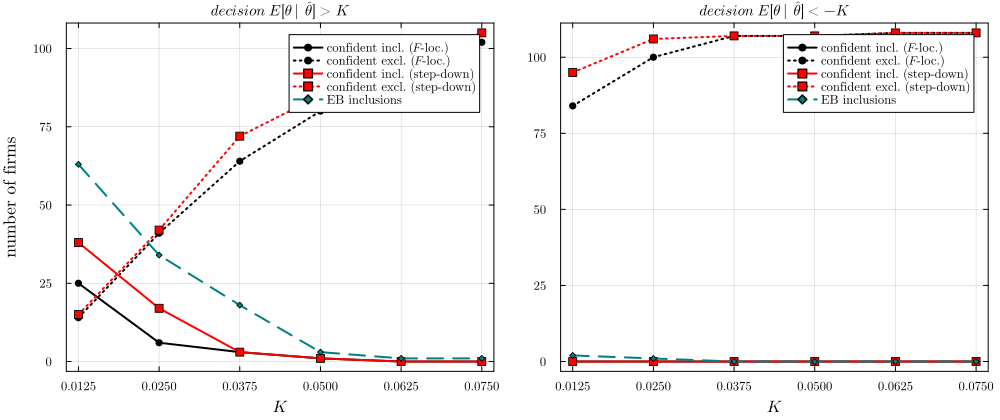}
    \caption{Number of confident inclusions and exclusions as identified by $F$-localizations and step-down methods, relative to the number of selections by an EB decision maker.}
    \label{fig:krw_eb_decision_vs_K_vis}
\end{figure}

To understand the welfare implications of the feasible EB auditing decision, I obtain confidence sets for the conditional mean callback differential between black and white applicants at all 108 firms. For many firms, this conditional mean callback differential is significantly larger than some of the choices of $K$ above, suggesting gains from auditing. Figure \ref{fig:caterpillar_eb_selected} plots these confidence sets using both $F$-localizations and the methods developed in this paper.
\begin{figure}[H]
    \centering
    \includegraphics[width=0.75\linewidth]{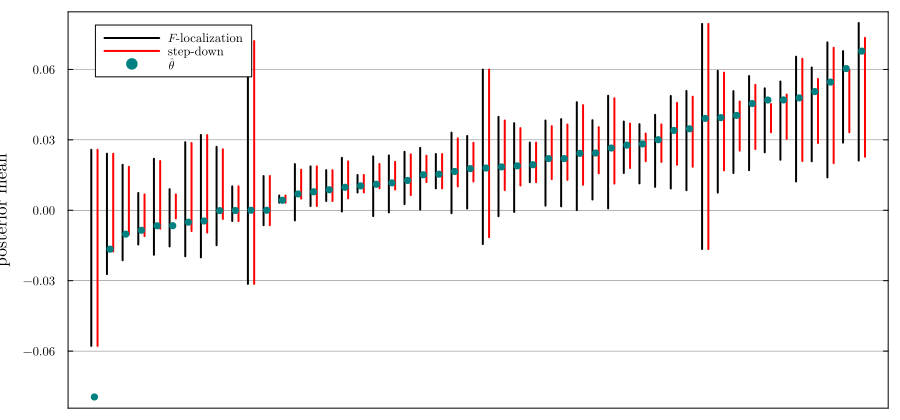}
    \caption{Confidence sets for conditional mean callback differentials at a random subset of 50 firms.}
    \label{fig:caterpillar_eb_selected}
\end{figure}

\end{appendix}

\end{document}